\documentclass{LMCS}
\pdfoutput=1

\def\doi{8(4:8)2012}
\lmcsheading%
{\doi}
{1--33}
{}
{}
{Oct.~14, 2011}
{Oct.~17, 2012}
{}
 
\usepackage{tikz}
\usetikzlibrary{arrows,automata,shapes,patterns}

\usepackage{latexsym,amssymb,amsmath,ae,aeguill,amscd,stmaryrd,multirow}

\usepackage{enumerate}
\usepackage{hyperref}

\theoremstyle{plain}
\def\eg{{\em e.g.}}

\usepackage{defs}

\begin{document}

\title[Off-line test selection with test purposes 
for non-det. timed automata]{Off-line Test Selection with Test Purposes 
for Non-Deterministic Timed Automata\rsuper*}
\thanks{This work was partly funded by the French project TESTEC (ANR-07-TLOG-022).}

\author[N.~Bertrand]{Nathalie Bertrand\rsuper a}	
\address{{\lsuper{a,b}}Inria Rennes - Bretagne Atlantique, Rennes, France}
\email{\{nathalie.bertrand, thierry.jeron\}@inria.fr}  

\author[T.~J\'eron]{Thierry J\'eron\rsuper b}	
\address{\vskip-6 pt}	

\author[A.~Stainer]{Am\'elie Stainer\rsuper c}	
\address{{\lsuper c}University of Rennes 1,  Rennes, France}	
\email{amelie.stainer@inria.fr}  

\author[M.~Krichen]{Moez Krichen\rsuper d}	
\address{{\lsuper d}University of Sfax, Tunisia}	
\email{moez.krichen@redcad.org}  



\titlecomment{{\lsuper*}This article is based on the material of the TACAS conference paper~\cite{BertrandJeronStainerKrichen-TACASS2011}}


\begin{abstract} 
  \noindent This article proposes novel off-line test generation
  techniques from non-deter\-ministic timed automata with inputs and
  outputs (TAIOs) in the formal framework of the \tioco\, conformance
  theory.  In this context, a first problem is the determinization of
  TAIOs, which is necessary to foresee next enabled actions after an
  observable trace, but is in general impossible because not all timed
  automata are determinizable. This problem is solved thanks to an
  approximate determinization using a game approach.
%
   The algorithm performs an io-abstraction which preserves the
    \tioco\, conformance relation and thus guarantees the soundness of
    generated test cases.
    A second problem is the selection of test cases from a TAIO
    specification.  The selection here relies on a precise description
    of timed behaviors to be tested which is carried out by expressive
    test purposes modeled by a generalization of TAIOs.  Finally, an
    algorithm is described which generates test cases in the form of
    TAIOs equipped with verdicts, using a symbolic co-reachability
    analysis guided by the test purpose.  Properties of test cases are
    then analyzed with respect to the precision of the approximate
    determinization: when determinization is exact,
%
    which is the case on known determinizable classes, in addition to
    soundness, properties characterizing the adequacy of test cases
    verdicts are also guaranteed.
\end{abstract}

\maketitle

\section*{Introduction}
\label{sec-intro}


Conformance testing is the process of testing whether some
implementation of a software system 
behaves correctly with respect to its specification.
In this testing framework, 
implementations are considered as {\em black boxes}, \ie the source
code is unknown, only their interface with the environment is known
and used to interact with the tester.  In {\em formal model-based conformance testing},
models are used to  describe  testing artifacts
(specifications, implementations, test cases, ...). 
Moreover, conformance is formally defined as a relation 
between implementations and specifications 
which reflects what are the correct behaviors of the implementation
with respect to those of the specification.
Defining such a relation requires the hypothesis that the implementation
behaves as a model.
Test cases with verdicts, which will be executed against the implementation
in order to check conformance, are generated automatically from the 
specification. 
Test generation algorithms should then ensure
properties relating verdicts of executions of test cases 
with the conformance relation (\eg~soundness),
thus improving the quality of testing compared to manual writing of test cases.

For timed systems, model-based conformance testing has already been
explored in the last decade, with different models and conformance
relations (see \eg~\cite{DBLP:conf/formats/SchmaltzT08} for a
survey), and various test generation
algorithms (e.g.~\cite{BrionesBrinksma-FATES05,KrichenTripakis09,NielsenSkou-STTT03}).
In this context, a very popular model is {\em timed automata with
  inputs and outputs} (TAIOs), a variant of {\em timed automata}
(TAs)~\cite{AlurDill94}, in which the alphabet of observable actions is
partitioned into inputs and outputs.  
We consider here a very general model,
partially observable and non-deterministic TAIOs with
invariants for the modeling of urgency.
We resort to the \tioco\, conformance relation defined for TAIOs~\cite{Krichen-Tripakis-2004}, which is equivalent to the \rtioco\, relation~\cite{Larsen-Mikucionis-Nielsen-2004}.
This relation compares the observable behaviors of timed systems,
made of inputs, outputs and delays, restricting attention to what happens 
after specification traces.
Intuitively, an implementation conforms to a specification if after 
any observable trace of the specification,
outputs and delays observed on the implementation after this trace should 
be allowed  by the specification.

One of the main difficulties encountered in test generation for those
partially observable, non-deterministic TAIOs is determinization.
In fact determinization 
is required in order to foresee the next enabled actions during
execution, and thus to emit a correct verdict depending on whether 
actions observed on the implementation 
are allowed by the specification model after the current observable behavior.
Unfortunately, TAs (and thus TAIOs) are not determinizable in general~\cite{AlurDill94}: the class of deterministic TAs is a strict subclass of TAs. 
Two different approaches have been taken for test generation from timed models,
which induce different treatments of non-determinism. 
\begin{iteMize}{$\bullet$}
\item 
In {\em off-line test generation} test cases are first generated as
timed automata (or timed sequences, or timed transition systems) 
and subsequently
executed on the implementation.  One advantage is that test cases can
be stored and further used \eg~for regression testing and serve for
documentation.  However, due to the non-determinizability of TAIOs,
the approach has often been limited to deterministic or determinizable TAIOs 
(see \eg~\cite{KhoumsiJeronMarchand-FATES03,NielsenSkou-STTT03}).
A notable exception is~\cite{KrichenTripakis09} where the problem is solved by the
use of an over-approximate determinization with fixed resources 
(number of clocks and maximal constant): 
a deterministic automaton with those resources is built, 
which simulates the behaviors of the non-deterministic one.
Another one is 
~\cite{DavidLarsenLiNielsen-ICST09} where winning strategies of
 timed games are used as test cases.
\item In {\em on-line test generation}, test cases are generated
  during their execution.  After the current observed trace, enabled
  actions after this trace are computed from the specification model
  and, either an allowed input is sent to the implementation, or a
  received output or an observed delay is checked.  This technique can
  be applied to any TAIO, as possible observable actions are computed
  only along the current finite execution (the set of possible states
  of the specification model after a finite trace, and their enabled
  actions are finitely representable and computable), thus
  avoiding a complete determinization.  On-line test generation is of
  particular interest to rapidly discover errors, can be applied
    to large and non-deterministic systems, but may sometimes be
  impracticable due to a lack of reactivity (the time needed to
  compute successor states on-line may sometimes be incompatible with
  real-time constraints).
\end{iteMize}
Our feeling is that off-line test generation from timed models 
did not receive much attention because of the inherent difficulty 
of determinization. 
However, recent works on approximate determinization of timed automata~\cite{KrichenTripakis09,BertrandStainerJeronKrichen-FOSSACS2011} 
open the way to new research approaches and results in this domain.

\subsection*{Contribution} 
In this paper, we propose to generate test cases off-line for the whole class of
non-deterministic TAIOs, in the formal context of the \tioco\, conformance
theory.  The determinization problem is tackled thanks to an
approximate determinization with fixed resources in the spirit 
of~\cite{KrichenTripakis09}, using a game
approach allowing to more closely simulate the non-deterministic TAIO~\cite{BertrandStainerJeronKrichen-FOSSACS2011}.
Our approximate determinization method is more precise
than~\cite{KrichenTripakis09} (see~\cite{BertrandStainerJeronKrichen-FOSSACS2011,BertrandStainerJeronKrichen-RR2010} for details), preserves the richness of our model by dealing with partial
observability and urgency, and can be adapted to testing by  
a different treatment of inputs, outputs and delays.  
Determinization is exact for known classes of determinizable TAIOs
(\eg~event-clock TAs, TAs with integer resets, strongly non-Zeno TAs) if
resources are sufficient. 
In the general case, determinization may over-approximate outputs and delays
and under-approximate inputs.
More precisely, it produces a deterministic {\em io-abstraction} of the TAIO 
for a particular {\em io-refinement} relation
which generalizes the one of~\cite{David-Larsen-etal-HSCC10}.  
As a consequence, if test cases are
generated from the io-abstract deterministic TAIO and are sound
for this TAIO, they are guaranteed to be sound for
the original (io-refined) non-deterministic TAIO.

Behaviors of specifications to be tested are identified by means of test purposes. Test purposes are often used in testing practice, 
and are particularly useful when one wants to focus testing
on particular behaviors, \eg~corresponding to requirements 
or suspected behaviors of the implementation.
In this paper they are defined as {\em open timed automata with inputs and
  outputs} (OTAIOs), a model generalizing TAIOs, allowing to precisely
target some behaviors according to actions and clocks of the
specification as well as proper clocks.  
Then, 
in the same spirit as for the TGV tool in the untimed case~\cite{jard04a},  
test selection is performed by a construction relying 
on a co-reachability analysis. 
Produced test cases are in the form of TAIOs, 
while most approaches generate less elaborated 
test cases in the form of timed traces or trees.
In addition to soundness, 
when determinization is exact, 
we also prove an exhaustiveness property,
and two other properties on the adequacy of test case verdicts.
To our knowledge, this whole work constitutes the most general and
advanced off-line test selection approach for TAIOs.

This article is a long version
of~\cite{BertrandJeronStainerKrichen-TACASS2011}. In addition to the
proofs of key properties, it also contains much more details,
explanations, illustrations by examples, complexity considerations,
and a new result on exhaustiveness of the test generation method.

\subsection*{Outline}
The paper is structured as follows. In the next section we introduce
the model of OTAIOs, its semantics, some notations and operations on this model
and the model of TAIOs.
Section 2 recalls the \tioco\, conformance theory for TAIOs, including 
properties of test cases 
relating conformance and verdicts, and introduces an io-refinement
relation which preserves \tioco. Section 3 presents our game approach for
the approximate determinization compatible with the io-refinement.  In
Section 4 we detail the test selection mechanism using test purposes
and prove some properties on generated test cases.
Section 5 discusses some issues related to test case execution 
and test purposes and some related work.

\section{A model of  open timed automata with inputs/outputs}
Timed automata (TAs)~\cite{AlurDill94} is 
a usual model for time constrained systems.
In the context of model-based testing, TAs have been extended 
to timed automata with inputs and outputs (TAIOs) 
whose sets of actions are partitioned into 
inputs, outputs and unobservable actions.
In this section, we further extend TAIOs 
by partitioning the set of clocks into proper clocks (\emph{i.e.},
controlled by the automaton) and observed clocks (\emph{i.e}, owned by
some other automaton).  
The resulting model of {\em open timed
    automata with inputs/outputs} (OTAIOs for short), allows one to
  describe observer timed automata that can test clock values from
  other automata.
 While the sub-model of TAIOs (with only proper
clocks) is sufficient for most testing artifacts (specifications,
implementations, test cases) observed clocks of OTAIOs will be useful
to express test purposes whose aim is to focus on the timed behaviors
of the specification.  Like in the seminal paper for
TAs~\cite{AlurDill94}, we consider OTAIOs and TAIOs with location
invariants to model urgency.

\subsection{Timed automata with inputs/outputs}

We start by introducing notations and useful definitions concerning
TAIOs and OTAIOs.

Given $X$ a finite set of {\it clocks}, 
a {\it clock valuation} is a
mapping $v:X \rightarrow \setRnn$, where $\setRnn$ 
is the set of non-negative real numbers.  
$\bar{0}$ stands for the valuation assigning $0$ to all clocks.
If $v$ is a valuation over
$X$ and 
$t\in \setRnn$, 
 then $v+t$ denotes the valuation which assigns
to every clock $x\in X$ the value $v(x)+t$.
For $X' \subseteq X$ we write $v_{[X'\leftarrow0]}$ for the valuation
equal to $v$ on $X \setminus X'$ and 
assigning $0$ to all clocks of $X'$.
Given $M$ a non-negative integer, an {\it $M$-bounded guard} (or
simply guard) over $X$ is a finite conjunction of constraints of the
form $x \sim c$ where $x\in X,\; c\in [0,M]\cap\mathbb{N}$ and $\sim
\in \{<,\le,=,\ge,>\}$. 
Given $g$ a guard and $v$ a valuation, we write $v
\models g$ if $v$ satisfies $g$.  We sometimes abuse notations and write $g$ for
the set of valuations satisfying $g$.  Invariants are restricted
cases of guards: given $M \in \setN$, an {\em $M$-bounded invariant} over
$X$ is a finite conjunction of constraints of the form $x \lhd c$
where $x \in X, c \in [0,M] \cap \setN$ and $\lhd \in \{<,\leq\}$.  
We denote by $G_M(X)$ (resp.  $I_M(X)$)
the set of $M$-bounded guards (resp. invariants) over $X$.


In the sequel, we write $\sqcup$ for the disjoint union of sets,
  and use it, when appropriate, to insist on the fact that sets are
  disjoint.

\begin{defi}[OTAIO]
\label{def-ta}
An {\em open timed automaton with inputs and outputs} (OTAIO) is a tuple
$\A=(L^\pA,\ell_0^\pA,\Sigma_?^\pA,\Sigma_!^\pA,\Sigma_\tau^\pA,X_p^\pA,X_o^\pA,M^\pA,\Inv^\pA,E^\pA)$ 
such that: 
\begin{iteMize}{$\bullet$}
\item $L^\pA$ is a finite set of {\em locations}, 
with $\ell_0^\pA\in L^\pA$ the {\em initial location}, 
\item $\Sigma_?^\pA$, $\Sigma_!^\pA$ and $\Sigma_\tau^\pA$ 
are disjoint finite  {\em alphabets} of
{\em input actions} (noted $a?, b?,\ldots$), 
{\em output actions} (noted $a!, b!, \ldots$),
and {\em internal actions} (noted $\tau_1, \tau_2, \ldots$). 
We note $\Sigma_{obs}^\pA=\Sigma_?^\pA \sqcup \Sigma_!^\pA$ 
for the alphabet of observable actions,
and $\Sigma^\pA=\Sigma_?^\pA \sqcup \Sigma_!^\pA \sqcup \Sigma_\tau^\pA$ 
for the whole set of actions.
\item $X_p^\pA$ and $X_o^\pA$ are disjoint finite sets of {\em proper
    clocks} and {\em observed clocks}, respectively.  We note $X^\pA
  =X_p^\pA \sqcup X_o^\pA$ for the whole set of {\em clocks}.
\item $M^\pA\in \mathbb{N}$ is the {\em maximal constant} of $\A$, 
and we will refer to $(|X^\pA|,M^\pA)$ as the {\em resources} of $\A$,
\item $\Inv^\pA: L^\pA \ra I_{M^\pA}(X^\pA)$ is a mapping  which labels each location with an {\em $M$-bounded invariant}, 
\item $E^\pA\subseteq L^\pA\times G_{M^\pA}(X^\pA) \times \Sigma^\pA \times 2^{X_p^\pA} \times L^\pA$ is a finite set of {\em edges} where {\em guards} 
are defined on $X^\pA$,
but {\em resets} are restricted to proper clocks in $X_p^\pA$.
\end{iteMize}
\end{defi}\medskip

\noindent One of the reasons for introducing the OTAIO model is to have a
uniform model (syntax and semantics) that will be next specialized for
particular testing artifacts.  In particular, an OTAIO with an empty
set of observed clocks $X_o^\pA$ is a classical TAIO, and will be the
model for specifications, implementations and test cases.  The
partition of actions reflects their roles in the testing context: the
tester cannot observe internal actions, but controls inputs and
observes outputs (and delays).  The set of clocks is also partitioned
into {\em proper clocks}, \ie usual clocks controlled by the system
itself through resets, as opposed to {\em observed clocks} referring
to proper clocks of another OTAIO (\eg~ modeling the system's
environment).  These cannot be reset to avoid intrusiveness, but
synchronization with them in guards and invariants is allowed.
This partition of clocks will be useful for test purposes which
  can have, as observed clocks, some proper clocks of specifications,
  with the aim of selecting time constrained behaviors of
  specifications to be tested.

\begin{figure}[htbp]
\begin{center}
\scalebox{0.65}{
\begin{tikzpicture}[->,>=stealth',shorten >=1pt,auto,node distance=2cm,
                    semithick]

  \tikzstyle{every state}=[text=black]

  \node[state, fill=white] (A) {$\ell_0$};
  \node[state, fill=white] (B) [right of=A, node distance=3cm, yshift=1cm] {$\ell_1$};
  \node[state, fill=white] (C) [right of=B, node distance=4cm, yshift=0cm] {$\ell_2$};
  \node[state, fill=white] (D) [right of=C, node distance=3cm, yshift=0cm] {$\ell_3$};
  \node[state, fill=white] (E) [right of=D, node distance=3cm, yshift=0cm] {$\ell_4$};
  \node[state, fill=white] (F) [right of=A, node distance=3cm, yshift=-1cm] {$\ell_5$};
  \node[state, fill=white] (G) [right of=F, node distance=4cm, yshift=0cm] {$\ell_6$};
  \node[state, fill=white] (H) [right of=G, node distance=3cm, yshift=0cm] {$\ell_7$};
  \node[state, fill=white] (I) [right of=H, node distance=3cm, yshift=0cm] {$\ell_8$};
  \node[state, fill=white, color=white] (A') [left of=A, node distance=2cm, yshift=0cm] {};
  \node[fill=white] (Ae) [above of=A, node distance=.8cm, yshift=0cm] {$x \le 1$};
  \node[fill=white] (Ce) [below of=C, node distance=.8cm, yshift=0cm] {$x \le 1$};
  \node[fill=white] (De) [below of=D, node distance=.8cm, yshift=0cm] {$x \le 1$};
  \node[fill=white] (Ge) [below of=G, node distance=.8cm, yshift=0cm] {$x = 0$};
  \node[fill=white] (He) [below of=H, node distance=.8cm, yshift=0cm] {$x = 0$};
  
  \path (A') edge node [above] {} (A)
        (A) edge node [above,sloped] {$x=1, \tau$} (B)    
        (B) edge node [above] {$1<x<2, a?, \{x\}$} (C)   
        (C) edge node [above] {$x=0, b!$} (D)   
        (D) edge node [above] {$b!$} (E)   
        (C) edge [loop above] node [above] {$x=1, \tau, \{x\}$} (C)    
        (A) edge node [above,sloped] {$x=1, \tau, \{x\}$} (F)    
        (F) edge node [above] {$x<1, a?, \{x\}$} (G)   
        (G) edge node [above] {$b!$} (H)   
        (H) edge node [above] {$b!$} (I)   
;
\end{tikzpicture}
}
\caption{Specification $\A$}\label{ExSpec}
\end{center}
\end{figure}

\begin{exa}
Figure~\ref{ExSpec} represents a TAIO for a specification $\A$
that will serve as a running example in this paper.
Its clocks are $X=X_p^\pA=\{x\}$, its maximal constant is $M^\pA=2$, 
it has a single input $\Sigma_?^\pA=\{a\}$, 
a single output $\Sigma_!^\pA=\{b\}$ and 
one internal action $\Sigma_\tau^\pA=\{\tau\}$.
Informally, its behavior is as follows.
It may stay in the initial location $\ell_0$ while $x\leq 1$,
and at $x=1$, has the choice, either to go to $\ell_1$ with action $\tau$,
or go to $\ell_5$ with action $\tau$ while resetting $x$.
In $\ell_1$, it may receive $a$ and move to $\ell_2$ 
when $x$ is between $1$ and $2$, and reset $x$.
In $\ell_2$ it may stay while $x\leq 1$ and, 
either send $b$ and go to $\ell_3$ at $x=0$, 
or loop silently when $x=1$  while resetting $x$.
This means that $b$ can be sent at any integer delay after entering
$\ell_2$.
In $\ell_3$ it may stay while $x\leq 1$ and move to $\ell_4$ when sending $b$.
In $\ell_5$, one can move to $\ell_6$ before $x=1$ by receiving $a$ 
and resetting  $x$. Due to invariants $x=0$ in $\ell_6$ and 
$\ell_7$, the subsequent behavior consists in the immediate 
transmission of two $b$'s.
\end{exa}

\subsection{The semantics of OTAIOs}
Let
$\A=(L^\pA,\ell_0^\pA,\Sigma_?^\pA,\Sigma_!^\pA,\Sigma_\tau^\pA,X_p^\pA,X_o^\pA,M^\pA,\Inv^\pA,E^\pA)$ be an OTAIO.
The semantics of $\A$ is a {\it timed transition system}
$\mathcal{T}^{\pA}=(S^\pA,s_0^\pA,\Gamma^\pA,\rightarrow_\pA)$ 
where
\begin{iteMize}{$\bullet$}
\item 
$S^\pA=L^\pA \times \setRnn^{X^\pA}$ is the set of {\em states} \ie
pairs $(\ell,v)$ consisting in a location and a valuation of clocks;
\item $s_0^\pA=(\ell_0^\pA,\overline{0}) \in S^\pA$ is the {\em initial state}; 
\item 
$\Gamma^\pA=\setRnn \sqcup E^\pA \times 2^{X_o^\pA}$ is the set of transition {\em labels}
consisting in either a delay $\delta$ or a pair  $(e,X'_o)$ formed by 
an edge $e \in E$ and a set $X_o' \subseteq X_o^\pA$ of observed clocks;
\item 
the  transition relation 
$\rightarrow_\pA \subseteq S^\pA \times \Gamma^\pA \times S^\pA$ is the smallest set of the following moves:
\begin{iteMize}{$-$}
\item
{\em Discrete moves:} 
$(\ell,v) \relt{\pA}{(e,X'_o)} (\ell',v')$ 
whenever there exists $e=(\ell,g,a,X'_p,\ell') \in E^\pA$ 
such that $v \models g \wedge \Inv^\pA(\ell)$, 
$X_o' \subseteq X_o^\pA$ is an arbitrary subset of observed clocks, 
$v'=v_{[X'_p\sqcup X'_o\leftarrow 0]}$ and $v' \models \Inv^\pA(\ell')$.
Note that $X'_o$ is unconstrained as observed clocks are not controlled by $\A$ but by a peer OTAIO. 
\item
{\em Time elapse:}
$(\ell,v) \relt{\pA}{\delta} (\ell,v+\delta)$
for $\delta \in \setRnn$ if  
$v + \delta\models \Inv^\pA(\ell)$.
\end{iteMize}
\end{iteMize}\medskip
The semantics of OTAIOs generalizes the usual semantics of TAIOs.
The difference lies in the treatment of the additional observed clocks
as the evolution of those clocks is controlled by a peer OTAIO.
The observed clocks evolve at the same speed as the proper clocks, 
thus continuous moves are simply extended to proper and observed clocks.
For discrete moves however, resets of observed clocks are uncontrolled,
thus all possible resets have to be considered.

A  {\it partial run} of $\A$ is a finite sequence of subsequent moves
in $(S^\pA \times \Gamma^\pA)^*.S^\pA$.
For example 
$\rho=s_0 \relt{\pA}{\delta_1} s'_0 \relt{\pA}{(e_1,X_o^1)} s_1\cdots
s_{k-1}\relt{\pA}{\delta_k} s'_{k-1} \relt{\pA}{(e_k,X_o^k)} s_k$. 
The sum of 
delays in $\rho$ is noted  $time(\rho)$. 
A {\em run} is a partial run starting in $s_0^\pA$.
A state $s$ is {\em reachable} if there exists a run leading to $s$. 
A  state $s$ is {\em co-reachable} from a set $S'\subseteq S^\pA$ if there is 
a partial run from $s$  to a state in $S'$.
We note $\reach(\A)$ the set of reachable states and 
$\coreach(\A,S')$ the set of states co-reachable  from $S'$.

A (partial) {\em sequence} is a projection of a (partial) run where
states are forgotten, and discrete transitions are abstracted to
actions and proper resets which are grouped with observed resets.  
As an example, the
sequence corresponding to a run 
\[\rho=s_0 \relt{\pA}{\delta_1} s'_0
\relt{\pA}{(e_1,X_o^1)} s_1\cdots s_{k-1}\relt{\pA}{\delta_k} s'_{k-1}
\relt{\pA}{(e_k,X_o^k)} s_k\]
 is \[\mu = \delta_1.(a_1,X_p^1 \sqcup
X_o^1)\cdots \delta_k.(a_k,X_p^k \sqcup X_o^k)\] where 
$e_i=(\ell_i,g_i,a_i,X^i_p,\ell'_i)$ for all $i \in
[1,k]$.  We then note $s_0
\relt{\pA}{\mu} s_k$.  We write $s_0
 \relt{\pA}{\mu}$ 
if there exists
$s_k$ such that  $s_0
 \relt{\pA}{\mu} s_k$.  We note $\seq(\A)
\subseteq (\setRnn \sqcup (\Sigma^\pA \times 2^{X^\pA}))^*$
(respectively $\pseq(\A)$) the set of sequences (resp. partial
sequences) of $\A$. For a sequence $\mu$, $time(\mu)$
denotes the sum of delays in $\mu$.

For a (partial) sequence $\mu \in \pseq(\A)$, 
$Trace(\mu) \in (\setRnn \sqcup
\Sigma_{obs}^\pA)^*.\setRnn$ denotes the observable behavior obtained by erasing internal actions
and summing delays between observable ones. It is defined inductively as follows:
\begin{iteMize}{$\bullet$}
\item 
$Trace(\varepsilon) = 0$,
\item
$Trace(\delta_1\ldots \delta_k)=\Sigma_{i=1}^k \delta_i$,
\item 
$Trace(\delta_1\ldots \delta_k.(\tau,X').\mu)=
Trace ((\Sigma_{i=1}^k \delta_i).\mu)$,
\item 
$Trace(\delta_1\ldots \delta_k.(a,X').\mu)=
(\Sigma_{i=1}^k \delta_i).a.Trace(\mu)$ if $a\in \Sigma_{obs}^\pA$.
\end{iteMize}
For example $Trace(1.(\tau,X^1).2.(a,X^2). 2.(\tau,X^3))=3.a.2$ and
$Trace(1.(\tau,X^1).2.(a,X^2))=3.a.0$.  When a trace ends by a
$0$-delay, we sometimes omit it and write \eg~$3.a$ for $3.a.0$.

When concatenating two traces, the last delay of the first trace and
the initial delay of the second one must be added up as follows: if
$\sigma_1=\delta_1.a_1.\cdots a_n.\delta_{n+1}$ and
$\sigma_2=\delta'_1.a'_1.\cdots a'_m.\delta'_{m+1}$ then
$\sigma_1.\sigma_2= \delta_1.a_1.\cdots
a_n.(\delta_{n+1}+\delta'_1).a'_1.\cdots a'_m.\delta'_{m+1}$.  
Concatenation allows one to define the notion of prefix. Given a
  trace $\sigma$, $\sigma_1$ is a \emph{prefix} of $\sigma$ if there
  exists some $\sigma_2$ with $\sigma = \sigma_1.  \sigma_2$. Under
  this definition, $1.a.1$ is a prefix of $1.a.2.b$.

For a run $\rho$ projecting onto a sequence $\mu$, 
we also write $Trace(\rho)$ for $Trace(\mu)$.  
The set of traces of runs of $\A$ is denoted by
$\traces(\A)\subseteq (\setRnn \sqcup \Sigma_{obs}^\pA)^*.\setRnn$
\footnote{Notice that formally, a trace always ends with a delay, which can be $0$.
This technical detail is useful later to define verdicts as soon as possible 
without waiting for a hypothetical next action.
}. 


 Two OTAIOs are said
\emph{equivalent} if they have the same sets of traces.

Let $\sigma \in (\setRnn \sqcup \Sigma_{obs}^\pA)^*.\setRnn$ be a trace, 
and $s \in S^\pA$ be a state, 
\begin{iteMize}{$\bullet$}
\item 
$\A \after \sigma =\{s \in S^\pA \mid \exists \mu \in \seq(\A),
s_0^\pA \relt{\pA}{\mu} s \wedge Trace(\mu)=\sigma\}$ denotes the set of
states where $\A$ can stay after observing the trace $\sigma$.
\item
$elapse(s)= \{t \in \setRnn\mid 
\exists \mu \in (\setRnn \sqcup (\Sigma_\tau^\pA\times 2^{X^\pA}))^*, 
s \relt{\pA}{\mu} \wedge \, time(\mu)=t\}$ is the set of
enabled delays in $s$ with no  observable action.
\item  $out(s) = \{a \in \Sigma_!^\pA \mid \exists X
\subseteq X^\pA, s \relt{\pA}{(a,X)}\} \cup elapse(s)$ (and $in(s) =
\{a \in \Sigma_?^\pA \mid s \relt{\pA}{(a,X)}\}$) for the set of
outputs and delays (respectively inputs) that can be observed from
$s$.  
For $S'\subseteq S^\pA$, $out(S')=\bigcup_{s \in S'}out(s)$
and $in(S')=\bigcup_{s \in S'}in(s)$.
\end{iteMize}
Using these last definitions, 
we will later describe 
the set of 
possible outputs and delays after the trace $\sigma$
by  $out(\A \after \sigma)$.

Notice that all notions introduced for OTAIOs apply to the subclass of TAIOs.

\subsection{Properties and operations}
A TAIO $\A$ is {\em deterministic} (and called a DTAIO) whenever for
any $\sigma \in \traces(\A), \A \after \sigma$ is a
singleton\footnote{Determinism is only defined (and used in the
  sequel) for TAIOs.  For OTAIOs, the right definition would consider
  the projection of $\A \after \sigma$ which forgets values of
  observed clocks, as these introduce ``environmental''
  non-determinism.}.  A TAIO $\A$ is {\it determinizable} if there
exists an equivalent DTAIO.  It is well-known that some timed automata
are not determinizable~\cite{AlurDill94}; moreover, the
determinizability of timed automata is an undecidable problem, even
with fixed resources~\cite{Tripakis-ipl06,Finkel-formats06}.


An OTAIO $\A$ is said {\em complete} if in every location $\ell$,
$\Inv^\pA(\ell)= \true$ and for every action $a\in \Sigma^\pA$, the
disjunction of all guards of transitions leaving $\ell$ and labeled
by $a$ is $\true$.  
This entails that 
$\seq(\A)\downarrow_{X^\pA_p} = (\setRnn \sqcup (\Sigma^\pA \times 2^{X_o^\pA}))^*$,
where $\downarrow_{X^\pA_p}$ is the projection 
that removes resets of proper clocks in $X^\pA_p$.
This means that $\A$ is universal for all the behaviors of its environment.

An OTAIO $\A$ is {\em input-complete} in a state $s \in \reach(\A)$, if
 $in(s) = \Sigma_?^\pA$.  
An OTAIO $\A$  is input-complete if it is input-complete in all its
reachable states.

An OTAIO
$\A$ is {\em non-blocking} if $\forall s \in \reach(\A), \forall t \in
\setRnn, \exists \mu \in \pseq(\A) \cap (\setRnn \sqcup ((\Sigma_!^\pA
\sqcup \Sigma_\tau^\pA) \times 2^{X^\pA}))^*, time(\mu) = t \wedge s
\stackrel{\mu}{\rightarrow_\pA}$.
This means that it never  blocks the evolution of time, waiting for an input.

For modeling the behavior of composed systems, in particular for
modeling the execution of test cases on implementations, we introduce
the classical parallel product. This operation
consists in the synchronization of two TAIOs on complementary
  observable actions (\eg~ $a!$, the emission of $a$ and $a?$ its reception)
and induces the intersection of the sets of
traces. It is only defined for \emph{compatible} TAIOs,
  \emph{i.e.}
  $\A^i=(L^{\petit{i}},\ell_0^{\petit{i}},\Sigma_?^{\petit{i}},\Sigma_!^{\petit{i}},\Sigma_\tau^{\petit{i}},X_p^{\petit{i}},M^{\petit{i}},\Inv^{\petit{i}},E^{\petit{i}})$
  for $i=1,2$ such that $\Sigma_!^{\petit{1}}=\Sigma_?^{\petit{2}}$,
  $\Sigma_?^{\petit{1}}=\Sigma_!^{\petit{2}}$,
  $\Sigma_\tau^{\petit{1}}\cap \Sigma_\tau^{\petit{2}}=\emptyset$ and
  $X_p^{\petit{1}} \cap X_p^{\petit{2}}= \emptyset$. 
  \begin{defi}[Parallel product]
\label{def:product}
 The \emph{parallel product} of two compatible TAIOs
   $\A^i=(L^{\petit{i}},\ell_0^{\petit{i}},\Sigma_?^{\petit{i}},\Sigma_!^{\petit{i}},\Sigma_\tau^{\petit{i}},X_p^{\petit{i}},M^{\petit{i}},\Inv^{\petit{i}},E^{\petit{i}})$
   $i=1,2$ 
 is a TAIO $\A^1 \| \A^2 =
  (L,\ell_0,\Sigma_?,\Sigma_!,\Sigma_\tau,X_p,M,\Inv,E)$ where:
\begin{iteMize}{$\bullet$}
\item $L=L^{\petit{1}} \times L^{\petit{2}}$, $\ell_0 = (\ell_0^{\petit{1}},\ell_0^{\petit{2}})$,
\item $\Sigma_?= \Sigma_?^{\petit{1}}$, 
$\Sigma_!= \Sigma_!^{\petit{1}} $ and 
$\Sigma_\tau= \Sigma_\tau^{\petit{1}} \sqcup \Sigma_\tau^{\petit{2}}$
\item $X_p = X_p^{\petit{1}} \sqcup X_p^{\petit{2}}$
\item $M= \max(M^{\petit{1}}, M^{\petit{2}})$
\item $\forall (\ell^{\petit{1}},\ell^{\petit{2}}) \in L, 
  \Inv((\ell^{\petit{1}},\ell^{\petit{2}})) = \Inv(\ell^{\petit{1}}) \wedge \Inv(\ell^{\petit{2}})$ 
\item $E$ is the smallest relation such that:
\begin{iteMize}{$-$}
\item for $a \in \Sigma_?^{\petit{1}}\sqcup \Sigma_!^{\petit{1}}$,  
if $(\ell^{\petit{1}},g^{\petit{1}},a,X'^{\petit{1}}_p,\ell'^{\petit{1}}) \in E^{\petit{1}}$ and 
$(\ell^{\petit{2}},g^{\petit{2}},a,X'^{\petit{2}}_p,\ell'^{\petit{2}}) \in E^{\petit{2}}$ then
$((\ell^{\petit{1}},\ell^{\petit{2}}),g^{\petit{1}} \wedge g^{\petit{2}},a,
X'^{\petit{1}}_p \cup X'^{\petit{2}}_p,(\ell'^{\petit{1}},\ell'^{\petit{2}})) \in E$, \ie complementary actions synchronize, corresponding to a communication; 
\item for $\tau_1 \in \Sigma_\tau^{\petit{1}}$, $\ell^{\petit{2}} \in L^{\petit{2}}$,
if $(\ell^{\petit{1}},g^{\petit{1}},\tau_1,X'^{\petit{1}}_p,\ell'^{\petit{1}}) \in E^{\petit{1}}$
then 
$((\ell^{\petit{1}},\ell^{\petit{2}}),g^{\petit{1}},\tau_1,
X'^{\petit{1}}_p,(\ell'^{\petit{1}},\ell^{\petit{2}})) \in E$, 
\ie internal actions of $\A_1$ progress independently;
\item for $\tau_2 \in \Sigma_\tau^{\petit{2}}$, $\ell^{\petit{1}} \in L^{\petit{1}}$,
if $(\ell^{\petit{2}},g^{\petit{2}},\tau_2,X'^{\petit{2}}_p,\ell'^{\petit{2}}) \in E^{\petit{2}}$
then 
$((\ell^{\petit{1}},\ell^{\petit{2}}),g^{\petit{2}},\tau_2,
X'^{\petit{2}}_p,(\ell^{\petit{1}},\ell'^{\petit{2}})) \in E$, 
\ie internal actions of $\A_2$ progress independently.
\end{iteMize}
\end{iteMize}
\end{defi}\medskip

\noindent By the definition of the transition relation $E$ of $\A^{\petit{1}} \| \A^{\petit{2}}$,
TAIOs synchronize exactly on complementary observable actions and time, 
and evolve independently on internal actions. 
As a consequence, the following equality on traces holds:
\begin{equation}
\label{eq-traces}
  \traces(\A^{\petit{1}} \| \A^{\petit{2}}) =
  \traces(\A^{\petit{1}})
  \cap
  \traces(\A^{\petit{2}}) 
\end{equation}

Notice that the definition is not absolutely symmetrical, 
as the  direction (input/output) of actions of the product is chosen 
with respect to $\A^{\petit{1}}$.
The technical reason is that, 
in the execution of a test case on an implementation,
we will need to keep the directions of actions of the implementation. 

\begin{center}
\begin{figure}[hbt]
\scalebox{0.65}{
\begin{minipage}{0.52\textwidth}
\begin{center}
  \begin{tikzpicture}[->,>=stealth',shorten >=1pt,auto,node distance=2cm,
                    semithick]

   \tikzstyle{every state}=[text=black]

  \node[state, fill=white] (A) {};
  \node[state, fill=white] (B) [below of=A, node distance=2cm, yshift=0cm] {};
  \node[fill=white] (B') [left of=B, node distance=1.1cm] {$y\le1$};
  \node[state, fill=white] (D) [below of=A, node distance=2cm, xshift=2cm] {};
  \node[state, fill=white] (C) [below of=B, node distance=2cm, yshift=0cm] {};
  \node[fill=white] (C'') [left of=C, node distance=1.1cm] {$y\le1$};
  \node[fill=white] (A') [left of=A, node distance=1.5cm, yshift=1cm] {$\A^{\petit{1}}$};
  \node[state, fill=white, color=white, scale=.5] (G) [above of=A, node distance=3cm, yshift=0cm] {};
  \node[fill=white] (C') [left of=C, node distance=3cm, yshift=-.5cm] {$X_p^{\petit{1}}=\{y\}$};
 
  \path (A) edge node [left] {$y\ge1,a?,\{y\}$} (B)
        (A) edge node [above, sloped] {$y\ge1,c?$} (D)
        (B) edge node [left] {$y\le 1,b!$} (C)
        (G) edge node [right] {} (A)
;
\end{tikzpicture}
\end{center}
\end{minipage}
\begin{minipage}{0.52\textwidth}
\begin{center}
 \begin{tikzpicture}[->,>=stealth',shorten >=1pt,auto,node distance=2cm,
                    semithick]

   \tikzstyle{every state}=[text=black]

  \node[state, fill=white] (A) {};
  \node[fill=white] (A'') [left of=A, node distance=1.1cm] {$x\le1$};
  \node[state, fill=white] (B) [below of=A, node distance=2cm, yshift=0cm] {};
  \node[state, fill=white] (C) [below of=B, node distance=2cm, yshift=0cm] {};
  \node[state, fill=white] (D) [below of=B, node distance=2cm, xshift=2cm] {};
  \node[fill=white] (A') [left of=A, node distance=1.5cm, yshift=1cm] {$\A^{\petit{2}}$};
  \node[state, fill=white, color=white, scale=.5] (G) [above of=A, node distance=3cm, yshift=0cm] {};
  \node[fill=white] (C') [left of=C, node distance=3cm, yshift=-.5cm] {$X_p^{\petit{2}}=\{x\}$};
 
  \path (A) edge node [left] {$x=1,a!,\{x\}$} (B)
        (B) edge node [left] {$x\ge1,b?$} (C)
        (B) edge node [above, sloped] {$x\ge1,c!$} (D)
        (G) edge node [right] {} (A)
;
\end{tikzpicture}
\end{center}
\end{minipage}
 \begin{minipage}{0.52\textwidth}
 \begin{center}
 \begin{tikzpicture}[->,>=stealth',shorten >=1pt,auto,node distance=2cm,
                    semithick]

   \tikzstyle{every state}=[text=black]

  \node[state, fill=white] (A) {};
  \node[fill=white] (A'') [left of=A, node distance=1.1cm] {$x\le1$};
  \node[state, fill=white] (B) [below of=A, node distance=2cm, yshift=0cm] {};
  \node[fill=white] (B') [left of=B, node distance=1.1cm] {$y\le1$};
  \node[state, fill=white] (C) [below of=B, node distance=2cm, yshift=0cm] {};
  \node[fill=white] (C'') [left of=C, node distance=1.1cm] {$y\le1$};
  \node[fill=white] (A') [left of=A, node distance=2cm, yshift=1cm] {$\A=\A^{\petit{1}} \| \A^{\petit{2}}$};
  \node[state, fill=white, color=white, scale=.5] (G) [above of=A, node distance=3cm, yshift=0cm] {};
  \node[fill=white] (C') [left of=C, node distance=3cm, yshift=-.5cm] {$X_p=\{x,y\}$};
 
  \path (A) edge node [left] {$y\ge1 \wedge x=1,a?,\{x,y\}$} (B)
        (B) edge node [left] {$y\le 1\wedge x\ge1,b!$} (C)
        (G) edge node [right] {} (A)
;
\end{tikzpicture} \end{center}
 \end{minipage}}

 \caption{Example of a parallel product $\A=\A^{\petit{1}}\|\A^{\petit{2}}$.}\label{ex_prodpar}
 \end{figure}
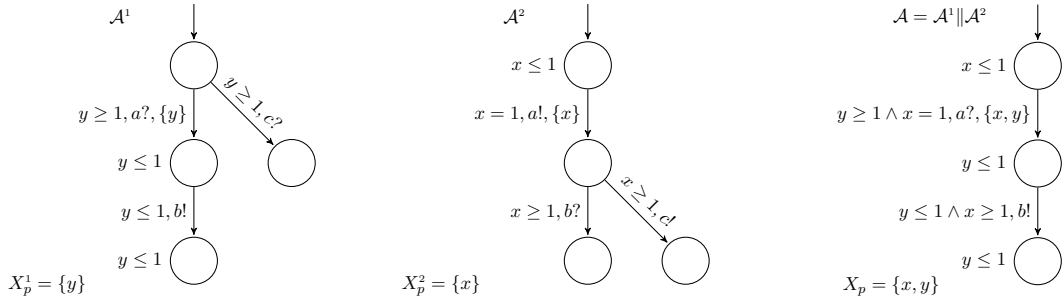   
 \end{center}

\begin{exa}

  The Figure~\ref{ex_prodpar} gives a very simple illustration of the
  parallel product. The intersection of the sets of traces is
  clear. Indeed, the parallel product recognizes exactly all
    prefixes of the trace $1.a.1.b$.  

\end{exa}

We now define a product operation on OTAIOs which extends the
classical product of TAs, with a particular attention to observed
clocks. This product is used later in the paper, to model the action of a test purpose which observes the clocks of a specification.

\begin{defi}[Product]\label{def_product}
Let 
 $\A^i=(L^{\petit{i}},\ell_0^{\petit{i}},\Sigma_?,\Sigma_!,\Sigma_\tau,X_p^{\petit{i}},X_o^{\petit{i}},M^{\petit{i}},\Inv^{\petit{i}},E^{\petit{i}})$,
$i=1,2$, be two OTAIOs with same alphabets and disjoint sets of proper clocks ($X_p^{\petit{1}} \cap X_p^{\petit{2}}= \emptyset$).
Their {\em product} 
is the OTAIO 
$\A^{\petit{1}} \times \A^{\petit{2}} = (L,\ell_0,\Sigma_?,\Sigma_!,\Sigma_\tau,X_p,X_o,M,\Inv,E)$ where:
\begin{iteMize}{$\bullet$}
\item 
$L=L^{\petit{1}} \times L^{\petit{2}}$; 
\item 
$\ell_0 = (\ell_0^{\petit{1}},\ell_0^{\petit{2}})$;
\item 
$X_p = X_p^{\petit{1}} \sqcup X_p^{\petit{2}}$, $X_o = (X_o^{\petit{1}} \cup X_o^{\petit{2}}) \setminus X_p$;
\item 
$M= \max(M^{\petit{1}}, M^{\petit{2}})$;
\item
$\forall (\ell^{\petit{1}},\ell^{\petit{2}}) \in L, \Inv((\ell^{\petit{1}},\ell^{\petit{2}})) = \Inv^{\petit{1}}(\ell^{\petit{1}}) \wedge \Inv^{\petit{2}}(\ell^{\petit{2}})$;
\item 
$((\ell^{\petit{1}},\ell^{\petit{2}}),g^{\petit{1}} \wedge g^{\petit{2}},a,X'^{\petit{1}}_p \sqcup X'^{\petit{2}}_p,(\ell'^{\petit{1}},\ell'^{\petit{2}})) \in E$ if
$ (\ell^{\petit{i}},g^{\petit{i}},a,X'^{\petit{i}}_p,\ell'^{\petit{i}}) \in E^{\petit{i}}$, i=1,2.
\end{iteMize}
\end{defi}

Intuitively, $\A^{\petit{1}}$ and $\A^{\petit{2}}$ synchronize on both
time and common actions (including internal
ones\footnote{
Synchronizing internal actions allows for more
    precision in test selection. This justifies to have a set of
    internal actions in the TAIO model.
}).  $\A^{\petit{2}}$ may
observe proper clocks of $\A^{\petit{1}}$ using its observed clocks
$X_p^{\petit{1}} \cap X_o^{\petit{2}}$, and {\em vice versa}.  The set
of proper clocks of $\A^{\petit{1}} \times A^{\petit{2}}$ is the union
of proper clocks of $\A^{\petit{1}}$ and $A^{\petit{2}}$, and observed
clocks of $\A^{\petit{1}} \times A^{\petit{2}}$ are observed clocks of
any OTAIO which are not proper.  For example, the OTAIO in
Figure~\ref{ExProd} represents the product of the TAIO $\A$ in
Figure~\ref{ExSpec} and the OTAIO $\TP$ of Figure~\ref{ExObj}.

\begin{center}
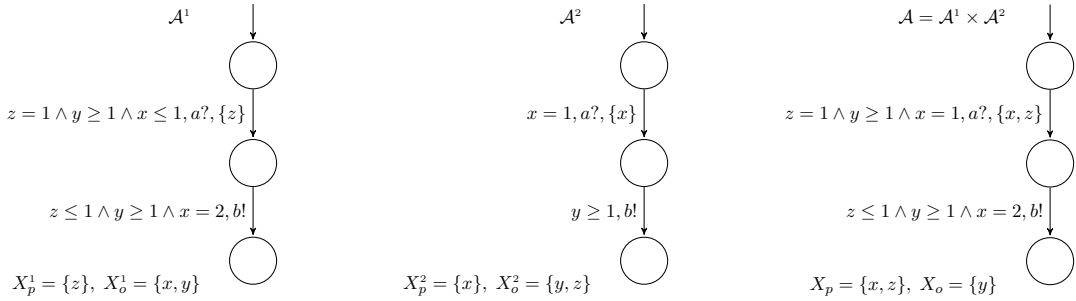
\begin{figure}[hbt]
\scalebox{0.65}{
\begin{minipage}{0.52\textwidth}
\begin{center}
  \begin{tikzpicture}[->,>=stealth',shorten >=1pt,auto,node distance=2cm,
                    semithick]

   \tikzstyle{every state}=[text=black]

  \node[state, fill=white] (A) {};
  \node[state, fill=white] (B) [below of=A, node distance=2cm, yshift=0cm] {};
  \node[state, fill=white] (C) [below of=B, node distance=2cm, yshift=0cm] {};
  \node[fill=white] (A') [left of=A, node distance=1.5cm, yshift=1cm] {$\A^{\petit{1}}$};
  \node[state, fill=white, color=white, scale=.5] (G) [above of=A, node distance=3cm, yshift=0cm] {};
  \node[fill=white] (C') [left of=C, node distance=3cm, yshift=-.5cm] {$X_p^{\petit{1}}=\{z\},\;X_o^{\petit{1}}=\{x,y\}$};
 
  \path (A) edge node [left] {$z=1\wedge y\ge1 \wedge x\le1,a?,\{z\}$} (B)
        (B) edge node [left] {$z\le1 \wedge y\ge 1\wedge x=2,b!$} (C)
        (G) edge node [right] {} (A)
;
\end{tikzpicture}
\end{center}
\end{minipage}
\begin{minipage}{0.52\textwidth}
\begin{center}
 \begin{tikzpicture}[->,>=stealth',shorten >=1pt,auto,node distance=2cm,
                    semithick]

   \tikzstyle{every state}=[text=black]

  \node[state, fill=white] (A) {};
  \node[state, fill=white] (B) [below of=A, node distance=2cm, yshift=0cm] {};
  \node[state, fill=white] (C) [below of=B, node distance=2cm, yshift=0cm] {};
  \node[fill=white] (A') [left of=A, node distance=1.5cm, yshift=1cm] {$\A^{\petit{2}}$};
  \node[state, fill=white, color=white, scale=.5] (G) [above of=A, node distance=3cm, yshift=0cm] {};
  \node[fill=white] (C') [left of=C, node distance=3cm, yshift=-.5cm] {$X_p^{\petit{2}}=\{x\},\;X_o^{\petit{2}}=\{y,z\}$};
 
  \path (A) edge node [left] {$x=1,a?,\{x\}$} (B)
        (B) edge node [left] {$y\ge 1,b!$} (C)
        (G) edge node [right] {} (A)
;
\end{tikzpicture}
\end{center}
\end{minipage}
 \begin{minipage}{0.52\textwidth}
 \begin{center}
 \begin{tikzpicture}[->,>=stealth',shorten >=1pt,auto,node distance=2cm,
                    semithick]

   \tikzstyle{every state}=[text=black]

  \node[state, fill=white] (A) {};
  \node[state, fill=white] (B) [below of=A, node distance=2cm, yshift=0cm] {};
  \node[state, fill=white] (C) [below of=B, node distance=2cm, yshift=0cm] {};
  \node[fill=white] (A') [left of=A, node distance=2cm, yshift=1cm] {$\A=\A^{\petit{1}} \times \A^{\petit{2}}$};
  \node[state, fill=white, color=white, scale=.5] (G) [above of=A, node distance=3cm, yshift=0cm] {};
  \node[fill=white] (C') [left of=C, node distance=3cm, yshift=-.5cm] {$X_p=\{x,z\},\;X_o=\{y\}$};
 
  \path (A) edge node [left] {$z=1\wedge y\ge1 \wedge x=1,a?,\{x,z\}$} (B)
        (B) edge node [left] {$z\le1 \wedge y\ge 1\wedge x=2,b!$} (C)
        (G) edge node [right] {} (A)
;
\end{tikzpicture} \end{center}
 \end{minipage}}

 \caption{Example of a product $\A = \A^{\petit{1}} \times \A^{\petit{2}}$.}\label{ex_prod}
 \end{figure}   
 \end{center}

Contrary to the parallel product, 
the set of traces of the product of two OTAIOs 
is not the intersection of the sets of traces of these TAIOs,
as illustrated by the following example. 

\begin{exa}
  Figure~\ref{ex_prod} artificially illustrates the notion of product
  of two OTAIOs.  One can see that $1.a?.1.b!$ is a trace of
  $\A^{\petit{1}}$ and $\A^{\petit{2}}$ but is not a trace of $\A
    = \A^{\petit{1}} \times \A^{\petit{2}}$.  Indeed, in
  $\A^{\petit{1}}$, $1.a?.1.b!$ is the trace of a sequence where $x$
  is not reset at the first action. Unfortunately, the clock $x$ is
  observed by $\A^{\petit{1}}$ but is a proper clock of
  $\A^{\petit{2}}$ which resets it at the first action. As a
  consequence, $1.a?.1.b!$ cannot be a trace of the product
  $\A^{\petit{1}} \times \A^{\petit{2}}$. In fact, the second edge
    in $\A$ can never be fired, 
  since clocks $z$ and $x$ agree on their values and cannot be
  simultaneously smaller than $1$ and equal to $2$.
\end{exa}

On the other hand, sequences are more adapted to express the underlying operation.
 To compare the sets of sequences of
  $\A^{\petit{1}} \times \A^{\petit{2}}$ with the sets of sequences of
  its factors, we introduce an operation that lifts the sets of clocks
  of factors to the set of clocks of the product: for $\A^{\petit{1}}$
  defined on $(X_p^{\petit{1}},X_o^{\petit{1}})$, and $X_p^{\petit{1}}
  \cap X_p^{\petit{2}}= \emptyset$,
  $\A^{\petit{1}}\!\!\uparrow^{(X^{\petit{2}}_p,X^{\petit{2}}_o)}$ denotes
  an automaton identical to $\A^{\petit{1}}$
  but defined on
  $(X_p^{\petit{1}},X_p^{\petit{2}}\cup X_o^{\petit{1}} \cup
  X_o^{\petit{2}}\setminus X_p^{\petit{1}})$.  The effect on the semantics is to
  duplicate moves of $\A^{\petit{1}}$ with unconstrained resets in
  $(X_p^{\petit{2}}\cup X_o^{\petit{2}}) \setminus  (X_p^{\petit{1}}\cup
  X_o^{\petit{1}})$, so that
  $\A^{\petit{1}}\!\!\uparrow^{(X^{\petit{2}}_p,X^{\petit{2}}_o)}$
  strongly bisimulates $\A^{\petit{1}}$.  
  The equivalence just consists in ignoring values of added clocks which do not interfere in the guards.
  Similarly
  $\A^{\petit{2}}\!\!\uparrow^{(X^{\petit{1}}_p,X^{\petit{1}}_o)}$ is
  defined on $(X_p^{\petit{2}},X_p^{\petit{1}}\cup X_o^{\petit{2}}
  \cup X_o^{\petit{1}}\setminus X_p^{\petit{2}})$.  
  Both
  $\A^{\petit{1}}\!\!\uparrow^{X^{\petit{2}}_p,X^{\petit{2}}_o}$ and
  $\A^{\petit{2}}\!\!\uparrow^{X^{\petit{1}}_p,X^{\petit{1}}_o}$ have
  sequences in
  $(\setRnn \sqcup (\Sigma_\tau^\pA \times (X_p^{\petit{1}} \cup X_p^{\petit{2}} \cup X_o^{\petit{1}} \cup X_o^{\petit{2}})))^*$. They synchronize on both delays and common actions with their resets.
The effect of the product is to restrict the respective environments  (observed clocks) 
by imposing the resets of the peer TAIO. 
The sequences of the product  are then characterized by\\
\begin{equation}
\label{eq-seq}
  \seq(\A^{\petit{1}} \times \A^{\petit{2}}) =
  \seq(\A^{\petit{1}}\!\!\uparrow^{(X^{\petit{2}}_p,X^{\petit{2}}_o)})
  \cap
  \seq(\A^{\petit{2}}\!\!\uparrow^{(X^{\petit{1}}_p,X^{\petit{1}}_o)})
\end{equation}
 meaning that the product of OTAIOs is the adequate operation for
intersecting sets of sequences.

\ignore{On the other hand, the product is the right operation for
intersecting sets of sequences.
In fact, let $\A^{\petit{1}}\!\!\uparrow^{(X^{\petit{2}}_p,X^{\petit{2}}_o)}$ (respectively
$\A^{\petit{2}}\!\!\uparrow^{(X^{\petit{1}}_p,X^{\petit{1}}_o)}$) denote the same
TAIO $\A^{\petit{1}}$ (resp.  $\A^{\petit{2}}$) defined on
$(X_p^{\petit{1}},X_p^{\petit{2}}\cup X_o^{\petit{2}} \cup
X_o^{\petit{1}}\setminus X_p^{\petit{1}})$ (resp. on
$(X_p^{\petit{2}},X_p^{\petit{1}}\cup X_o^{\petit{2}} \cup
X_o^{\petit{2}}\setminus X_p^{\petit{2}})$).  Then we get: 
\begin{equation}
\seq(\A^{\petit{1}} \times \A^{\petit{2}}) = \seq(\A^{\petit{1}}\!\!\uparrow^{(X^{\petit{2}}_p,X^{\petit{2}}_o)})
                         \cap \seq(\A^{\petit{2}}\!\!\uparrow^{(X^{\petit{1}}_p,X^{\petit{1}}_o)})
\end{equation}
}

An OTAIO  equipped 
with a set of states $F \subseteq S^\pA$ can play the role of an acceptor.
A run is  {\em accepted} in $F$ if it ends in $F$.
$\seq_{F}(\A)$ denotes the set of sequences of accepted runs
and $\traces_{F}(\A)$ the set of their traces.
By abuse of notation, if $L$ is a subset of locations in $L^\pA$,
we note $\seq_L(\A)$ for $\seq_{L\times\setRnn^{X^\pA}}(\A)$ and 
similarly for $\traces_L(\A)$. 
Note that for the product $\A^{\petit{1}} \times \A^{\petit{2}}$, if $F^{\petit{1}}$ and $F^{\petit{2}}$
are subsets of states of $\A^{\petit{1}}$ and $\A^{\petit{2}}$ respectively, additionally to~(\ref{eq-seq}), the following equality holds:
\begin{equation}
\label{eq-seq-acc} 
  \seq_{F^{\petit{1}}\times F^{\petit{2}}}(\A^{\petit{1}} \times
  \A^{\petit{2}}) =
  \seq_{F^{\petit{1}}}(\A^{\petit{1}}\!\!\uparrow^{(X^{\petit{2}}_p,X^{\petit{2}}_o)})
  \cap
  \seq_{F^{\petit{2}}}(\A^{\petit{2}}\!\!\uparrow^{(X^{\petit{1}}_p,X^{\petit{1}}_o)}).
\end{equation}

\section{Conformance testing theory}
\label{sec-conformance}

In this section, we recall the conformance theory for timed automata
based on the conformance relation \tioco~\cite{KrichenTripakis09}
 that formally defines the set of correct
implementations of a given TAIO specification.  $\tioco$ is a natural
extension of the $\ioco$ relation of Tretmans~\cite{Tretmans-SCT96} to timed
systems.  We then define test cases, formalize their executions,
verdicts and expected properties relating verdicts to conformance. 
Finally, we introduce a refinement
relation between TAIOs that preserves \tioco, and will be useful in proving 
test case properties.

\subsection{The \tioco~conformance theory}
We consider that the specification is given as a  (possibly
non-deterministic) TAIO $\A$.
The implementation is a black box, unknown
except for its alphabet of observable actions, which is the same as
the one of $\A$.  
As usual, in order to formally
reason about conformance, we assume that the implementation can be
modeled by an (unknown) TAIO.
Formally:
\begin{defi}[Implementation]
Let $\A=(L^\pA,\ell_0^\pA,\Sigma_?^\pA,\Sigma_!^\pA,\Sigma_\tau^\pA,X_p^\pA,\emptyset,M^\pA,\Inv^\pA,E^\pA)$ be a specification TAIO.
An implementation of $\A$ is an input-complete and non-blocking TAIO
$\Imp=(L^\imp,\ell_0^\imp,\Sigma_?,\Sigma_!,\Sigma_\tau^\imp,X_p^\imp,\emptyset,M^\imp,\Inv^\imp,E^\imp)$
with same observable alphabet as $\A$ ($\Sigma_?^\imp=\Sigma_?^\pA$ and $\Sigma_!^\imp=\Sigma_!^\pA$).
$\Imp(\A)$ denotes the set of possible implementations of $\A$.
\end{defi}

The requirements that an implementation is input-complete and non-blocking
will ensure that the execution of a test case on $\Imp$ does not
block before verdicts are emitted.

Among the possible implementations in $\Imp(\A)$, 
the conformance relation $\tioco$ (for {\em timed input-output conformance})~\cite{KrichenTripakis09} 
formally defines which ones conform to $\A$,
naturally extending the classical $\ioco$ relation of Tretmans~\cite{Tretmans-SCT96} to timed
systems:
\begin{defi}[Conformance relation]
Let $\A$ be a TAIO representing the specification and $\Imp \in \Imp(\A)$ be an implementation of $\A$. We say that $\Imp$ conforms to $\A$ and write
$\Imp \; \tioco \; \A \mbox{ if } \forall \sigma \in \traces(\A), out(\Imp \after \sigma) \subseteq out(\A \after \sigma)$.
\end{defi}

Note that \tioco\, is equivalent to the \rtioco\, relation  that was defined
independently in~\cite{Larsen-Mikucionis-Nielsen-2004} (see~\cite{DBLP:conf/formats/SchmaltzT08}).
Intuitively, $\Imp$ conforms to $\A$ if after any timed trace enabled
in $\A$, every output or delay of $\Imp$ is specified in $\A$. 
This means that $\Imp$ may accept more inputs than $\A$,
but is authorized to send less outputs, 
or send them during a more restricted time interval. 
The intuition is illustrated on the following simple example:

\begin{exa}
Figure~\ref{imps} represents a specification $\A$ and two possible 
implementations $\Imp_1$ and $\Imp_2$.
Note that $\Imp_1$ and $\Imp_2$ should be input-complete, but for simplicity of figures, 
we omit some inputs and consider that missing inputs loop 
to the current location. 
It is easy to see that $\Imp_1$ conforms to $\A$.
Indeed, it accepts more inputs, which is allowed
(after the trace $\epsilon$, $\Imp_1$ can receive $a$ and $d$
while $\A$ only accepts $a$), 
and emits the output $b$ during a more restricted interval 
of time ($out(\Imp_1 \after a.2)=[0,\infty)$ is included in $out(\A \after a.2)=[0,\infty) \sqcup \{b\}$).
On the other hand $\Imp_2$ does not conform to $\A$ for 
two reasons:
$\Imp_2$  may send a new output $c$ 
and may send $b$ during a larger time interval 
(\eg~$out(\Imp_2 \after a.1)=[0,\infty) \sqcup \{b,c\}$
is not included in ~$out(\A \after a.1)=[0,\infty)$). 
\end{exa}

\begin{center}
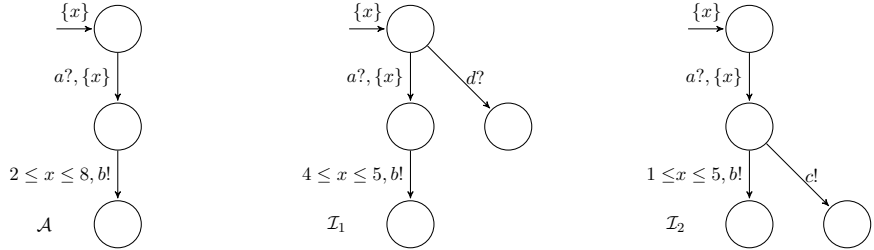
\begin{figure}[hbt]
\scalebox{0.65}{
\begin{minipage}{0.45\textwidth}
\begin{center}
  \begin{tikzpicture}[->,>=stealth',shorten >=1pt,auto,node distance=2cm,
                    semithick]

   \tikzstyle{every state}=[text=black]

  \node[state, fill=white] (A) {};
  \node[state, fill=white] (B) [below of=A, node distance=2cm, yshift=0cm] {};
  \node[state, fill=white] (C) [below of=B, node distance=2cm, yshift=0cm] {};
  \node[fill=white] (C') [left of=C, node distance=1.5cm, yshift=0cm] {$\A$};
  \node[state, fill=white, color=white, scale=.5] (G) [left of=A, node distance=3cm, yshift=0cm] {};
 
  \path (A) edge node [left] {$a?,\{x\}$} (B)
        (B) edge node [left] {$2\le x\le 8,b!$} (C)
        (G) edge node [above] {$\{x\}$} (A)
;
\end{tikzpicture}
\end{center}
\end{minipage}
\begin{minipage}{0.45\textwidth}
\begin{center}
  \begin{tikzpicture}[->,>=stealth',shorten >=1pt,auto,node distance=2cm,
                    semithick]

   \tikzstyle{every state}=[text=black]

  \node[state, fill=white] (A) {};
  \node[state, fill=white] (B) [below of=A, node distance=2cm, yshift=0cm] {};
  \node[state, fill=white] (C) [below of=B, node distance=2cm, yshift=0cm] {};
  \node[fill=white] (C') [left of=C, node distance=1.5cm, yshift=0cm] {$\Imp_1$};
  \node[state, fill=white, color=white, scale=.5] (G) [left of=A, node distance=3cm, yshift=0cm] {};
   \node[state, fill=white] (D) [right of=B, node distance=2cm, yshift=0cm] {};
  \path (G) edge node [above] {$\{x\}$} (A)
        (A) edge node [left] {$a?,\{x\}$} (B)
        (B) edge node [left] {${4 \leq x \leq 5},b!$} (C)
        (A) edge node [right] {${d?}$} (D)
;
\end{tikzpicture}
\end{center}
\end{minipage}
 \begin{minipage}{0.45\textwidth}
 \begin{center}
  \begin{tikzpicture}[->,>=stealth',shorten >=1pt,auto,node distance=2cm,
                    semithick]

   \tikzstyle{every state}=[text=black]

  \node[state, fill=white] (A) {};
  \node[state, fill=white] (B) [below of=A, node distance=2cm, yshift=0cm] {};
  \node[state, fill=white] (C) [below of=B, node distance=2cm, yshift=0cm] {};
  \node[fill=white] (C') [left of=C, node distance=1.5cm, yshift=0cm] {$\Imp_2$};
  \node[state, fill=white, color=white, scale=.5] (G) [left of=A, node distance=3cm, yshift=0cm] {};
  \node[state, fill=white] (D) [right of=C, node distance=2cm, yshift=0cm] {};
 
  \path (G) edge node [above] {$\{x\}$} (A)
        (A) edge node [left] {$a?,\{x\}$} (B)
        (B) edge node [left] {${1\le} x\le 5,b!$} (C)
        (B) edge node [right] {${c!}$} (D)        
;
 \end{tikzpicture}
 \end{center}
 \end{minipage}}

 \caption{Example of a specification $\A$ and two implementations $\Imp_1$ and $\Imp_2$.}\label{imps}
 \end{figure}   
 \end{center}
In practice, conformance is checked by test cases run on implementations.
In our setting, we define test cases as deterministic TAIOs
equipped with verdicts defined by a partition of states.

\begin{defi}[Test suite, test case]
Given a specification TAIO $\A$, 
a {\em test suite} is a set of {\em test cases},
where 
a {\em test case} is a pair $(\TC,\Verdicts)$
consisting of:
\begin{iteMize}{$\bullet$}
\item a deterministic TAIO 
$\TC=(L^\tc,\ell_0^\tc,\Sigma_?^\tc,\Sigma_!^\tc,\emptyset,X_p^\tc,\emptyset,M^\tc,\Inv^\tc,E^\tc)$,
\item a partition $\Verdicts$ of the set of states 
$S^\TC=\None \sqcup  \Inconc \sqcup \Pass \sqcup \Fail$.
States outside $\None$ are called {\em verdict states}. 
\end{iteMize} 
We also require that
\begin{iteMize}{$\bullet$}
\item 
$\Sigma_?^\tc=\Sigma_!^\pA$ and $\Sigma_!^\tc=\Sigma_?^\pA$,
\item 
$\TC$ is non-blocking, (\eg\ $\Inv^\tc(\ell)= \true$ for all $\ell \in L^\tc$),%
\item 
 $\TC$ is input-complete in all $\None$ states, 
meaning that it is ready to receive any input from the implementation
before reaching a verdict.
\end{iteMize}
\end{defi}\medskip

\noindent In the following, for simplicity we will sometimes abuse notations and
write $\TC$ instead of $(\TC, \Verdicts)$.  
Let us give some intuition about the different verdicts of test
cases. $\Fail$ states are those where the test case rejects an
implementation.  The intention is thus to detect a non-conformance.
$\Pass$ and $\Inconc$ states are linked to test purposes (see
Section~\ref{sec-generation}): the intention is that $\Pass$ states
should be those where no non-conformance has been detected and the
test purpose is satisfied, whereas $\Inconc$ states should be those
states where no non-conformance has been detected, but the test
purpose cannot be satisfied anymore.  $\None$ states are all other
states.  We insist on the fact that those are intentional
characterizations of the verdicts. Properties of test cases defined
later specify whether these intentions are satisfied by test cases.
We will see that it is not always the case for all properties.

The execution of a test case $\TC\in Test(\A)$ on an implementation $\Imp \in \Imp(\A)$
is  modeled by the parallel product $ \Imp \|\TC$,
which entails that $\traces(\Imp \| \TC) = 
\traces(\Imp) \cap \traces(\TC)$.
The facts that $\TC$ is input-complete (in $\None$ states) 
and non-blocking 
while  
$\Imp$ is input-complete (in all states) and non-blocking
ensure that no deadlock occurs before a verdict is reached.

We say that the verdict of an execution of trace $\sigma \in \traces(\TC)$,
noted ${\tt Verdict}(\sigma,\TC)$, is $\Pass$, $\Fail$, $\Inconc$ or $\None$
if $\TC \after \sigma$ is included in the corresponding states
set~\footnote{Note that $TC$ being deterministic, $\TC \after \sigma$ is a singleton.}.  
We write $\Imp \; {\tt fails} \; \TC$
if some execution $\sigma$ of $\Imp \| \TC$ leads $\TC$ to a $\Fail$
state, \ie when $\traces_{\pFail}(\TC) \cap \traces(\Imp) \neq
\emptyset$, which means that there exists $\sigma \in \traces(\Imp)
\cap \traces(\TC)$ such that 
${\tt Verdict}(\sigma,\TC)= \Fail$.
\ignore{~\footnote{The execution of a test case $\TC$ 
on an implementation $\Imp$
is usually modeled by the standard parallel composition $\Imp  \| \TC$. 
Due to space limitations, $\|$ is not defined here,
but we use its trace properties: $\traces(\Imp \| \TC) = 
\traces(\Imp) \cap \traces(\TC)$.}.}
Notice that this is only a possibility to reach the
$\Fail$ verdict among the infinite set of executions of $\Imp  \| \TC$.  
Hitting one of these executions is not ensured both because 
of the lack of control of $\TC$ 
on $\Imp$ and of timing constraints imposed by these executions.
 
We now introduce soundness, a crucial property ensured by our test generation method. 
We also introduce exhaustiveness 
 and strictness that will be ensured  
when determinization is exact (see Section~\ref{sec-generation}).

\begin{defi}[Test suite soundness, exhaustiveness and strictness]
A  test suite $\TS$ for $\A$ is:
\begin{iteMize}{$\bullet$}
\item {\em sound} if
$\forall \Imp \in \Imp(\A)$, $\forall \TC \in \TS$, $\Imp \; {\tt fails} \; \TC \Ra \neg (\Imp \; \tioco \;\A)$,
\item {\em exhaustive} if 
$\forall \Imp \in \Imp(\A)$, $\neg (\Imp \; \tioco \; \A) \Ra \exists \TC \in \TS$, $\Imp \; {\tt fails} \; \TC$,
\item {\em strict} if  
$\forall \Imp \in \Imp(\A), \forall \TC \in \TS, 
\neg (\Imp \| \TC \; \tioco \; \A) \Ra \Imp \; {\tt fails} \; \TC$.
\end{iteMize}
\end{defi}

Intuitively, soundness means that no conformant implementation can be
rejected by the test suite, \ie any failure of a test case during its
execution characterizes a non-conformance.  Conversely, exhaustiveness means
that every non-conformant implementation may be rejected by the test
suite.  Remember that the definition of $\Imp \; {\tt fails} \; \TC$
indicates only a possibility of reject.  Finally, strictness means
that non-conformance is detected once it occurs.  In fact, $\neg (\Imp
\| \TC \; \tioco \; \A)$ means that there is a trace common to $\TC$
and $\Imp$ which does not conform to $\A$.  The universal
quantification on $\Imp$ and $\TC$ implies that any such trace will
fail $\TC$. In particular, this implies that failure will be detected
as soon as it occurs.  


\begin{center}
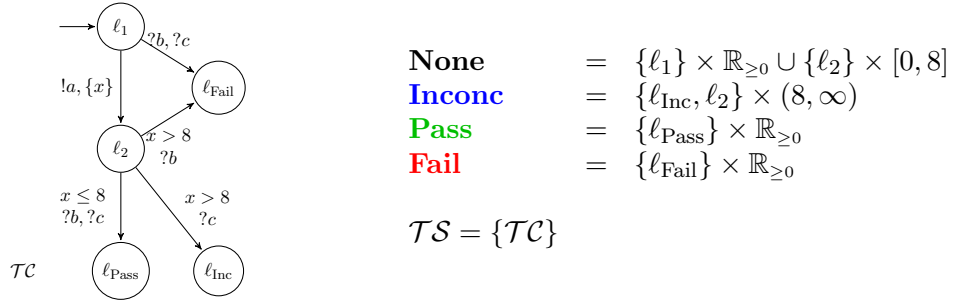
\begin{figure}[hbt]
\scalebox{0.65}{
\begin{minipage}{0.7\textwidth}
\begin{center}
  \begin{tikzpicture}[->,>=stealth',shorten >=1pt,auto,node distance=2cm,
                    semithick]

   \tikzstyle{every state}=[text=black]

  \node[state, fill=white] (A) {$\ell_1$};
  \node[state, fill=white] (B) [below of=A, node distance=2.5cm, yshift=0cm] {$\ell_2$};
  \node[state, fill=white] (C) [below of=B, node distance=2.5cm, yshift=0cm] {$\ell_{\mbox{\scriptsize{Pass}}}$};
  \node[fill=white] (C') [left of=C, node distance=2cm, yshift=0cm] {$\TC$};
  \node[state, fill=white, color=white, scale=.5] (G) [left of=A, node distance=3cm, yshift=0cm] {};
   \node[state, fill=white] (D) [right of=B, node distance=2cm, yshift=1.25cm] {$\ell_{\mbox{\scriptsize{Fail}}}$};
   \node[state, fill=white] (E) [right of=C, node distance=2cm, yshift=0cm] {$\ell_{\mbox{\scriptsize{Inc}}}$};
  \path (G) edge node [above] {} (A)
        (A) edge node [left] {$!a,\{x\}$} (B)
        (B) edge node [left] {$\begin{array}{c}x\le8\\ ?b,?c\end{array}$} (C)
        (A) edge node [above] {$?b,?c$} (D)
        (B) edge node [below] {$\begin{array}{c}x>8\\?b\end{array}$} (D)
        (B) edge node [right] {$\begin{array}{c}x>8\\?c\end{array}$} (E)
;
\end{tikzpicture}
\end{center}
\end{minipage}}
 \begin{minipage}{0.5\textwidth}
$\begin{array}{lll}\None&=&\{\ell_1\}\times \setRnn\cup\{\ell_2\}\times\lbrack0,8\rbrack \\
 \bleu{\Inconc}&=&\{\ell_{\mbox{\scriptsize{Inc}}},\ell_2\}\times (8, \infty)\\
  \verte{\Pass}&=&\{\ell_{\mbox{\scriptsize{Pass}}}\}\times\setRnn\\ 
  \rouge{\Fail}&=&\{\ell_{\mbox{\scriptsize{Fail}}}\}\times\setRnn\\
  \\
  \TS=\{\TC\}
  \end{array}$
 \end{minipage}
 \caption{Example of a sound but not strict test suite for the specification $\A$ (Figure~\ref{imps}).}\label{ex_sound}
 \end{figure}   
 \end{center}

\begin{exa}
Figure~\ref{ex_sound} represents a test suite composed of a single test case $\TC$ for the specification $\A$ of the Figure~\ref{imps}. Indeed, $\TC$ is a TAIO which is input-complete in the $\None$ states.
$\TS$ is sound because
the $\Fail$ states of $\TC$ are reached only when a conformance error occurs,
\eg~on trace $1.b$. 
However, this test case can observe non-conformant traces without detecting them, hence $\TS$ is not strict. For example, $1.a.1.b$, $1.a.1.c$ and $1.a.9.c$ are non-conformant traces that do not imply a  $\Fail$ verdict. 
These traces are \eg~traces of $\Imp_2$ (Figure~\ref{imps}) which should allow
to detect that $\neg(\Imp_2 \,\tioco\, \A)$.
\end{exa}

\subsection{Refinement preserving \tioco}
We introduce an io-refinement relation between two TAIOs, a
generalization to non-deterministic TAIOs of the 
io-refinement  between DTAIOs introduced
in~\cite{David-Larsen-etal-HSCC10}, itself a generalization of
alternating simulation~\cite{Alur-Henzinger-Kupferman-Vardi-CONCUR98}.
Informally $\A$ io-refines $\B$ if $\A$ specifies more inputs and allows 
less outputs and delays.
As a consequence,  if $\A$ and $\B$ are specifications,
$\A$ is more restrictive than $\B$ with respect to conformance.
We thus prove that io-abstraction (the inverse relation) preserves \tioco:
if $\Imp$ conforms to $\A$, it also conforms to any io-abstraction
$\B$ of $\A$.  This will  ensure that  soundness of test cases is
preserved by the approximate determinization defined in Section~\ref{sec-determinisation}.

\begin{defi}
  Let $\A$ and $\B$ be two TAIOs with same input and output alphabets,
  we say that $\A$ {\em io-refines} $\B$ (or $\B$ {\em io-abstracts}
  $\A$) and note $\A \preceq \B$ if
\begin{equation*}
\begin{aligned}
(i) \quad \forall \sigma \in \traces(\B), \; & out(\A \after  \sigma)  \subseteq out(\B  \after \sigma) \text{ and,}\\
(ii) \quad \forall \sigma \in \traces(\A), \; & in(\B  \after  \sigma)  \subseteq in(\A  \after  \sigma).
\end{aligned}
\end{equation*}
\end{defi}

As we will see below, $\preceq$ is a preorder relation.  
Moreover, as condition (ii)
is always satisfied if $\A$ is input-complete, for $\Imp \in
\Imp(\A)$, $\Imp \; \tioco\; \A$ is equivalent to $\Imp \preceq \A$.
By transitivity of $\preceq$, it follows that io-refinement preserves
conformance (see Proposition~\ref{prop:tioco-alt-sim}).

\begin{lem}
The io-refinement $\preceq$ is a preorder relation.
\end{lem}
\begin{proof}
The relation $\preceq$ is trivially reflexive and we prove that it is transitive.

Suppose that $\A \preceq \B$ and $\B \preceq \C$.
By definition of $\preceq$ we have:
\begin{equation*}
\begin{aligned}
\forall \sigma \in \traces(\B), \; & out(\A \after  \sigma)  \subseteq out(\B  \after \sigma) \quad (1)\\
\forall \sigma \in \traces(\A), \; & in(\B  \after  \sigma)  \subseteq in(\A  \after  \sigma) \quad (2)  \quad and \\
\forall \sigma \in \traces(\C), \; & out(\B \after  \sigma)  \subseteq out(\C  \after \sigma) \quad (3)\\
\forall \sigma \in \traces(\B), \; & in(\C  \after  \sigma)  \subseteq in(\B  \after  \sigma) \quad (4)
\end{aligned}
\end{equation*}
We want to prove that $\A \preceq \C$ thus that
\begin{equation*}
\begin{aligned}
\forall \sigma \in \traces(\C), \; & out(\A \after  \sigma)  \subseteq out(\C  \after \sigma) \quad (5) \\
\forall \sigma \in \traces(\A), \; & in(\C  \after  \sigma)  \subseteq in(\A  \after  \sigma) \quad (6)
\end{aligned}
\end{equation*}
In order to  prove $(5)$, 
let $\sigma \in \traces(\C)$, and examine the two cases: 
\begin{iteMize}{$\bullet$}
\item If $\sigma \in \traces(\B) \cap \traces(\C) $ then $(1)$ and $(3)$ 
imply    $out(\A \after  \sigma)  \subseteq out(\B  \after \sigma)$ and 
$out(\B  \after \sigma) \subseteq out(\C  \after \sigma)$. Thus 
$out(\A  \after \sigma) \subseteq out(\C  \after \sigma)$ and we are done.
\item If $\sigma \in \traces(\C)\setminus \traces(\B)$, 
there exist $\sigma', \sigma'' \in (\Sigma_{obs} \sqcup \setRnn)^*$ 
and $a\in \Sigma_{obs} \sqcup \setRnn$ such that 
$\sigma=\sigma'.a.\sigma''$ with $\sigma'\in \traces(\B) \cap \traces(\C)$ 
and   $\sigma'.a \in \traces(\C)\setminus \traces(\B)$.
As $\B \preceq \C$, by $(4)$ we get that $a \in \Sigma_! \sqcup \setRnn$. 
But as $\A \preceq \B$, and $\sigma' \in \traces(\B)$, 
the condition $(1)$ induces that $out(\A \after  \sigma')  \subseteq out(\B  \after \sigma')$,
and then $\sigma'.a \in \traces(\C)\setminus \traces(\A)$.
We deduce that  $out(\A \after  \sigma'.a)=\emptyset$,  and thus
 $out(\A \after  \sigma) = \emptyset \subseteq out(\C \after  \sigma)$.
\end{iteMize}
The proof of $(6)$ is similar.
\end{proof}


\begin{prop}
\label{prop:tioco-alt-sim}
If $\A \preceq \B$ then $\forall \Imp \in \Imp(\A)$ $(=\Imp(\B))$, 
$\Imp \; \tioco \; \A \Ra \Imp \; \tioco \; \B$.
\end{prop}
\begin{proof}
%

This proposition 
is a direct consequence of the transitivity of $\preceq$.
In fact when $\Imp$ is input-complete, by definition $\forall \sigma \in \traces(\Imp), in(\Imp \after \sigma) = \Sigma_?$, thus condition (ii) of $\preceq$ 
trivially holds:
$\forall \sigma \in \traces(\Imp), in(\A \after \sigma) \subseteq in(\Imp \after \sigma)$.
Thus $\Imp \,\tioco\, \A$ (which is defined by $\forall \sigma \in \traces(\A), out(\Imp \after \sigma) \subseteq out(\A \after \sigma)$) is equivalent to $\Imp \preceq \A$.
Now suppose $\A \preceq B$ and $\Imp \,\tioco\, \A$ then the transitivity of $\preceq$ gives $\Imp \,\tioco \,\B$.
\end{proof}

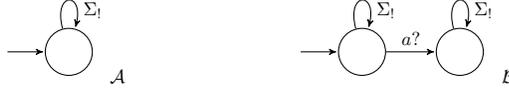
\begin{figure}[htb]
\scalebox{0.65}{
\begin{minipage}{0.45\textwidth}
\begin{center}
  \begin{tikzpicture}[->,>=stealth',shorten >=1pt,auto,node distance=2cm,
                    semithick]

   \tikzstyle{every state}=[text=black]

  \node[state, fill=white] (A) {};
  \node[state, fill=white, color=white, scale=.5] (G) [left of=A, node distance=3cm, yshift=0cm] {};
    \node[fill=white] (C) [right of=A, node distance=1cm, yshift=-.5cm] {$\A$};
 
  \path (A) edge [loop above] node [right, pos=.8] {$\Sigma_!$} (A)
        (G) edge node [above] {} (A)
;
\end{tikzpicture}
\end{center}
\end{minipage}
\begin{minipage}{0.45\textwidth}
\begin{center}
  \begin{tikzpicture}[->,>=stealth',shorten >=1pt,auto,node distance=2cm,
                    semithick]

   \tikzstyle{every state}=[text=black]

  \node[state, fill=white] (A) {};
  \node[state, fill=white] (B) [right of=A, node distance=2cm, yshift=0cm] {};
  \node[state, fill=white, color=white, scale=.5] (G) [left of=A, node distance=3cm, yshift=0cm] {};
    \node[fill=white] (C) [right of=B, node distance=1cm, yshift=-.5cm] {$\B$};
 
  \path (A) edge [loop above] node [right, pos=.8] {$\Sigma_!$} (A)
        (G) edge node [above] {} (A)
        (A) edge node [above] {$a?$} (B)
        (B) edge [loop above] node [right, pos=.8] {$\Sigma_!$} (B)
;
\end{tikzpicture}
\end{center}
\end{minipage}}
\caption{Counter-example to converse of Proposition~\ref{prop:tioco-alt-sim}.}
\label{fig-counter-ex}
\end{figure}

\paragraph{\bf Remark:} unfortunately, the converse of
Proposition~\ref{prop:tioco-alt-sim} is in general false, already in
the untimed case.  
This is illustrated 
in Figure~\ref{fig-counter-ex}.
It is clear that the automaton $\A$  accepts all implementations.
$\B$ also accepts all implementations as,
from the conformance point of view, when a
specification does not specify an input after a trace, this is
equivalent to specifying this input and then to accept the universal
language on $\Sigma_{obs}$.  
Thus
$\Imp \;\tioco\; \A \Rightarrow \Imp \;\tioco\; B$.
However 
$\neg (\A \preceq \B)$ as 
$in(\B \after \epsilon)=\{a\}$ 
but  $in(\A \after \epsilon)=\emptyset$.
Notice that this example also works for the untimed case in the $\ioco$ conformance theory.

As a corollary of Proposition~\ref{prop:tioco-alt-sim}, we get that
io-refinement preserves soundness of test suites:
\begin{cor}
\label{cor}
If $A \preceq B$ 
then any sound test suite for $\B$ is also sound for $\A$.
\end{cor}
%
\begin{proof}
Let $\TS$ be a  sound test suite for $\B$.
By definition, for any $\Imp \in \Imp(\B)$, for any $\TC \in \TS$, $\Imp \; {\tt fails} \; \TC \Ra \neg (\Imp \; \tioco \;\B)$.
As we have $\A \preceq \B$, by Proposition~\ref{prop:tioco-alt-sim}, 
we obtain $\neg (\Imp \; \tioco \;\B) \Ra \neg (\Imp \; \tioco \;\A)$
which implies 
that for any $\Imp \in \Imp(\B)$, for any $\TC \in \TS$, $\Imp \; {\tt fails} \; \TC \Ra \neg (\Imp \; \tioco \;\A)$.
Thus $\TS$ is also sound for $\A$.
\end{proof}
In the sequel, this corollary will justify 
our methodology:
from $\A$ a non-deterministic TAIO, 
build a deterministic io-abstraction $\B$ of $\A$,
then any test case generated from $\B$ and sound
is also sound for $\A$.

\section{Approximate determinization preserving conformance}
\label{sec-determinisation}

We recently proposed a game approach to determinize or provide a
deterministic over-approximation for TAs~\cite{BertrandStainerJeronKrichen-FOSSACS2011}.
Determinization is exact on all known classes of determinizable TAs
(\eg  ~event-clock TAs, TAs with integer resets, strongly non-Zeno TAs) if
resources (number and clocks and maximum constant) are sufficient.  
This
method can be adapted to the context of testing for building a
deterministic io-abstraction of a given TAIO. Thanks to
Proposition~\ref{prop:tioco-alt-sim}, the construction preserves
$\tioco$.

The approximate determinization uses the classical region\footnote{Note that it could be adapted to zones with some loss in precision.}
construction~\cite{AlurDill94}.  As for classical timed automata, the
regions form a partition of valuations over a given set of clocks
which allows to make abstractions in order to decide properties such as
the reachability of a
location.  
We note $\Reg_{(X,M)}$ the set of regions over clocks $X$ with maximal
constant $M$.  A region $r'$ is a {\em time-successor} of a region $r$
if $\exists v \in r, \exists t \in \setRnn, v+t \in r'$.  Given
  $X$ a set of clocks, a relation over $X$ is a finite conjunction $C$
  of atomic constraints of the form $x-y \sim c$ where $x, y \in X$,
  $\sim \in \{<,=,>\}$ and $c\in \setN$. When all constants $c$ belong
  to $[-M,M]$ for some constant $M \in \setN$ we denote by
  $\Rel{X}{}{M}$ for the set of relations over $X$.
Given a region $r$, we write $\overleftrightarrow{r}^{\petit{M}}$ for the smallest
relation in $\Rel{X}{}{M}$ containing $r$. 

\subsection{A game approach to determinize timed automata}
The technique presented
in~\cite{BertrandStainerJeronKrichen-FOSSACS2011} applies first to
TAs, \ie the alphabet only consists of one kind of actions (say output
actions), and the invariants are all trivial. Given such a TA
$\mathcal{A}$ over set of clocks $X^\pA$, a deterministic TA $\B$
with a new set of clocks $X^\pB$ is built, with $\traces(\A)=
\traces(\B)$ as often as possible, or $\traces(\A)\subseteq
\traces(\B)$.
Resources of $\B$ 
are fixed, and the goal is to simulate the clocks of $\A$ by choosing
the right resets in $\B$.  To this aim, letting $k=|X^\pB|$, a finite
2-player zero-sum turn-based safety game
$\mathcal{G}_{\pA,(k,M^\pB)}=(\V_S,\V_D,\v_0,\delta_S\sqcup\delta_D,\Bad)$
is built. 
The two players, Spoiler and Determinizator, alternate moves, the
objective of player Determinizator being to remain in a set of safe
states
where intuitively, for sure no over-approximation has been performed.
In this game, every strategy for Determinizator yields a deterministic
automaton $\B$ with $\traces(\A)\subseteq \traces(\B)$, and every
winning strategy induces a deterministic TA $\B$ equivalent to
$\A$. 
It is well known that for safety games, winning strategies can be
chosen positional (\emph{i.e.}, only based on the current state) and
computed in linear time in the size of the arena (see
\emph{e.g.}~\cite{mazala-lncs2500}).

The game $\mathcal{G}_{\pA,(k,M^\pB)}=(\V_S,\V_D,\v_0,\delta_S\sqcup\delta_D,\Bad)$
is defined as follows: 
\begin{iteMize}{$\bullet$}
\item $\V_S= 2^{L^\pA\times\Rel{X^\pA\sqcup
      X^\pB}{}{\max(M^\pA,M^\pB)}\times\{\bot,\top\}}\times
  \Reg_{(X^\pB,M^\pB)}$ is the set of states of Spoiler.  Each state
  is a pair $\v_S=(\mathcal{E}, r)$ where $r$ is a region over
  $X^\pB$, and $\mathcal{E}$ is a finite set of \emph{configurations}
  of the form $(\ell,C,b)$ where $\ell$ is a location of $\A$, $C$ is
  a relation over $X^\pA \sqcup X^\pB$ with respect to the maximal
  constant $M=\max(M^\pA,M^\pB)$, and $b$ is a boolean marker ($\top$ or
  $\bot$).  A state of Spoiler thus constitutes a state estimate of
  $\A$, and the role of the marker $b$ is to indicate whether
  over-approximations possibly happened.
\item $\V_D= \V_S \times (\Sigma \times \Reg_{(X^\pB,M^\pB)})$ is the
  set of states of Determinizator. Each state $\v_D=(\v_S,(a,r'))$
  consists of a state of Spoiler, together with an action and a region
  over $X^\pB$ which role is to remember the last move of Spoiler.
\item $\v_0=(\{(\ell_0,C_0,b_0)\},\{\overline{0}\})\in \V_S$, the
  initial state of the game, is a state of Spoiler consisting of a
  single configuration with
the initial location $\ell_0$ of $\A$, the simple relation $C_0$
over $X^\pA \sqcup X^\pB$: $\forall x, y \in X^\pA\sqcup X^\pB,\
x-y=0$, a marker $b_0=\top$ (no over-approximation was done so far),
together with the null region over $X^{\pB}$.
\item $\delta_S \subseteq \V_S \times (\Sigma \times
  \Reg_{(X^\pB,M^\pB)}) \times \V_D$ and $\delta_D \subseteq \V_D \times
  2^{X^\pB} \times \V_S$ are inductively defined from $\v_0$ as follows:
\begin{iteMize}{$-$}
\item moves of Spoiler are pairs $(a,r')$ and the successor of a
  state $\v_S=(\mathcal{E},r) \in \V_S$ by the move $(a,r')$ is simply
  $\v_D=((\mathcal{E},r),(a,r'))$, \ie a copy of $\v_S$ together with a
  challenge for Determinizator consisting in an action $a$ and a
  region $r'\in \Reg_{(X^\pB,M^\pB)}$, a time-successor of $r$;
\item moves of Determinizator are resets $Y \subseteq X^\pB$ and
  the successor of a state $\v_D=((\mathcal{E},r),(a,r')) \in \V_D$ by the
  reset $Y \subseteq X^{\pB}$, is the state of Spoiler
  $(\mathcal{E}',r'_{[Y \leftarrow 0]})\in \V_S$ where $\mathcal{E}'=
  \{\Succ_e[(a,r'),Y](\ell,C,b) \mid (\ell,C,b)\in \mathcal{E}\}$ and
$$
\Succ_e[(a,r'),Y](\ell,C,b) =\left\{(\ell',C',b') \, \left|\, \begin{array}{l} \exists \ell
  \xrightarrow{g,a,X} \ell' \in E \textrm{ s.t. } [r' \cap C]_{|X^\pA}
  \cap g \neq \emptyset \\ C' = \overleftrightarrow{(r' \cap C \cap g)_{[X
    \leftarrow 0][Y \leftarrow 0]}}^{\petit{M}}\\ b' = b \wedge ([r' \cap
  C]_{|X^\pA} \subseteq g) \end{array} \right\}\right..
$$
In words, $\mathcal{E}'$ is the set of elementary successors of
configurations in $\mathcal{E}$ by $(a,r')$ and by resetting $Y$.  An
elementary successor of a configuration $(\ell,C, b)$ by a transition
$\ell \xrightarrow{g,a,X} \ell'$ exists only if the guard $[r' \cap C]
_{|X^\pA}$ over $X^\pA$ induced by the guard $r'$ over $X^\pB$ through
the relation $C$ intersects $g$. 
Intuitively, the transition is possible in $\ell$ according to the
state estimate $(\ell,C)$ and the region $r'$. 
The resulting configuration $(\ell',C',b')$ is such that:
\begin{iteMize}{$*$}
\item $\ell'$ is
the location reached by the transition; 
\item $C'$ is the relation between
clocks in $X^\pA$ and $X^\pB$ after the moves of the two players, that
is after satisfying the guard $g$ in $r'\cap C$, resetting $X
\subseteq X^\pA$ and $Y\subseteq X^\pB$; 
\item $b'$ is a boolean set to
$\top$ if both $b = \top$ and the induced guard $[r' \cap C]_{|X^\pA}$
over $X^\pA$ implies $g$.  
Intuitively, $b'$ becomes $\bot$ when
  $r'$ encodes more values than $g$, thus an over-approximation
  possibly happens.
\end{iteMize}
\end{iteMize}
Note that during the construction of $\delta_S$ and $\delta_D$, the
states of Determinizator whose successors by $\delta_D$ have an empty
set of configurations are removed, together with the moves in
$\delta_S$ leading to them.  Indeed these moves have no counterpart in
$\A$.
\item $\Bad = \{(\mathcal{E},r)\in \V_S \mid \forall (\ell,C,b) \in \mathcal{E}, b= \bot\}$. Bad states Determinizator wants to avoid are states where
all configurations are marked $\bot$, \ie configurations where
an approximation possibly happened.
Note that a single configuration marked $\top$ in a state is enough to
ensure that no over-approximation happened. Indeed, for any path in
the game leading to such a state, starting from a $\top$-marked
configuration, and taking elementary predecessors, one can build
backwards a sequence of configurations following this path. By
definition of the marker's update, these configurations are all marked
$\top$, and the sequence thus corresponds to real traces in the
non-deterministic automaton.
\end{iteMize}

\begin{exa} Figure~\ref{fig:ex-ta} represents a simple non-deterministic timed automaton $\A$. 
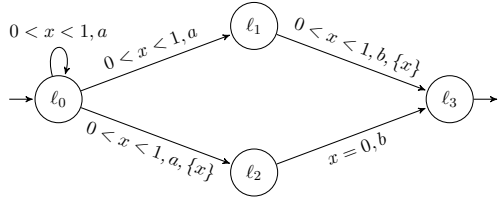
\begin{figure}[htbp]
\begin{center}
\scalebox{0.65}{
\begin{tikzpicture}[->,>=stealth',shorten >=1pt,auto,node distance=2cm,
                    semithick]

  \tikzstyle{every state}=[text=black]

  \node[state, fill=white] (A) {$\ell_0$};
  \node[state, fill=white] (B) [right of=A, node distance=4cm, yshift=1.5cm] {$\ell_1$};
  \node[state, fill=white] (C) [right of=A, node distance=4cm, yshift=-1.5cm] {$\ell_2$};
  \node[state, fill=white] (D) [right of=A, node distance=8cm] {$\ell_3$};
  \node[state, fill=white, color=white] (E) [left of=A, node distance=1.5cm, yshift=0cm] {};
\node[state, fill=white, color=white] (F) [right of=D, node distance=1.5cm, yshift=0cm] {}
;
  
  \path (A) edge node [above,sloped] {$0<x<1,a$} (B)
        (E) edge (A)
        (A) edge [loop above] node [above] {$0<x<1,a$} (A)
        (A) edge node [below,sloped] {$0<x<1,a,\{x\}$} (C)
	(B) edge node [above,sloped] {$0<x<1,b,\{x\}$} (D)
	(C) edge node [below,sloped] {$x=0,b$} (D)      
        (D) edge (F)
;
\end{tikzpicture}}
\caption{Non-deterministic timed automaton $\A$.  }\label{fig:ex-ta}
\end{center}
\end{figure}
Let us explain how to construct the game
$\mathcal{G}_{\mathcal{A},(1,1)}$ for $\A$ with resources $(1,1)$,
that is a single clock $y$ and maximal constant $1$.  We only detail
part of the construction in Figure~\ref{fig:ex-game}, but the complete
game can be found in~\cite{BertrandStainerJeronKrichen-FOSSACS2011}.
 \begin{figure}[htb]
\begin{center}
\scalebox{.65}{
\begin{tikzpicture}[->,>=stealth',shorten >=1pt,auto,node distance=2cm,
                     semithick]

   \tikzstyle{every state}=[text=black]

   \node[state, fill=white] [rectangle] (A) {$\begin{array}{ll} 
\ell_0,x-y=0,\top&\{0\}\end{array}$};
   \node[state, fill=white] (A1) [right of=A, node distance=5cm,  
yshift=0cm] {};
   \node[rectangle] (B) [below of=A1, node  
distance=2cm, xshift=0cm] {};
    \node[state, fill=white] [rectangle] (C)  [right of=A1, node  
distance=3.9cm, xshift=0cm] {$\begin{array}{ll}\ell_0,0<x-y<1,\top &  
\multirow{3}{*}{\{\emph{0}\}} \\\ell_1,0<x-y<1,\top& \\\ell_2,x-y=0,\top&  
\end{array}$};
   \node[state, fill=white] (C1)  [right of=C, node distance=5.2cm,  
yshift=0cm] {};
   \node[state, fill=GrisC!90!black] [rectangle] (C11) [right of=C1, node  
distance=4cm, xshift=0cm] {$\begin{array}{ll}\ell_0,0<x-y<1,\bot&  
\multirow{3}{*}{\emph{(0,1)}}\\\ell_1,0<x-y<1,\bot&\\\ell_2,-1<x-y<0,\bot 
\end{array}$};
   \node[rectangle] (C12) [below of=C1, node  
distance=2cm, xshift=0cm] {};
   \node (C2) [below of=C, node distance=2cm,  
xshift=0cm] {};
   \node (C111) [below of=C11, node distance=2cm,  
xshift=0cm] {};
   \node[fill=white, color=white] (H) [above of=A, node  
distance=1.5cm] {};

   \path (H) edge node [left] {$\{y\}$} (A)
         (A) edge node [above] {$0<y<1,a$} (A1)
         (C) edge node [above] {$0<y<1,a$} (C1)
         (C) edge [dashed] node [left] {} (C2)            
         (A1) edge node [above] {$\{y\}$} (C)
        (A1) edge [dashed] node [left] {$\emptyset$} (B)
         (C1) edge node [above] {$\emptyset$} (C11)
         (C11) edge [dashed] node [left] {} (C111)
        (C1) edge [dashed] node [left] {$\{y\}$}  (C12)
;

\end{tikzpicture}}
\caption{Part of the game $\mathcal{G}_{\mathcal{A},(1,1)}$.}
\label{fig:ex-game}
\end{center}
\end{figure}
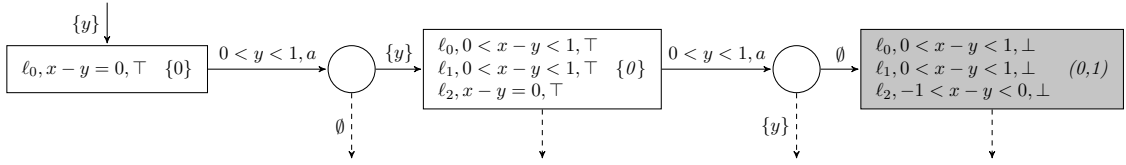

As defined above, the initial state of the game is simply
$\v_0=(\{(\ell_0,x-y=0,\top)\},\{0\})$. 

From $\v_0$, the only move of Spoiler compatible with behaviors of
$\A$ is $0<y<1,a$.
Corresponding transitions in $\A$  lead to
locations $\ell_0$, $\ell_1$ and $\ell_2$, and only in this last
location $x$ has been reset.  Each transition of $\A$ yields a
configuration in the next state of Spoiler, and assuming
Determinizator chooses to reset $y$, the three different
configurations are the following:
 \begin{iteMize}{$\bullet$}
 \item one with location $\ell_0$, where $x\in(0,1)$ (no reset in $\A$) and $y=0$ (reset in $\G{\A}{1,1}$),
 \item one with location $\ell_1$, where $x\in(0,1)$ and $y=0$,
 \item and one with location $\ell_2$, where $x=0$ (reset in $\A$) and
   $y=0$.
 \end{iteMize} 
 In the two first configurations, the new relation is
 $\overleftrightarrow{(y=0<x<1)}^1$, that is $0 <x-y <1$, and in the
 last configuration, the new relation is simply $x -y =0$. As a
 consequence the successor state is $\v_1 = (\{(\ell_0,0<x-y<1,\top),
 (\ell_1,0<x-y<1,\top), (\ell_2,x-y=0,\top)\} ,\{0\})$. Note that all
 markers are $\top$ since the guard on $y$ faithfully represented the
 ones on $x$. 

 From state $\v_1$, if Spoiler chooses the move $0<y<1, a$, it is not
 obvious to which transitions in $\A$ this corresponds, and we thus
 explain in details how to compute the successor state. 
 First observe that the only configuration in $\v_1$ from which an $a$
 action is possible is the first one, with location $\ell_0$. In this
 configuration, the relation is $0<x-y<1$.  Let us now explain what
 guard over $x$ is induced by the relation $C = 0<x-y<1$ and the region
 $r' = 0<y<1$.
\begin{figure}[htbp]
\begin{center}
\scalebox{.65}{
\begin{tikzpicture}[scale=2]
\draw [fill=blue!30,color=blue!30]
(0,1) -- (2.9,1) -- (2.9,0) -- (0,0) -- (0,1);

  \draw (-0.3,.5) node[right] {\Large{$r'$}};
  
\draw [fill=yellow!30,color=yellow!30]
     (1,1) -- (2.9,2.9) -- (2.9,1.9) -- (2,1) -- (1,1);

\draw[fill=green!30,color=green!30]
(0,0) -- (1,1) -- (2,1) -- (1,0) -- (0,0);

\begin{scope}[opacity=.6]

\draw [pattern=dots,pattern color =blue]
(0,1) -- (2.9,1) -- (2.9,0) -- (0,0) -- (0,1);
\draw [pattern=north east lines, pattern color=orange]
      (0,0) -- (2.9,2.9) -- (2.9,1.9) -- (1,0) -- (0,0);

\end{scope}
  \draw (2.2,1.7) node[right] {\Large{$C$}};
     
  \draw (.7,.5) node[right] {\Large{$r' \cap C$}};
 \draw[style=help lines] (0,0) grid (2.9,2.9);

  \draw[->] (-0.2,0) -- (3,0) node[right] {\Large{$x$}};
  \draw[->] (0,-0.2) -- (0,3) node[above] {\Large{$y$}};

  \foreach \x/\xtext in {1/1, 2/2}
    \draw[shift={(\x,0)}] (0pt,2pt) -- (0pt,-2pt) node[below] {$\xtext$};

  \foreach \y/\ytext in {1/1, 2/2}
    \draw[shift={(0,\y)}] (2pt,0pt) -- (-2pt,0pt) node[left] {$\ytext$};
\draw[line width= 4pt,color=green!75!black] (0,0) -- (2,0);
  \draw (0.5,-.4) node[right] {\Large{$[r'\cap C]_{|\{x\}}$}};
  \draw[line width= 4pt,color=blue] (0,0) -- (0,1);
\end{tikzpicture}}\caption{Construction of the induced guard.
}\label{fig:ind-guard}
\end{center}
\end{figure}
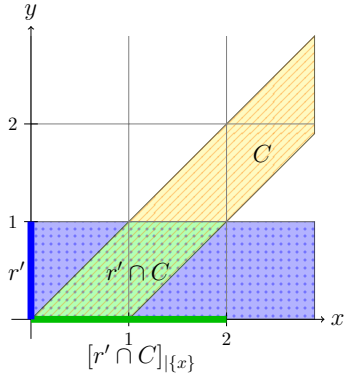
Figure~\ref{fig:ind-guard} illustrates this computation.  The dotted
area represents the set of the valuations over $\{x,y\}$ satisfying
the guard $r'=0<y<1$ and the dashed area represents the relation
$C=0<x-y<1$. The induced guard $[r'\cap C]_{|\{x\}}$ (\ie the guard
over $x$ encoded by the guard $r'$ on $y$ through the relation $C$) is
then the projection over clock $x$ of the intersection of these two
areas. In this example, the induced guard is $0<x<2$. Therefore, the
transitions of $\A$ corresponding to the choice of Spoiler $0<y<1, a$
are as before the three ones originating in $\ell_0$, but this time they
are over-approximated. Indeed, the induced guard $[r'\cap C]_{|\{x\}}$
is not included in the original guard $0<x<1$ in $\A$, \ie~{\em a priori}
$r'$ encodes more values than $g$.  As a consequence, all the
configurations in Spoiler's successor state are marked $\bot$.
Last, let us detail how the new relations are computed.
Assuming Determinizator chooses not to reset $y$ leads to state $\v_2$,
in which for the configuration with location $\ell_0$, the relation is
the smallest one containing $(0<x-y<1)\cap(0<y<1)\cap(0<x<1)$, namely
$0<x-y<1$. The relation for the last configuration in $\v_2$ is
$\overleftrightarrow{\bigl((0<x-y<1)\cap (0<y<1) \cap
  (0<x<1)\bigr)_{[x \leftarrow 0]}}^1$, which is same as
$\overleftrightarrow{(x=0 <y<1)}^1$, namely $-1 < x-y <0 $.

\end{exa}

As explained earlier, a strategy for Determinizator chooses in each
state of $\V_D$ a set $Y \subseteq X^\pB$ of clocks to reset.  With
every strategy $\Pi$ for Determinizator we associate the TA
$\B=\mathsf{Aut}(\Pi)$ obtained by merging a transition of Spoiler
with the transition chosen by Determinizator just after.  The
following theorem links strategies of Determinizator with
deterministic over-approximations of the original traces language and
enlightens the interest of the game:
\begin{thm}[\cite{BertrandStainerJeronKrichen-FOSSACS2011}]
\label{th:game}
Let $\A$ be a TA, and $k,M^\pB \in \setN$. For any
strategy $\Pi$ of Determinizator in
$\mathcal{G}_{\pA,(k,M^\pB)}$, $\B=\mathsf{Aut}(\Pi)$ is a
deterministic timed automaton over resources $(k,M^\pB)$ and satisfies
$\traces(\A)\subseteq \traces(\B)$. 
Moreover, if $\Pi$ is winning, then
$\traces(\A) =
\traces(\B)$.
\end{thm}
When there is no winning strategy, one can either try to increase
  resources (number of clocks and/or maximal constant), or try to
  choose the best losing strategy, which is a concern.  Indeed, the
  language inclusion seems to be a good criterion to compare two
  losing strategies, but it is not a total ordering.  Alternatively,
  one can use the natural heuristics which tends to lose as late as
  possible (see~\cite{BertrandStainerJeronKrichen-RR2010}).  In
  particular, for a game with $k$ clocks and same maximal constant as
  the original timed automaton, there is a strategy which ensures not
  to lose before $k$ moves (of each players): by choosing to reset a
  new clock at each of its moves, Determinizator ensures 
to perfectly  encode all clocks of the original timed automaton. 
 Other
  alternatives would be to consider heuristics based on quantitative
  measures over languages.

\subsection{Extensions to TAIOs and adaptation to \tioco}
In the context of model-based testing, the above-mentioned
determinization technique must be adapted to TAIOs, as detailed
in~\cite{BertrandStainerJeronKrichen-RR2010}, and summarized below.  
The model of TAIOs
is an expressive model of timed automata incorporating internal actions and
invariants.  
Moreover, inputs and outputs must be treated differently
in order to build
from a TAIO $\A$ a DTAIO $\B$ such that $\A \preceq \B$,
and then to preserve \tioco.

\ignore{
In the context of model-based testing, the alphabet of actions is
partitioned into inputs and outputs \NB{(and internal actions)}. In
order to construct from a specification an approximation while 
preserving conformance, input and output actions have to be treated
differently: successors by outputs need to be over-approximated to
allow for more behaviors whereas successors by inputs must be
under-approximated to specify less behaviors (\NB{clarifier
  encore}). 
Moreover, it is necessary to be able to deal with models incorporating
internal actions and invariants, since they increase the expressive
power of the specifications by allowing silent transitions and notions
of urgency.

The game approach can be adapted to deal with these issues and the
extensions, detailed in~\cite{BertrandStainerJeronKrichen-RR2010}, are summarized below.  }

\begin{iteMize}{$\bullet$}
\item
  \textit{Internal actions} are naturally part of the specification
  model.  They cannot be observed during test executions and should
  thus be removed during determinization.  In order to do so, a
  closure by internal actions is performed for each state during the
  construction of the game, 
  that is, in each state, all the configurations reachable by internal
  actions are added to the set of configurations. 
To this attempt, states of the game have to be extended since internal 
actions might be enabled 
from a subset of time-successors 
of the region associated with the state.  Therefore, each
configuration is associated with a proper region which is a
time-successor of the initial region of the state.
The closure by internal actions is effectively computed the same way
as successors in the original construction when Determinizator is
  not allowed to reset any clock.
  It is well known that timed automata with silent transitions are
  strictly more expressive than standard timed
  automata~\cite{BerardGastinPetit-STACS96}.  Therefore, our
  approximation can be coarse, but it performs as well as possible
  with its available clock information.
\item
  \textit{Invariants} are classically used to model urgency in timed
  systems.
  Taking into account urgency of outputs is quite important, indeed
  without the ability to express it, for instance, any 
  dummy system would conform to all specifications. Ignoring all
  invariants in the approximation as done in~\cite{KrichenTripakis09}
  surely yields an io-abstraction: delays (considered as outputs) are
  over-approximated.  In order to be more precise, while preserving
  the io-abstraction relation $\preceq$, with each state of the game
  is associated the most restrictive invariant containing invariants
  of all the configurations in the state.  In the computation of the
  successors, invariants are treated as guards and their validity is
  verified at both ends of the transition.  A state whose
  invariant is strictly over-approximated is treated as unsafe in the
  game.
\item
Rather than over-approximating a given TAIO $\A$, we aim here at
building a DTAIO $\B$ \textit{io-abstracting} $\A$ ($\A \preceq
\B$).
Successors by outputs are over-approximated as in the original game,
while successors by inputs must be under-approximated. The over-approximated closure by silent transitions is not suitable to under-approximation.
Therefore, states of the game are extended to contain both over-approximated and under-approximated closures.
Thus, the
unsafe successors by an input (where possibly an over-approximation
would occur), are not built.
\end{iteMize}
\begin{exa}
\label{exa3.3}
Figure~\ref{ExProddet} represents a non-deterministic
  timed automaton $\A'$ that has invariants and internal actions. It
  is a sub-automaton of the timed automaton we use in the next section
  (see Figure~\ref{ExProd}) to illustrate the approximate
  determinization for our test selection.
\begin{figure}[htbp]
\begin{center}
\scalebox{0.65}{
\begin{tikzpicture}[->,>=stealth',shorten >=1pt,auto,node distance=2cm,
                    semithick]

  \tikzstyle{every state}=[text=black]

  \node[state, fill=white] (A) {$\ell_0$};
  \node[state, fill=white] (B) [right of=A, node distance=3.5cm, yshift=1cm] {$\ell_1$};
  \node[state, fill=white] (C) [right of=B, node distance=4.5cm, yshift=0cm] {$\ell_2$};
  \node[state, fill=white] (F) [right of=A, node distance=3.5cm, yshift=-1cm] {$\ell_5$};
  \node[state, fill=white] (G) [right of=F, node distance=4.5cm, yshift=0cm] {$\ell_6$};
  \node[state, fill=white, color=white] (A') [left of=A, node distance=2cm, yshift=0cm] {};
  \node[fill=white] (Ae) [below of=A, node distance=.9cm, yshift=0cm] {$x \le 1$};
  \node[fill=white] (Ce) [below of=C, node distance=.9cm, yshift=0cm] {$x \le 1$};
  \node[fill=white] (Ge) [below of=G, node distance=.85cm, yshift=0cm] {$x = 0$};
  
  \path (A') edge node [above] {} (A)
        (A) edge node [above,sloped] {$x=1, \tau$} (B)    
        (B) edge node [above] {$1<x<2, a?, \{x\}$} (C)   
        (C) edge [loop above] node [above] {$x=1, \tau, \{x\}$} (C)    
        (A) edge node [above,sloped] {$x=1, \tau, \{x\}$} (F)    
        (F) edge node [above] {$x<1, a?, \{x\}$} (G)   
;

\end{tikzpicture}}
\caption{Non-deterministic timed automaton $\A'$ (with invariants and
  internal actions).  }\label{ExProddet}
\end{center}
\end{figure}
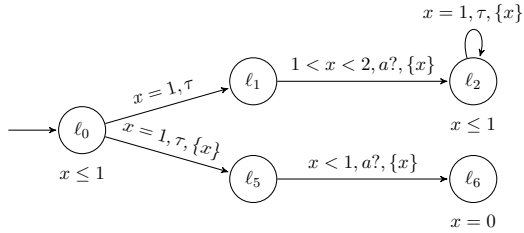

Using this automaton $\A'$, let us illustrate how the game
construction is adapted to deal with internal actions and invariants,
by detailing part of the game $\mathcal{G}_{\A',(1,2)}$ represented in
Figure~\ref{ExJeudet}.
\begin{figure}[htbp]
\begin{center}
\scalebox{0.7}{
\begin{tikzpicture}[->,>=stealth',shorten >=1pt,auto,node distance=2cm,
                    semithick]

  \tikzstyle{every state}=[text=black]

  \node[state, fill=white] [rectangle] (A)
  {$\begin{array}{ll}(\ell_0,x-y=0,\top)& \{0\}\\ 
 (\ell_1,x-y=0,\top)& \{1\}\\ (\ell_5,x-y=-1,\top) &\{1\} \end{array}$};
  \node[state, fill=white] (AB) [right of=A, node distance=4.8cm, yshift=0cm] {};
  \node[state, fill=white] [rectangle] (B1) [right of=AB, node distance=3.8cm, yshift=1cm] {$(\ell_6,x-y=0,\top)\; \{0\}$};
  \node[state, fill=white] (AC) [right of=A, node distance=4.8cm, yshift=-2cm] {};
  \node[fill=white] (C2) [below of=AC, node distance=2cm, yshift=0cm] {};
  \node[state, fill=GrisC!90!black,inner xsep=0] [rectangle] (C1) [right of=AC, node distance=4.2cm, yshift=0cm] {$\begin{array}{ll}(\ell_6,x-y=0,\top)& \{0\}\\
(\ell_2,x-y=0,\top)& \{0\} \\ 
(\ell_2,x-y=-1,\top)& \{1\}\\ 
 (\ell_2,x-y=-2,\top)& \{2\} \\ \hline \hline 
(\ell_2,x-y<-2, \bot) & (2,\infty)\end{array}$};
  \node[state, fill=white, color=white] (A') [left of=A, node distance=3.7cm, yshift=0cm] {};
    \node[fill=white] (Ae) [below of=A, node distance=1.2cm, yshift=0cm] {$\true,\top$};
   \node[fill=white] (C1e) [below of=C1, node distance=1.7cm, xshift=0cm] {$\true,\bot$};
  \node[fill=white] (B1e) [below of=B1, node distance=.9cm, yshift=0cm] {$y = 0,\top$};
  \node[rectangle] (B2) [above of=AB, node distance=2cm] {};
  
    \path (A') edge node [above] {} (A)
        (A) edge node [above] {$y=1, a?$} (AB)    
        (AB) edge node [above,pos=.3] {$\{y\}$} (B1)   
        (AB) edge [dashed] node [right] {$\emptyset$} (B2)   
        (A) edge node [below, sloped] {$1<y<2,a?$} (AC)   
        (AC) edge node [above] {$\{y\}$} (C1)   
        (AC) edge [dashed] node [right] {$\emptyset$} (C2)    
;

\end{tikzpicture}}
\caption{Part of the game $\mathcal{G}_{\A',(1,2)}$.
}\label{ExJeudet}
\end{center}
\end{figure}
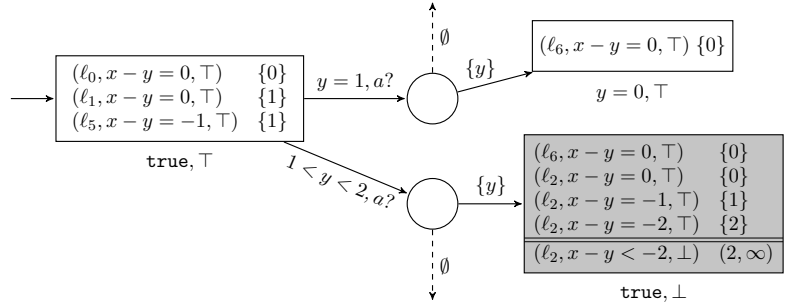

A state of Spoiler in the game is a triple $(S_-,S_+,(I,b_I))$ where
$S_-$ (resp. $S_+$) is the 
under-approximated (resp. over-approximated) closure by unobservable
actions of the successors by some observable action, 
$I$ is the invariant and $b_I$ is the marker which indicates a risk of
approximation of the invariant. The invariant and the marker of
Spoiler's states are written below the states.

In the initial state of the game, $(\ell_0,x-y=0,\top,\{0\})\in S_-
\subseteq S_+$. Moreover, an internal action $\tau$ can be fired for
$x=1$ along two different edges, which add two configurations,
associated with the region $y=1$ (because $x-y=0$ in the first
configuration). Determinizator cannot reset $y$ along an internal
action, hence the relation for the configuration with location
$\ell_5$ is $x-y=-1$. Note that the region $y=1$ is associated with
the two last configurations in the initial state, reflecting that the
internal action fired and thus the least value for $y$ is $1$.  Also
in this case, the closure (by internal actions) is not approximated,
hence $S_-=S_+$.  
On the other hand, it may be surprising that the invariant of this
initial state is $\true$ whereas the invariant of the initial state of
$\A'$ is $x\le1$.  In fact, the invariant of a state is the smallest
invariant containing the union, over all its configurations, of induced
invariants.  On this example, after an internal action from $\ell_0$,
delays are not constrained anymore in $\ell_1$ and $\ell_6$
(invariants are $\true$).  Thus the invariant in the initial state of
the game is not approximated, so its marker is $\top$.

From this initial state, Spoiler can choose the regions $y=1$ or
$1<y<2$ together with action $a?$. For $y=1$, this can only happen
from the configuration with location $\ell_5$. Indeed, the relation $x
-y =0$ and the guard $y=1$ induce a guard $x=1$ which is not
compatible with the outgoing edge from $\ell_1$ in $\A'$. The
computation of the successor state, \eg~when Determinizator chooses to
reset $y$, is simple: no internal action is fireable and the
invariant in $\ell_6$ is precisely expressed by $y=0$. The situation
is more complex when Spoiler chooses the region $1<y<2$: in this case
there are two successors by the observable action $a?$ (leading to
locations $\ell_6$ and $\ell_2$), and for the first one internal actions
may follow. We thus have to compute the closure by internal actions of
the successor configuration by observable action $a?$. Before
computing the closure, and assuming that Determinizator resets clock
$y$, the successor state is composed of two configurations:
$(\ell_2,x-y =0,\top)$ and $(\ell_6,x-y=0,\top)$ together with region
$y=0$. Along the $\tau$-loop on location $\ell_2$, $x$ is reset in
$\A'$ whereas $y$ cannot be reset in the game (because it is an
internal action). Starting from configuration $(\ell_2,x-y
=0,\top,\{0\})$ and performing once the internal action $\tau$, the
resulting configuration is thus $(\ell_2,x-y =-1,\top,\{1\})$. This
computation is iterated to obtain the closure by internal actions,
which in such a case, will depend on the maximal constant (here
$2$). Indeed, after $(\ell_2,x-y =-1,\top,\{1\})$, the next
configuration is $(\ell_2,x-y=-2,\top,\{2\})$ and starting from
$(\ell_2,x-y=-2,\top,\{2\})$ the effect of one internal action would
yield to $(\ell_2,x-y=-3,\top,\{3\})$. However, $x-y =-3$ cannot be
expressed in $\Rel{\{x,y\}}{}{2}$, so it is approximated by the least
relation of $\Rel{\{x,y\}}{}{2}$ containing it, that is
$x-y<-2$. Similarly, region $y=3$ is approximated by $y>2$. As a
consequence, the configuration $(\ell_2,x-y=-3,\top,\{3\})$ is
approximated by $(\ell_2,x-y<-2,\bot,(2,\infty))$ in $S_+$. Note that
this latter configuration is in $S_+ \setminus S_-$ and thus separated
from configurations in $S_-$ by two horizontal lines on
Figure~\ref{ExJeudet}.
Moreover, taking the union of all the invariants, we obtain $\true$ as
invariant for this state, but since it is approximated for the last
configuration $(\ell_2,x-y<-2,\bot,(2,\infty))$, its marker is $\bot$.
\end{exa}

All in all, these modifications allow to deal with the full TAIO model
with invariants, internal transitions and inputs/outputs.  In
particular, the treatment of invariants is consistent with the
io-abstraction: delays are considered as outputs, thus
over-approximated.
Figure~\ref{ExJeu} represents a part of this game for the TAIO of Figure~\ref{ExProd}.
The new game then enjoys the following nice
property:
\begin{prop}[\cite{BertrandStainerJeronKrichen-RR2010}]
\label{altSim}
  Let $\A$ be a TAIO, and $k,M^\pB \in \setN$. 
  For any
  strategy $\Pi$ of Determinizator in the game
  $\mathcal{G}_{\pA,(k,M^\pB)}$, $\B=\mathsf{Aut}(\Pi)$ is a
  DTAIO over resources $(k,M^\pB)$ with $\A \preceq \B$. 
  Moreover, if $\Pi$ is winning, then
  $\traces(\A) =  \traces(\B)$.
\end{prop}

In other words, the approximations produced by our method are
deterministic io-abstractions of the initial specification, hence the
approximate determinization preserves
\tioco~(Proposition~\ref{prop:tioco-alt-sim}), and conversely,
sound test cases of the approximate determinization remain sound for the original specification (Corollary~\ref{cor}).  Note that the proof of
proposition~\ref{altSim} in~\cite{BertrandStainerJeronKrichen-RR2010}
considers a stronger refinement relation, thus implies the same result
for the present refinement relation. In comparison with our method,
the algorithm proposed in~\cite{KrichenTripakis09} always performs an
over-approximation, and thus preserves \tioco~only if the
specification is input-complete; moreover all invariants are
set to $\true$ in the resulting automata, so the construction does
not preserve urgency.

\ignore{
\subsection{Properties}

\TJ{A revoir et reduire}
The algorithm proposed in~\cite{KrichenTripakis09} is an over-approximation,
thus preserves \tioco\; only if the specification is input-complete. 
Moreover it does not preserve urgency. 
On the contrary, our method always preserves \tioco, 
and incorporates urgency as much as possible.



\paragraph{Urgency and internal actions}
Concerning expressivity, modeling urgency is quite important. Without
the ability to express urgency, for instance, any system doing nothing
would  conform to any specification. It
is classical to use invariants to model urgency. Our determinization
method allows to preserve them as much as possible. Moreover, in
practice, not all actions are observable. Specifications naturally can
include internal actions but they cannot be observed during test
executions. In the determinization step, these internal actions are
seen as silent transitions.  
It is well known that timed automata with silent transitions are
strictly more expressive than standard timed automata~\cite{BerardGastinPetit-STACS96}. 
Therefore, of course, our approximation can be coarse, but it
performs as well as possible with its available clock information.
%
%
\paragraph{Conformance}
For test generation, the most important property is the soundness of
the test cases. As a consequence, when generating tests from an
approximation, the approximation should preserve soundness. In other
words, approximation must preserve {\tioco}.  The approximations
produced by our method are deterministic abstractions of the initial
specification, hence, our approach preserves conformance.
%
%
%
Approximations can be more or less precise, information contained in
the game states allows to establish finer verdicts, ensuring in some
cases the strictness of a test.
}

\subsection*{Complexity} The number of regions (resp. relations)
  over a set of clocks is exponential in the number of clocks. Thus,
  the number of possible configurations in the game is at most exponential in the
  cardinality of $X\sqcup Y$ and linear in the number of locations in
  $\A$. As a consequence, the size of the game (\emph{i.e.}, number of
  states in the arena) is at most doubly exponential in $|X\sqcup Y|$
  and exponential in $|L^\pA|$. In particular this bound also holds
  for the size of the generated deterministic TAIO, for every
  \emph{memoryless} strategy of Determinizator. The overall complexity
  of this io-abstracting determinization algorithm is thus doubly
  exponential in the size of the instance (original TAIO and
  resources).

\section{Off-line test case generation}
\label{sec-generation}
In this section, we describe the off-line generation of test cases from 
timed automata specifications and test purposes.
We first define test purposes, their role in test generation 
and their formalization as OTAIOs.
We then detail the process of
off-line test selection guided by test purposes, which uses 
the approximate determinization just defined.
We also prove properties of generated test cases with respect 
to conformance and test purposes.

\subsection{Test purposes}

In testing practice, especially when test cases are
generated manually, each test case has a particular objective,
informally described by a sentence called test purpose.
In formal test generation, test purposes should be formal models 
interpreted as means to select behaviors to be tested,
either focusing on usual behaviors, or on suspected errors in
implementations~\cite{jard04a}, thus typically reachability properties.  
They complement other selection mechanisms 
such as coverage methods~\cite{Zhu-Hall-May-97} which, contrary to test purposes,
are most often based on syntactical criteria rather than semantic aspects.
Moreover, the set of goals
covering a given criterion  (\eg~states, transitions, etc) 
may be translated into a set of test purposes,
each test purpose focusing on one such goal.

As test purposes are selectors of behaviors, 
a natural way to formalize them is to use a logical formula
characterizing a set of behaviors 
or an automaton accepting those behaviors.
In this work we choose to describe test purposes as OTAIOs equipped 
with accepting states.
The motivation is to use a model close to the specification model,
easing the description of targeted specification behaviors.
The following definition formalizes test purposes, 
and some alternatives are discussed in Section~\ref{sec-discuss}. 

\begin{defi}[Test purpose]
 Let  $\A=(L^\pA,\ell_0^\pA,\Sigma_?^{\pA},\Sigma_!^{\pA},\Sigma_\tau ^{\pA},
  X_p^\pA,\emptyset, M^\pA,\Inv^\pA,E^\pA)$ be a TAIO specification.
 A {\em test purpose} for $\A$ is a pair
  $(\TP,\accept^\tp)$ where:
\begin{iteMize}{$\bullet$}
\item $\TP=
  (L^\tp,\ell_0^\tp,\Sigma_?^{\pA},\Sigma_!^{\pA},\Sigma_\tau ^{\pA},
  X_p^\tp,X_o^\tp,M^\tp,\Inv^\tp,E^\tp)$ is a complete OTAIO (in
  particular $\Inv^\tp(\ell) = \true$ for any $\ell \in L^\tp$) with
  $X_o^\tp=X_p^\pA$ ($\TP$ observes proper clocks of $\A$) and
  $X_p^\tp\cap X_p^\pA = \emptyset$,
\item
  $\accept^\tp \subseteq L^\tp$ is a subset of trap 
  locations.
\end{iteMize}
\end{defi}

In the following, we will sometimes abuse notations and use $\TP$
instead of the pair $(\TP, \accept^\tp)$.  During the test generation
process, test purposes are synchronized with the specification, and
together with their $\accept$ locations, they will play the role of
acceptors of timed behaviors.  They are non-intrusive in order not to
constrain behaviors of the specification.  This explains why they are
complete, thus allowing all actions in all locations, and are not
constrained by invariants.  They observe behaviors of specifications
by synchronizing with their actions (inputs, outputs and internal
actions) and their proper clocks (by the definition of 
the product (Definition~\ref{def_product}), observed clocks of
$\TP$ are proper clocks of $\A$, which mean that $\TP$ does not reset
those clocks).  However, in order to add some flexibility in the
description of timed behaviors, they may have their own proper clocks.

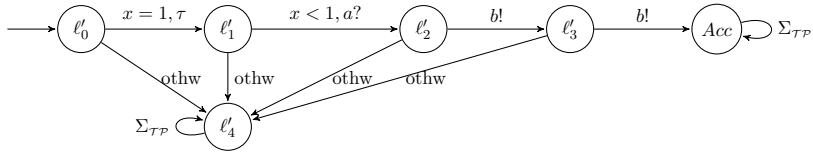
\begin{figure}[htbp]
\begin{center}
\scalebox{0.65}{
\begin{tikzpicture}[->,>=stealth',shorten >=1pt,auto,node distance=2cm,
                    semithick]

  \tikzstyle{every state}=[text=black]

  \node[state, fill=white] (A) {$\ell'_0$};
  \node[state, fill=white] (B) [right of=A, node distance=3cm, yshift=0cm] {$\ell'_1$};
  \node[state, fill=white] (C) [right of=B, node distance=4cm, yshift=0cm] {$\ell'_2$};
  \node[state, fill=white] (D) [right of=C, node distance=3cm, yshift=0cm] {$\ell'_3$};
  \node[state, fill=white] (E) [right of=D, node distance=3cm, yshift=0cm] {$Acc$};
  \node[state, fill=white] (F) [below of=B, node distance=2cm, yshift=0cm] {$\ell'_4$};
   \node[state, fill=white, color=white] (A') [left of=A, node distance=2cm, yshift=0cm] {};
  
  \path (A') edge node [above] {} (A)
        (A) edge node [above] {$x=1, \tau$} (B)    
        (B) edge node [above] {$x<1, a?$} (C)   
        (C) edge node [above] {$b!$} (D)   
        (D) edge node [above] {$b!$} (E)   
        (A) edge node [right] {othw} (F)    
        (B) edge node [right] {othw} (F)   
        (C) edge node [right] {othw} (F)   
        (D) edge node [right] {othw} (F)   
        (F) edge [loop left] node [left] {$\Sigma_{\tp}$} (F) 
        (E) edge [loop right] node [right] {$\Sigma_{\tp}$} (E) 
;

\end{tikzpicture}
}
\caption{Test purpose $\TP$.
}\label{ExObj}
\end{center}
\end{figure}

\begin{exa} 
Figure~\ref{ExObj} represents a test purpose $\TP$ for the specification $\A$
of Figure~\ref{ExSpec}. This one has no proper clock and observes the unique
clock $x$ of $\A$.  It accepts sequences where $\tau$ occurs at $x=1$,
followed by an input $a$ at $x<1$ (thus focusing on the lower branch
of $\A$ where $x$ is reset), and two subsequent $b$'s. The label $othw$ (for otherwise) on a transition is an abbreviation for the complement 
of specified transitions leaving the same location.
For example in location $\ell'_1$, $othw$ stands for 
$\{(\true, \tau),(\true,b!),(x\geq 1,a?)\}$.
\end{exa}

\subsection{Principle of test generation}

Given a specification TAIO $\A$ and a test purpose $(\TP,
\accept^\tp)$, the aim is to build a sound and, if possible strict 
test case $(\TC,\Verdicts)$ focusing on behaviors accepted by $\TP$.  
As $\TP$ accepts sequences of $\A$, but  test cases observe timed traces,
the intention is that $\TC$ should deliver $\Pass$ verdicts 
on traces of sequences of $\A$ accepted by $\TP$ in $\accept^\tp$.  
This property is formalized by the following definition: 
\begin{defi}
\label{def:precise}
A test suite $\TS$ for $\A$ and $\TP$ is said to be {\em precise} 
if for any test case $\TC$ in $\TS$,
for any timed observation $\sigma$ in $\traces(\TC)$, 
${\tt Verdict}(\sigma, \TC) = \Pass$
if and only if 
$\sigma \in \traces(\seq(\A\!\!\uparrow^{(X_p^\tp,X_o^\tp)}) \cap \seq_{\accept^\tp}(\TP))$.
\end{defi}

Let $\A=(L^\pA,\ell_0^\pA,\Sigma_?^{\pA},\Sigma_!^{\pA},\Sigma_\tau
^{\pA}, X_p^\pA,\emptyset, M^\pA,\Inv^\pA,E^\pA)$ be the specification
TAIO, and $\TP=
(L^\tp,\ell_0^\tp,\Sigma_?^{\pA},\Sigma_!^{\pA},\Sigma_\tau ^{\pA},
X_p^\tp,X_o^\tp,M^\tp,\Inv^\tp,E^\tp)$ be a test purpose for $\A$,
with its set $\accept^\tp$ of accepting locations.  The generation of
a test case $\TC$ from $\A$ and $\TP$ proceeds in several steps.
First, sequences of $\A$ accepted by $\TP$ are identified by the
computation of the product $\AxTP$ of those OTAIOs.  Then a
determinization step is necessary to characterize conformant traces as
well as traces of accepted sequences.  Then the resulting
deterministic TAIO $\DP$ is transformed into a test case TAIO $\TC'$
with verdicts assigned to states.  Finally, the test case $\TC$ is
obtained by a selection step which tries to avoid some $\Inconc$
verdicts.  The different steps of the test generation process from
$\A$ and $\TP$ are detailed in the following paragraphs.

\subsubsection*{Computation of the product:} 
First, the product $\AxTP = \A \times \TP$
is built
(see Definition~\ref{def_product} for the definition of the product), 
associated with the set
of marked locations $\accept^\axtp = L^\pA \times \accept^\tp$.
Let $P= (L^\axtp,\ell_0^\axtp,\Sigma_?^{\pA},\Sigma_!^{\pA},\Sigma_\tau ^{\pA},
  X_p^\axtp,X_o^\axtp,M^\axtp,\Inv^\axtp,E^\axtp)$. 
As $X_o^\tp=X_p^\pA$, we get $X_o^\axtp= \emptyset$ and $X_p^\axtp=X_p^\pA\sqcup X_p^\tp$, thus $\AxTP$ is in fact a TAIO.

The effect of the product is to unfold $\A$
and to mark locations of the product by $\accept^\axtp$, 
so that sequences of $\A$ accepted by $\TP$ are identified.
As $\TP$ is complete, 
$\seq(\TP)\downarrow_{X_p^\tp} = (\setRnn \times (\Sigma^\tp \times 2^{X_o^\tp}))^*$,
thus, by the properties of the product (see equation~\ref{eq-seq}),
$\seq(\AxTP)\downarrow_{X_p^\tp} = \seq(\A)$ 
\ie the sequences of the product after removing resets of proper clocks of $\TP$
are the sequences of $\A$.
As a consequence 
$\traces(\AxTP) = \traces(\A)$, which entails that $\AxTP$ and $\A$
define the same sets of conformant implementations.  

 Considering accepted sequences of the product $\AxTP$, 
by equation~\ref{eq-seq-acc} we  get the equality
$\seq_{\accept^\axtp}(\AxTP) = \seq(\A\!\!\uparrow^{(X_p^\tp,X_o^\tp)}) \cap
\seq_{\accept^\tp}(\TP)$, 
which induces the desired characterization of accepted traces: 
$\traces_{\accept^\axtp}(\AxTP) = \traces(\seq(\A\!\!\uparrow^{(X_p^\tp,X_o^\tp)}) \, \cap \,\seq_{\accept^\tp}(\TP))$.

Using the notation $\pref(T)$ 
for the set of prefixes
of traces in a set of traces $T$,
we note $\rtraces(\A,\TP) = \traces(\A) \setminus \pref(\traces_{\accept^\axtp}(\AxTP))$
for the set of traces of $\A$ which are not prefixes of accepted traces of $\AxTP$.
In the sequel, 
the principle of test selection will be to try to select traces in
$\traces_{\accept^\axtp}(\AxTP)$ (and assign to them the $\Pass$ verdict)
and to try to avoid or at least detect (with an $\Inconc$ verdict) 
those traces in $\rtraces(\A,\TP)$, as
these traces cannot be prefixes of traces of sequences satisfying the test purpose.

\begin{exa}
Figure~\ref{ExProd} represents the product $\AxTP$ for the specification
$\A$ in Figure~\ref{ExSpec} and the test purpose $\TP$ in
Figure~\ref{ExObj}. 
As $\TP$ describes one branch of $\A$, the product is very simple in this case,
\eg~intersection of guards are trivial.
The only difference with $\A$ is the  tagging with $\accept^\axtp$.
\end{exa}

\begin{figure}[htbp]
\begin{center}
\scalebox{0.65}{
\begin{tikzpicture}[->,>=stealth',shorten >=1pt,auto,node distance=2cm,
                    semithick]

  \tikzstyle{every state}=[text=black]

  \node[state, fill=white] (A) {$\ell_0\ell'_0$};
  \node[state, fill=white] (B) [right of=A, node distance=3.5cm, yshift=1cm] {$\ell_1\ell'_1$};
  \node[state, fill=white] (C) [right of=B, node distance=4.5cm, yshift=0cm] {$\ell_2\ell'_4$};
  \node[state, fill=white] (D) [right of=C, node distance=3cm, yshift=0cm] {$\ell_3\ell'_4$};
  \node[state, fill=white] (E) [right of=D, node distance=3cm, yshift=0cm] {$\ell_4\ell'_4$};
  \node[state, fill=white] (F) [right of=A, node distance=3.5cm, yshift=-1cm] {$\ell_5\ell'_1$};
  \node[state, fill=white] (G) [right of=F, node distance=4.5cm, yshift=0cm] {$\ell_6\ell'_2$};
  \node[state, fill=white] (H) [right of=G, node distance=3cm, yshift=0cm] {$\ell_7\ell'_3$};
  \node[state, fill=white] (I) [right of=H, node distance=3cm, yshift=0cm] {$\ell_8Acc$};
  \node[state, fill=white, color=white] (A') [left of=A, node distance=2cm, yshift=0cm] {};
  \node[fill=white] (Ae) [above of=A, node distance=.9cm, yshift=0cm] {$x \le 1$};
  \node[fill=white] (Ce) [below of=C, node distance=.9cm, yshift=0cm] {$x \le 1$};
  \node[fill=white] (De) [below of=D, node distance=.9cm, yshift=0cm] {$x \le 1$};
  \node[fill=white] (Ge) [below of=G, node distance=.85cm, yshift=0cm] {$x = 0$};
  \node[fill=white] (He) [below of=H, node distance=.85cm, yshift=0cm] {$x = 0$};
  
  \path (A') edge node [above] {} (A)
        (A) edge node [above,sloped] {$x=1, \tau$} (B)    
        (B) edge node [above] {$1<x<2, a?, \{x\}$} (C)   
        (C) edge node [above] {$x=0, b!$} (D)   
        (D) edge node [above] {$b!$} (E)   
        (C) edge [loop above] node [above] {$x=1, \tau, \{x\}$} (C)    
        (A) edge node [above,sloped] {$x=1, \tau, \{x\}$} (F)    
        (F) edge node [above] {$x<1, a?, \{x\}$} (G)   
        (G) edge node [above] {$b!$} (H)   
        (H) edge node [above] {$b!$} (I)   
;

\end{tikzpicture}}
\caption{Product $\AxTP= \A \times \TP$.
}\label{ExProd}
\end{center}
\end{figure}
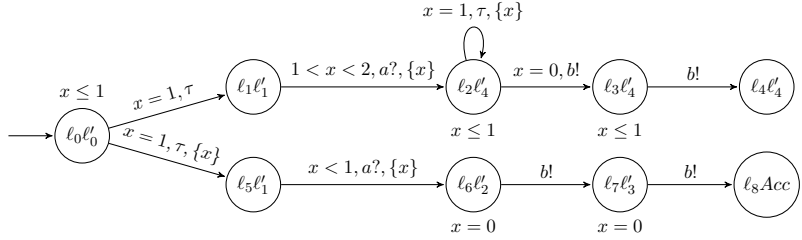

\subsubsection*{Approximate determinization of $\AxTP$ into $\DP$:}
We now want to transform $\AxTP$ into a deterministic TAIO $\DP$
such that $\AxTP \preceq \DP$, which  by Proposition~\ref{prop:tioco-alt-sim})
will entail that implementations conformant to $\AxTP$ (thus to $\A$) 
are still conformant to $\DP$. 
If $\AxTP$ is already deterministic, we simply take $\DP=\AxTP$.
Otherwise, 
the approximate determinization of Section~\ref{sec-determinisation} provides a solution.
The user fixes some resources $(k,M^\pdp)$,
then a deterministic io-abstraction  $\DP$ of $\AxTP$ with  resources $(k,M^\pdp)$ is computed.
By Proposition~\ref{altSim}, we thus get that  $\DP$ io-abstracts $\AxTP$.
$\DP$ is equipped with the set of marked locations $\accept^\pdp$
consisting of locations in $L^\pdp$ containing some configuration
whose location is in $\accept^\axtp$. 
As a consequence traces of $\DP$ which are traces of sequences accepted by $\AxTP$ 
in $\accept^\axtp$ are accepted by $\DP$ in $\accept^\pdp$,  formally 
$\traces(\DP) \cap \traces(\seq_{\accept^\axtp}(\AxTP)) = 
\traces(\DP) \cap \traces_{\accept^\axtp}(\AxTP) \subseteq \traces_{\accept^\pdp}(\DP)$.
This means that extra accepted traces may be added due to over-approximations, 
some traces may be lost (including accepted ones) 
by under-approximations, but if the under-approximation preserves
some traces that are accepted in $\AxTP$, these are still accepted in $\DP$.   
If the determinization is exact (or $\AxTP$ is already deterministic), 
of course we get more precise relations between the traces and accepted traces of $\AxTP$ and $\DP$,
namely
 $\traces(\DP) = \traces(\AxTP)$ and
$\traces_{\accept^\pdp}(\DP) = \traces_{\accept^\axtp}(\AxTP)$.  

\ignore{
\textcolor{red}{
\TJ{on devrait caracteriser
 les traces acceptees dans $\DP$ par rapport a celles de $\AxTP$. 
Or si on a fait une sous-approximation, on peut perdre des traces acceptees,
alors que si on a fait une sur-approximation, on peut en ajouter.
}.
}
}

\begin{exa}
Figure~\ref{ExJeu}
partially represents the game $\mathcal{G}_{\AxTP,(1,2)}$ for the TAIO $\AxTP$ of Figure~\ref{ExProd} where,
for readability reasons, some behaviors not co-reachable from
$\accept^\pdp$ (dotted green states) are omitted.
Notice that the construction of the initial part of the game was explained  
in Example~\ref{exa3.3}.
A strategy $\Pi$ for Determinizator is represented by bold arrows.
$\Pi$ is not winning (the unsafe configuration, in gray, is unavoidable from the initial state), 
and in fact an approximation is performed.
$\DP$, represented in Figure~\ref{fig:dp} 
is simply obtained from $\mathcal{G}_{\AxTP,(1,2)}$ and the strategy $\Pi$ 
by merging transitions of Spoiler and those of Determinizator in the strategy.
\end{exa}

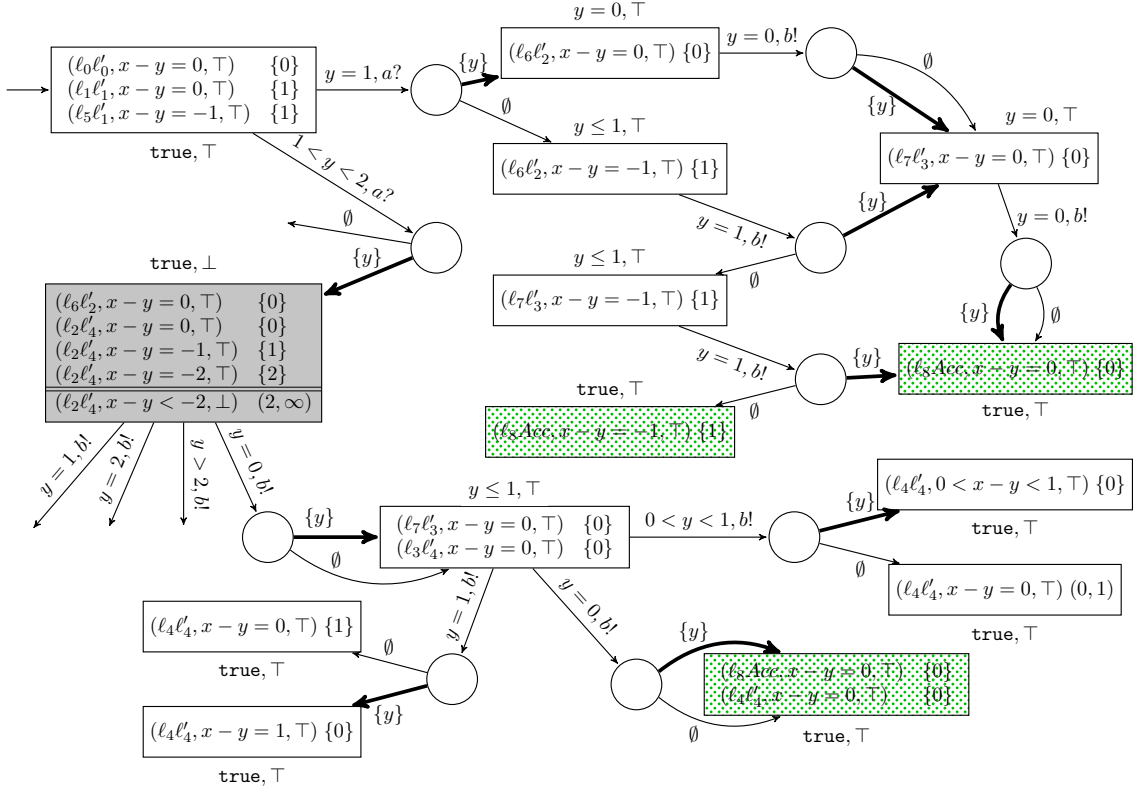
\begin{figure}[htbp]
\begin{center}
\scalebox{0.7}{
\begin{tikzpicture}[->,>=stealth',shorten >=1pt,auto,node distance=2cm,
                    semithick]

  \tikzstyle{every state}=[text=black]

  \node[state, fill=white] [rectangle] (A)
  {$\begin{array}{ll}(\ell_0\ell'_0,x-y=0,\top)& \{0\}\\ 
 (\ell_1\ell'_1,x-y=0,\top)& \{1\}\\ (\ell_5\ell'_1,x-y=-1,\top)& \{1\} \end{array}$};
  \node[state, fill=white] (AB) [right of=A, node distance=4.8cm, yshift=0cm] {};
  \node[state, fill=white] [rectangle] (B1) [right of=AB, node distance=3.3cm, yshift=.7cm] {$(\ell_6\ell'_2,x-y=0,\top)\; \{0\}$};
  \node[state, fill=white] (AC) [right of=A, node distance=4.8cm, yshift=-3cm] {};
  \node[fill=white] (C2) [left of=AC, node distance=3cm, yshift=.5cm] {};
  \node[state, fill=GrisC!90!black,inner xsep=0] [rectangle] (C1) [below of=A, node distance=5cm, yshift=0cm] {$\begin{array}{ll}(\ell_6\ell'_2,x-y=0,\top)& \{0\}\\
(\ell_2\ell'_4,x-y=0,\top)& \{0\}\\ 
 (\ell_2\ell'_4,x-y=-1,\top)& \{1\}\\
 (\ell_2\ell'_4,x-y=-2,\top)& \{2\} \\ \hline\hline
 (\ell_2\ell'_4,x-y<-2, \bot) & (2,\infty)\end{array}$};
  \node[state, fill=white] (C1F) [right of=C1, node distance=1.6cm, yshift=-3.5cm] {};
  \node[state, fill=white] [rectangle] (F) [right of=C1F, node distance=4.5cm, yshift=-0cm] {$\begin{array}{ll}(\ell_7\ell'_3,x-y=0,\top)& \{0\}\\ (\ell_3\ell'_4,x-y=0,\top)& \{0\}\end{array}$};
  \node[state, fill=white] (FH) [below of=F, node distance=2.8cm, xshift=2.5cm] {};
  \node[state, fill=green, pattern color= green!75!black, pattern=crosshatch dots] [rectangle] (H) [right of=FH, node
  distance=3.8cm, yshift=0cm]
  {$\begin{array}{ll}(\ell_8Acc,x-y=0,\top)& \{0\}\\
      (\ell_4\ell'_4,x-y=0,\top) &\{0\} \end{array}$};
  \node[color=white] (A') [left of=A, node distance=3.5cm, yshift=0cm] {};
  \node[state, fill=white] (B1V) [right of=B1, node distance=4.2cm, yshift=0cm] {};
  \node[state, fill=white] [rectangle] (V) [right of=B1V, node distance=3cm, yshift=-2cm] {$(\ell_7\ell'_3,x-y=0,\top)\; \{0\}$};
  \node[state, fill=white] (VW) [right of=V, node distance=.7cm, yshift=-2cm] {};
  \node[state, fill=green, pattern color= green!75!black, pattern=crosshatch dots] [rectangle] (W) [below of=VW, node distance=2cm, xshift=-.2cm] {$(\ell_8Acc,x-y=0,\top)\; \{0\}$};
  \node[fill=white] (Ae) [below of=A, node distance=1.2cm, yshift=0cm] {$\true,\top$};
   \node[fill=white] (C1e) [above of=C1, node distance=1.7cm, xshift=0cm] {$\true,\bot$};
   \node[fill=white] (He) [below of=H, node distance=1cm, yshift=0cm] {$\true,\top$};
   \node[fill=white] (We) [below of=W, node distance=.8cm, yshift=0cm] {$\true,\top$};
  \node[fill=white] (Fe) [above of=F, node distance=.9cm, yshift=0cm] {$y \le 1,\top$};
  \node[fill=white] (Ve) [above of=V, node distance=.8cm, xshift=1cm] {$y = 0,\top$};
  \node[fill=white] (B1e) [above of=B1, node distance=.8cm, yshift=0cm] {$y = 0,\top$};

  \node[state, fill=white] [rectangle] (B2) [right of=AB, node distance=3.3cm, yshift=-1.5cm] {$(\ell_6\ell'_2,x-y=-1,\top)\; \{1\}$};
  \node[state, fill=white] (B2V2) [right of=B2, node distance=4cm, yshift=-1.5cm] {};
  \node[state, fill=white] [rectangle] (V2) [left of=B2V2, node distance=4cm, yshift=-1cm] {$(\ell_7\ell'_3,x-y=-1,\top)\; \{1\}$};
  \node[state, fill=white] (V2W2) [right of=V2, node distance=4cm, yshift=-1.5cm] {};
  \node[state, fill=green, pattern color=green!75!black, pattern=crosshatch dots] [rectangle] (W2) [below of=V2W2, node distance=1cm, xshift=-4cm] {$(\ell_8Acc,x-y=-1,\top)\; \{1\}$};
  \node[fill=white] (V2e) [above of=V2, node distance=.8cm, yshift=0cm] {$y \le 1,\top$};
  \node[fill=white] (W2e) [above of=W2, node distance=.8cm, yshift=0cm] {$\true,\top$};
  \node[fill=white] (B2e) [above of=B2, node distance=.8cm, yshift=0cm] {$y \le 1,\top$};

  \node[fill=white] [ellipse] (D2) [below of=C1, node distance=3.5cm, xshift=-3cm] {};
  \node[fill=white] [ellipse] (D3) [below of=C1, node distance=3.5cm, xshift=-1.5cm] {};
  \node[fill=white] [ellipse] (D4) [below of=C1, node distance=3.5cm, xshift=0cm] {};
  
  \node[state, fill=white] (F2H2) [right of=F, node distance=5.5cm, yshift=0cm] {};
  \node[state, fill=white] [rectangle] (H21) [right of=F2H2, node distance=4cm, yshift=1cm] {$ (\ell_4\ell'_4,0<x-y<1,\top)\; \{0\}$}; 
    \node[state, fill=white] [rectangle] (H22) [right of=F2H2, node distance=4cm, yshift=-1cm] {$ (\ell_4\ell'_4,x-y=0,\top)\; (0,1)$}; 
  
  \node[state, fill=white] (F3H3) [below of=F, node distance=2.7cm, xshift=-1cm] {};
    \node[state, fill=white] [rectangle] (H31) [left of=F3H3, node distance=3.8cm, yshift=-1cm] {$ (\ell_4\ell'_4,x-y=1,\top)\; \{0\}$}; 
    \node[state, fill=white] [rectangle] (H32) [left of=F3H3, node distance=3.8cm, yshift=1cm] {$ (\ell_4\ell'_4,x-y=0,\top)\; \{1\}$};

   \node[fill=white] (H21e) [below of=H21, node distance=.85cm, yshift=0cm] {$\true,\top$};
   \node[fill=white] (H22e) [below of=H22, node distance=.85cm, yshift=0cm] {$\true,\top$};
   \node[fill=white] (H31e) [below of=H31, node distance=.85cm, yshift=0cm] {$\true,\top$};
   \node[fill=white] (H32e) [below of=H32, node distance=.85cm, yshift=0cm] {$\true,\top$};

  \path (A') edge node [above] {} (A)
        (A) edge node [above] {$y=1, a?$} (AB)    
        (AB) edge [line width=2pt] node [above,pos=.3] {$\{y\}$} (B1)   
        (AB) edge node [above] {$\emptyset$} (B2)   
        (A) edge node [above, sloped] {$1<y<2,a?$} (AC)   
        (AC) edge [line width=2pt] node [above] {$\{y\}$} (C1)    
        (AC) edge node [above] {$\emptyset$} (C2)    
         (C1) edge node [above,sloped] {$y=1, b!$} (D2)      
         (C1) edge node [sloped, above] {$y=2, b!$} (D3)      
         (C1) edge node [sloped, above] {$y>2, b!$} (D4)    
         (C1) edge node [sloped, above] {$y=0, b!$} (C1F)       
         (F) edge node [sloped,above] {$y=0, b!$} (FH)    
         (C1F) edge [line width=2pt] node [above,pos=.3] {$\{y\}$} (F)     
         (C1F) edge [bend right] node [above,pos=.3] {$\emptyset$} (F)    
         (FH) edge [bend left, line width=2pt] node [above,pos=.3] {$\{y\}$} (H)     
         (FH) edge [bend right] node [below,pos=.3] {$\emptyset$} (H)    
         (F) edge node [above,sloped] {$0<y<1, b!$} (F2H2)      
         (F) edge node [sloped, above] {$y=1, b!$} (F3H3)     
         (B1) edge node [above] {$y=0, b!$} (B1V)       
         (V) edge node [right, pos=.6] {$y=0, b!$} (VW)        
        (B1V) edge [line width=2pt] node [below,pos=.3] {$\{y\}$} (V)    
        (VW) edge [bend right, line width=2pt] node [left] {$\{y\}$} (W)    
        (B1V) edge [bend left] node [above] {$\emptyset$} (V) 
        (VW) edge [bend left] node [right] {$\emptyset$} (W) 
         (B2) edge node [below, sloped] {$y=1, b!$} (B2V2)       
         (V2) edge node [below, sloped] {$y=1, b!$} (V2W2)        
        (B2V2) edge [line width=2pt] node [above,pos=.3] {$\{y\}$} (V)    
        (V2W2) edge [line width=2pt] node [above,pos=.4] {$\{y\}$} (W)    
        (B2V2) edge node [below] {$\emptyset$} (V2) 
        (V2W2) edge node [below] {$\emptyset$} (W2) 
        (F2H2) edge [line width=2pt] node [above] {$\{y\}$} (H21)    
        (F2H2) edge node [below] {$\emptyset$} (H22) 
        (F3H3) edge [line width=2pt] node [below] {$\{y\}$} (H31)    
        (F3H3) edge node [above] {$\emptyset$} (H32) 

;

\end{tikzpicture}}
\caption{Game $\mathcal{G}_{\AxTP,(1,2)}$.
}\label{ExJeu}
\end{center}
\end{figure}

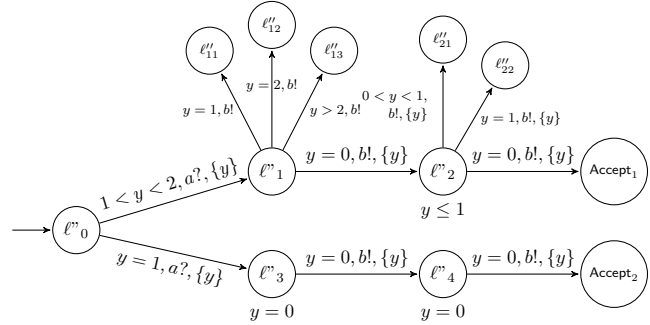
\begin{figure}[htbp]
\begin{center}
\scalebox{0.65}{
\begin{tikzpicture}[->,>=stealth',shorten >=1pt,auto,node distance=2cm,
                    semithick]

  \tikzstyle{every state}=[text=black]

  \node[state, fill=white] (A) {$\ell"_0$};
  \node[state, fill=white] (B) [right of=A, node distance=4cm,
yshift=1.2cm] {$\ell"_1$};

  \node[state, fill=white] (B1) [above left of=B, node distance=1.8cm,
yshift=1.2cm] {\scriptsize{$\ell''_{11}$}};
  \node[state, fill=white] (B2) [above of=B, node distance=1.8cm,
yshift=1.2cm] {\scriptsize{$\ell''_{12}$}};
  \node[state, fill=white] (B3) [above right of=B, node distance=1.8cm,
yshift=1.2cm] {\scriptsize{$\ell''_{13}$}};

  \node[state, fill=white] (C) [right of=B, node distance=3.5cm,
yshift=0cm] {$\ell"_2$};

  \node[state, fill=white] (C1) [above  of=C, node distance=1.8cm,
yshift=.9cm] {\scriptsize{$\ell''_{21}$}};
  \node[state, fill=white] (C2) [above right of=C, node distance=1.8cm,
yshift=.9cm] {\scriptsize{$\ell''_{22}$}};

  \node[state, fill=white] (D) [right of=C, node distance=3.5cm,
yshift=0cm] {\scriptsize{$\accept_1$}};
  \node[state, fill=white] (F) [right of=A, node distance=4cm,
yshift=-.9cm] {$\ell"_3$};

  \node[state, fill=white] (G) [right of=F, node distance=3.5cm,
yshift=0cm] {$\ell"_4$};

\node[fill=white] (Ci) [below of=C, node distance=.8cm, yshift=0cm] {$y \leq 1$};
 \node[fill=white] (Fi) [below of=F, node distance=.8cm, yshift=0cm] {$y=0$};
\node[fill=white] (Gi) [below of=G, node distance=.8cm, yshift=0cm] {$y=0$};

  \node[state, fill=white] (H) [right of=G, node distance=3.5cm,
yshift=0cm] {\scriptsize{\bf $\accept_2$}};
  \node[state, fill=white, color=white] (A') [left of=A, node
distance=1.8cm, yshift=0cm] {};

  \path (A') edge node [above] {} (A)
        (A) edge node [below,sloped] {$y=1, a?,\{y\}$} (F)
        (F) edge node [above] {$y=0, b!, \{y\}$} (G)
        (G) edge node [above] {$y=0, b!, \{y\}$} (H)
        (A) edge node [above,sloped] {$1<y<2, a?,\{y\}$} (B)
        (B) edge node [above] {$y=0, b!, \{y\}$} (C)
        (C) edge node [above] {$y=0, b!, \{y\}$} (D);

\path       (B) edge node [left] {\scriptsize{$y=1,b!$}} (B1)
        (B) edge node [above] {\scriptsize{$y=2,b!$}} (B2)
        (B) edge node [right] {\scriptsize{$y > 2,b!$}} (B3)
        (C) edge node [left] {\scriptsize{$\begin{array}{r} 0<y<1,\\b!,\{y\}\end{array}$}} (C1)
        (C) edge node [right] {\scriptsize{$y=1,b!,\{y\}$}} (C2)
;

\end{tikzpicture}}
\end{center}
\caption{Deterministic automaton $\DP=Aut(\Pi)$.}
\label{fig:dp}
\end{figure}

\subsubsection*{Generating  $\TCa$ from $\DP$:}
The next step consists in building a test case $(\TCa,\Verdicts)$ from $\DP$.
The main point is the computation of verdicts.
$\Pass$  verdicts are simply defined from $\accept^\pdp$.
$\Fail$ verdicts that should detect unexpected outputs and delays, 
rely on a complementation. 
The difficult part is the computation of  
$\Inconc$ states
which should detect when $\accept^\pdp$ is not reachable 
(or equivalently 
$\None$ states, those states where $\accept^\pdp$ is still reachable)
and thus relies on an analysis of the co-reachability to locations $\accept^\pdp$.
Another interesting point is the treatment of invariants.
First $\TCa$ will have no invariants 
(which ensures that it is non-blocking).
Second, invariants in $\DP$ 
are shifted to guards in $\TCa$ and
in the definition of $\Fail$
so that 
test cases check that  the urgency specified in $\A$ is satisfied by $\Imp$.

The test case constructed from
$\DP=(L^\pdp,\ell_0^\pdp,\Sigma_?^\pdp,\Sigma_!^\pdp,\emptyset,X_p^\pdp,\emptyset,M^\pdp,\Inv^\pdp,E^\pdp)$
and $\accept^\pdp$ is the pair $(\TCa, \Verdicts)$
where:
\begin{iteMize}{$\bullet$}
\item $\TCa=(L^\tca,\ell_0^\tca,\Sigma_?^\tca,\Sigma_!^\tca,\emptyset ,X_p^\tca,\emptyset,M^\tca,\Inv^\tca,E^\tca)$ is the TAIO
such that: 
\begin{iteMize}{$-$}
\item $L^\tca= L^\pdp \sqcup \{\lfail\}$ where $\lfail$ is a new
location; 
\item $\ell_0^\tca=\ell_0^\pdp$ is the initial location;
\item $\Sigma_?^\tca=\Sigma_!^\pdp=\Sigma_!^\pA$ and
$\Sigma_!^\tca=\Sigma_?^\pdp=\Sigma_?^\pA$, \ie input/output alphabets
are mirrored in order to reflect the opposite role of actions in the
synchronization of $\TCa$ and $\Imp$; 
\item $X_p^\tca = X_p^\pdp$ and
$X_o^\tca=X_o^\pdp=\emptyset$; 
\item $M^\tca = M^\pdp$; 
\item $\Inv^\tca(\ell)=\true$ for any $\ell \in L^\tca$; 
\item $E^\tca=E_\Inv^\pdp
\sqcup E_{\lfail}$ where
 $$ \quad \quad \quad \quad 
\begin{array}{l}
E_\Inv^\pdp=\{(\ell,g\wedge\Inv^\pdp(\ell),a,X',\ell') \mid
(\ell,g,a,X',\ell') \in E^\pdp\} \mbox{ and}\\
E_{\lfail}=\left \{(\ell,\bar{g}\wedge\Inv^\pdp(\ell),a,X_p^\tca,\lfail) \,  \left|\, 
\begin{array}{l}
\ell \in L^\pdp,\ a \in \Sigma_!^\pdp \\
\textrm{ and } \bar{g}= \neg \bigvee_{(\ell,g,a,X',\ell') \in
  E^\pdp} g \end{array} \right\}\right. \end{array}$$
\end{iteMize}
\item 
 $\Verdicts$ is the partition of $S^\pdp$ defined as follows:
\begin{iteMize}{$-$}
\item $\Pass=
\bigcup_{\ell\in\accept^\pdp} (\{\ell\}\times\Inv^\pdp(\ell))$,
\item $\None =
\coreach(\DP,\Pass) \setminus \Pass$, 
\item $\Fail=\{\lfail\} \times \setRnn^{X^\tca} \sqcup
\{(\ell,\neg \Inv^\pdp(\ell)) | \ell \in L^\pdp\}$;
\item $\Inconc= S^\pdp \setminus
(\Pass \sqcup \Fail \sqcup \None)$,
 \end{iteMize}
\end{iteMize}\medskip

\noindent The important points to understand in the construction of $\TCa$ are
the completion to $\Fail$ and the computation of $\None$,
which, together with $\Pass$,  define $\Inconc$ by complementation.  

For the
completion to $\Fail$, the idea is to detect unspecified outputs and delays 
with respect to $\DP$.  
Remember that outputs of $\DP$ are inputs of $\TCa$. 
Moreover, authorized delays in $\DP$ are defined by invariants, 
but remember that test cases have no invariants (they are $\true$ in all locations).
First, all states in $(\ell,\neg\Inv^\pdp(\ell)), \ell \in L^\pdp$, \ie
states where the invariant runs out, are put into $\Fail$
which reflects the counterpart in $\TC'$ of the urgency in $\DP$. 
Then, in each location $\ell$, the invariant $\Inv^\pdp(\ell)$ in $\DP$ is
removed and shifted to guards of all transitions leaving $\ell$ in
$\TCa$, as defined in $E_I^\pdp$.
Second, in any location $\ell$, for each input $a\in\Sigma_?^\tca=\Sigma_!^\pdp$, a transition
leading to $\lfail$ is added, labeled with $a$, and whose guard is
the conjunction of $\Inv(\ell)$ 
with the negation of the disjunction of all guards of transitions labeled
by $a$ and leaving $\ell$ (thus $\true$ if no $a$-action leaves
$\ell$), as defined in $E_{\lfail}$.  
It is then easy to see that $\TC'$ is input-complete in all states.


The computation of $\None$ is based on an analysis of the
co-reachability to $\Pass$.  
$\None$ contains all states co-reachable from locations in $\Pass$.  
Notice that the set of states $\coreach(\DP,\Pass)$, and thus $\None$,
can be computed symbolically as usual in the region graph
of $\DP$, or more efficiently using zones.

\begin{exa}
  Figure~\ref{fig:tca} represents the test case $\TCa$ obtained from
  $\DP$.  For readability reasons, we did not represent transitions in
  $E_{\lfail}$, except the one leaving $\ell"_0$.  In fact these are
  removed in the next selection phase as they are only fireable from
  states where a verdict has already been issued.  The rectangles
  attached to locations represent the verdicts in these locations when
  clock $y$ progresses between 0 and 2, and after 2: dotted green for
  $Pass$, black for $\None$, blue grid for $\Inconc$ and
  crosshatched red for $\Fail$.  For example, in $\ell"_2$, the
  verdict is initially $\None$, becomes $\Inconc$ if no $b$ is
  received immediately, and even $\Fail$ if no $b$ is received before
  one time unit.  Notice that in order to reach a $\Pass$ verdict, one
  should initially send $a$ after one and strictly before two time
  units, and expect to receive two consecutive $b$'s immediately
  after.
\end{exa}

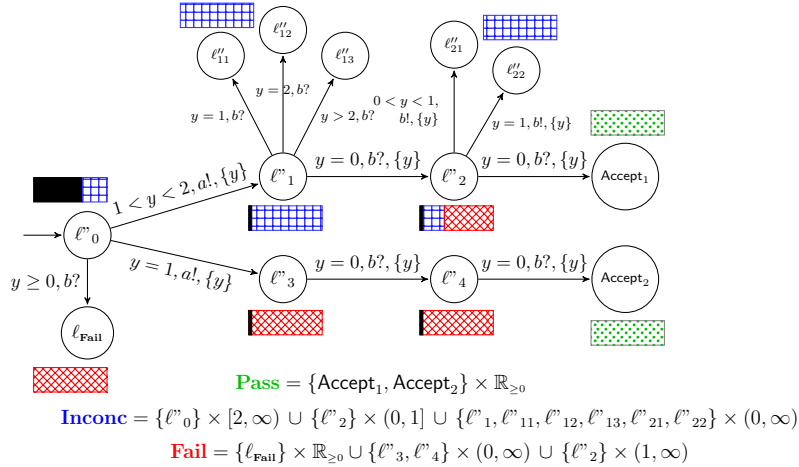
\begin{figure}[htbp]
\begin{center}
\scalebox{0.65}{
\begin{tikzpicture}[->,>=stealth',shorten >=1pt,auto,node distance=2cm,
                    semithick]

  \tikzstyle{every state}=[text=black]

  \node[state, fill=white] (A) {$\ell"_0$};
  \node[state, fill=white] (B) [right of=A, node distance=4cm,
yshift=1.2cm] {$\ell"_1$};

  \node[state, fill=white] (B1) [above left of=B, node distance=1.8cm,
yshift=1.2cm] {\scriptsize{$\ell''_{11}$}};
  \node[state, fill=white] (B2) [above of=B, node distance=1.8cm,
yshift=1.2cm] {\scriptsize{$\ell''_{12}$}};
  \node[state, fill=white] (B3) [above right of=B, node distance=1.8cm,
yshift=1.2cm] {\scriptsize{$\ell''_{13}$}};

  \node[state, fill=white] (C) [right of=B, node distance=3.5cm,
yshift=0cm] {$\ell"_2$};

  \node[state, fill=white] (C1) [above  of=C, node distance=1.8cm,
yshift=.9cm] {\scriptsize{$\ell''_{21}$}};
  \node[state, fill=white] (C2) [above right of=C, node distance=1.8cm,
yshift=.9cm] {\scriptsize{$\ell''_{22}$}};

  \node[state, fill=white] (D) [right of=C, node distance=3.5cm,
yshift=0cm] {\scriptsize{$\accept_1$}};
  \node[state, fill=white] (E) [below of=A, node distance=2cm, xshift=0cm]
{$\lfail$};
  \node[state, fill=white] (F) [right of=A, node distance=4cm,
yshift=-.9cm] {$\ell"_3$};

  \node[state, fill=white] (G) [right of=F, node distance=3.5cm,
yshift=0cm] {$\ell"_4$};

  \node[state, fill=white] (H) [right of=G, node distance=3.5cm,
yshift=0cm] {\scriptsize{\bf $\accept_2$}};
  \node[state, fill=white, color=white] (A') [left of=A, node
distance=1.8cm, yshift=0cm] {};

  \node[fill=white] (Ge) [below of=A, node distance=4.45cm, xshift=7cm]
{\large{$\rouge{\Fail} = \{\lfail\}\times\setRnn \cup \{\ell"_3,\ell"_4\}\times (0,\infty) 
\,\cup\, \{\ell"_2\}\times (1,\infty)$}};
  \node[fill=white] (He) [below of=A, node distance=3.75cm, xshift=7cm]
{\large{$\bleu{\Inconc} = \{\ell"_0\}\times [2,\infty) \,\cup\, \{\ell"_2\}\times(0,1] \,\cup\,
\{\ell"_1,\ell"_{11},\ell"_{12},\ell"_{13},\ell"_{21},\ell"_{22}\}\times(0,\infty) $}};
  \node[fill=white] (Ie) [below of=A, node distance=3.05cm, xshift=6cm]
{\large{$\verte{\Pass} = \{\accept_1,\accept_2\}\times \setRnn$}};

  \path (A') edge node [above] {} (A)
        (A) edge node [below,sloped] {$y=1, a!,\{y\}$} (F)
        (F) edge node [above] {$y=0, b?, \{y\}$} (G)
        (G) edge node [above] {$y=0, b?, \{y\}$} (H)
        (A) edge node [left] {$y\ge0, b?$} (E)
        (A) edge node [above,sloped] {$1<y<2, a!,\{y\}$} (B)
        (B) edge node [above] {$y=0, b?, \{y\}$} (C)
        (C) edge node [above] {$y=0, b?, \{y\}$} (D);

\path       (B) edge node [left] {\scriptsize{$y=1,b?$}} (B1)
        (B) edge node [above] {\scriptsize{$y=2,b?$}} (B2)
        (B) edge node [right] {\scriptsize{$y > 2,b?$}} (B3)
        (C) edge node [left] {\scriptsize{$\begin{array}{r} 0<y<1,\\b!,\{y\}\end{array}$}} (C1)
        (C) edge node [right] {\scriptsize{$y=1,b!,\{y\}$}} (C2)
;

\def\rectanglepath{-- ++(.5cm,0cm) -- ++(0cm,.5cm) -- ++(-.5cm,0cm) -- cycle}
\def\rectanglepathone{-- ++(.5cm,0cm) -- ++(0cm,.5cm) -- ++(-.5cm,0cm) -- cycle}
\def\rectanglepathtwo{-- ++(1cm,0cm) -- ++(0cm,.5cm) -- ++(-1cm,0cm) -- cycle}
\def\rectanglepaththree{-- ++(1.5cm,0cm) -- ++(0cm,.5cm) -- ++(-1.5cm,0cm) -- cycle}

\draw [fill=black,color=black] (-1.1,.68) \rectanglepathtwo;
\draw [fill=blue!80!white,color=blue!80!white,pattern color=blue,pattern=grid] (-.1,.68) \rectanglepathone;

\draw [fill=red,color=red,pattern color=red,pattern=crosshatch] (-1.1,-3.21) \rectanglepaththree;

\draw [fill=blue!80!white,color=blue!80!white,pattern color=blue,pattern=grid] (3.3,0.1) \rectanglepaththree;

\draw [fill=blue!80!white,color=blue!80!white,pattern color=blue,pattern=grid] (6.8,0.1) \rectanglepathone;
\draw [fill=red,color=red,pattern color=red,pattern=crosshatch] (7.3,0.1) \rectanglepathtwo;

\draw [fill=green,color=gray,pattern=crosshatch dots,pattern color=green!75!black] (10.3,2.05) \rectanglepaththree;

\draw [fill=red,color=red,pattern color=red,pattern=crosshatch] (3.3,-2.04) \rectanglepaththree;

\draw [fill=red,color=red,pattern color=red,pattern=crosshatch] (6.8,-2.04) \rectanglepaththree;

\draw [fill=green,color=gray,pattern=crosshatch dots,pattern color=green!75!black] (10.3,-2.25) \rectanglepaththree;

\draw [fill=blue!80!white,color=blue!80!white,pattern color=blue,pattern=grid] (1.9,4.25) \rectanglepaththree;
\draw [fill=blue!80!white,color=blue!80!white,pattern color=blue,pattern=grid] (8.1,4.0) \rectanglepaththree;

\def\rectanglepath2{-- ++(.06cm,0cm) -- ++(0cm,.5cm) -- ++(-.06cm,0cm) -- cycle}

\draw [fill=black,color=black] (3.3,0.1) \rectanglepath2;
\draw [fill=black,color=black] (6.8,0.1) \rectanglepath2;

\draw [fill=black,color=black] (3.3,-2.04) \rectanglepath2;
\draw [fill=black,color=black] (6.8,-2.04) \rectanglepath2;

\end{tikzpicture}}
\end{center}
\caption{Test case $\TCa$ with verdicts}
\label{fig:tca}
\end{figure}

\ignore{
\begin{figure}[htbp]
\begin{center}
\scalebox{0.65}{
\begin{tikzpicture}[->,>=stealth',shorten >=1pt,auto,node distance=2cm,
                    semithick]

  \tikzstyle{every state}=[text=black]

  \node[state, fill=white] (A) {$\ell"_0$};
  \node[state, fill=white] (B) [right of=A, node distance=4cm,
yshift=.9cm] {$\ell"_1$};
  \node[state, fill=white] (C) [right of=B, node distance=3.5cm,
yshift=0cm] {$\ell"_2$};
  \node[state, fill=white] (D) [right of=C, node distance=3.5cm,
yshift=0cm] {$\accept_1$};
  \node[state, fill=white] (E) [below of=A, node distance=2cm, xshift=0cm]
{$\lfail$};
  \node[state, fill=white] (F) [right of=A, node distance=4cm,
yshift=-.9cm] {$\ell"_3$};
  \node[state, fill=white] (G) [right of=F, node distance=3.5cm,
yshift=0cm] {$\ell"_4$};
  \node[state, fill=white] (H) [right of=G, node distance=3.5cm,
yshift=0cm] {$\accept_2$};
  \node[state, fill=white, color=white] (A') [left of=A, node
distance=1.8cm, yshift=0cm] {};
  \node[fill=white] (Ge) [below of=A, node distance=3.2cm, xshift=6cm]
{\large{$\Fail = \{\lfail\}\times\setRnn \sqcup \{\ell"_1,\ell"_2\}\times ]0,\infty[ 
\,\sqcup\, \{\ell"_4\}\times ]1,\infty[$}};
  \node[fill=white] (He) [below of=A, node distance=2.5cm, xshift=6cm]
{\large{$\Inconc = \{\ell"_0\}\times [2,\infty[ \,\cup\,
\{\ell"_3\}\times]0,\infty[ \,\cup\, \{\ell"_4\}\times]0,1]$}};
  \node[fill=white] (Ie) [below of=A, node distance=1.8cm, xshift=6cm]
{\large{$\Pass = \{\accept_1,\accept_2\}\times \setRnn$}};

  \path (A') edge node [above] {} (A)
        (A) edge node [above,sloped] {$y=1, a!,\{y\}$} (B)
        (B) edge node [above] {$y=0, b?, \{y\}$} (C)
        (C) edge node [above] {$y=0, b?, \{y\}$} (D)
        (A) edge node [left] {$y\ge0, b?$} (E)
        (A) edge node [below,sloped] {$1<y<2, a!,\{y\}$} (F)
        (F) edge node [above] {$y=0, b?, \{y\}$} (G)
        (G) edge node [above] {$y=0, b?, \{y\}$} (H)
;

\end{tikzpicture}}
\caption{Test case $\TC$}
\label{fig:tc}
\end{center}
\end{figure}
}

\subsubsection*{Selection of $\TC$:}
So far, the construction of $\TC'$ determines $\Verdicts$, but does
not perform any selection of behaviors.  A last step consists in
trying to control the behavior of $\TC'$ in order to avoid $\Inconc$
states (thus stay in $\pref(\traces_{\accept^\axtp}(\AxTP))$),
because reaching $\Inconc$ means that $\Pass$ is unreachable, thus $\TP$ 
cannot be satisfied anymore.  
To this aim, guards of transitions of $\TC'$ are refined in the final
test case $\TC$ in two complementary ways.  First, transitions leaving
a verdict state ($\Fail$, $\Inconc$ or $\Pass$) are useless, because
the test case execution stops when a verdict is issued.  Thus for each
transition, the guard is intersected with the predicate characterizing
the set of valuations associated with $\None$ in the source location.
This does not change the verdict of traces.  
Second, transitions
arriving in $\Inconc$ states and carrying outputs can be avoided
(outputs are controlled by the test case), thus for any transition
labeled by an output, the guard is intersected with the predicate
characterizing $\None$ and $\Pass$ states in the target location (\ie~
states that are not in $\Inconc$, as $\Fail$ cannot be reached by an
output).  The effect is to suppress some traces leading to $\Inconc$
states.  
All in all, traces in $\TC$ are exactly those of $\TC'$ that
traverse only $\None$ states (except for the last state), 
and do not end in $\Inconc$ with an output.
This selection does not impact on the properties of test suites
(soundness, strictness, precision and exhaustiveness) as will be seen later.

\begin{exa}
Figure~\ref{fig:tc} represents the test case obtained after this selection phase.
One can notice that locations $\ell"_{11},\ell"_{12},\ell"_{13}$ and $\ell"_{21},\ell"_{22}$
have been removed since they can only be reached from $\Inconc$ states, 
thus a verdict will have been emitted before reaching those locations.
The avoidance of $\Inconc$ verdicts by outputs cannot be observed on this example.
However,  with a small modification of $\A$ 
consisting in adding initially the reception of an $a$  before one time unit, 
and not followed by two $b$'s but \eg~one $c$, the resulting transition 
labeled with $(0\leq y <1, a!)$ in $\TC'$ could be cut, producing the same $\TC$. 
\end{exa}




\begin{figure}[htbp]
\begin{center}
\scalebox{0.65}{
\begin{tikzpicture}[->,>=stealth',shorten >=1pt,auto,node distance=2cm,
                    semithick]

  \tikzstyle{every state}=[text=black]

  \node[state, fill=white] (A) {$\ell"_0$};
  \node[state, fill=white] (B) [right of=A, node distance=4cm,
yshift=1.2cm] {$\ell"_1$};

  \node[state, fill=white] (C) [right of=B, node distance=3.5cm,
yshift=0cm] {$\ell"_2$};

  \node[state, fill=white] (D) [right of=C, node distance=3.5cm,
yshift=0cm] {\scriptsize{$\accept_1$}};
  \node[state, fill=white] (E) [below of=A, node distance=2cm, xshift=0cm]
{$\lfail$};
  \node[state, fill=white] (F) [right of=A, node distance=4cm,
yshift=-.9cm] {$\ell"_3$};

  \node[state, fill=white] (G) [right of=F, node distance=3.5cm,
yshift=0cm] {$\ell"_4$};

  \node[state, fill=white] (H) [right of=G, node distance=3.5cm,
yshift=0cm] {\scriptsize{\bf $\accept_2$}};
  \node[state, fill=white, color=white] (A') [left of=A, node
distance=1.8cm, yshift=0cm] {};

  \node[fill=white] (Ge) [below of=A, node distance=4.45cm, xshift=7cm]
{\large{$\rouge{\Fail} = \{\lfail\}\times\setRnn \cup \{\ell"_3,\ell"_4\}\times (0,\infty) 
\,\cup\, \{\ell"_2\}\times (1,\infty)$}};
  \node[fill=white] (He) [below of=A, node distance=3.75cm, xshift=7cm]
{\large{$\bleu{\Inconc} = \{\ell"_0\}\times [2,\infty) \,\cup\, \{\ell"_2\}\times(0,1] \,\cup\,
\{\ell"_1\}\times(0,\infty) $}};
  \node[fill=white] (Ie) [below of=A, node distance=3.05cm, xshift=6cm]
{\large{$\verte{\Pass} = \{\accept_1,\accept_2\}\times \setRnn$}};

  \path (A') edge node [above] {} (A)
        (A) edge node [below,sloped] {$y=1, a!,\{y\}$} (F)
        (F) edge node [above] {$y=0, b?, \{y\}$} (G)
        (G) edge node [above] {$y=0, b?, \{y\}$} (H)
        (A) edge node [left] {$y\ge0, b?$} (E)
        (A) edge node [above,sloped] {$1<y<2, a!,\{y\}$} (B)
        (B) edge node [above] {$y=0, b?, \{y\}$} (C)
        (C) edge node [above] {$y=0, b?, \{y\}$} (D);

\def\rectanglepath{-- ++(.5cm,0cm) -- ++(0cm,.5cm) -- ++(-.5cm,0cm) -- cycle}
\def\rectanglepathone{-- ++(.5cm,0cm) -- ++(0cm,.5cm) -- ++(-.5cm,0cm) -- cycle}
\def\rectanglepathtwo{-- ++(1cm,0cm) -- ++(0cm,.5cm) -- ++(-1cm,0cm) -- cycle}
\def\rectanglepaththree{-- ++(1.5cm,0cm) -- ++(0cm,.5cm) -- ++(-1.5cm,0cm) -- cycle}

\draw [fill=black,color=black] (-1.1,.68) \rectanglepath;
\draw [fill=black,color=black] (-.6,.68) \rectanglepath;
\draw [fill=blue!80!white,color=blue!80!white,pattern color=blue,pattern=grid] (-.1,.68) \rectanglepathone;

\draw [fill=red,color=red,pattern color=red,pattern=crosshatch] (-1.1,-3.21) \rectanglepaththree;

\draw [fill=blue!80!white,color=blue!80!white,pattern color=blue,pattern=grid] (3.3,0.1) \rectanglepaththree;

\draw [fill=blue!80!white,color=blue!80!white,pattern color=blue,pattern=grid] (6.8,0.1) \rectanglepathone;
\draw [fill=red,color=red,pattern color=red,pattern=crosshatch] (7.3,0.1) \rectanglepathtwo;

\draw [fill=green!200,color=gray,pattern=crosshatch dots,pattern color=green!75!black] (10.3,2.05) \rectanglepaththree;

\draw [fill=red,color=red,pattern color=red,pattern=crosshatch] (3.3,-2.04) \rectanglepaththree;

\draw [fill=red,color=red,pattern color=red,pattern=crosshatch] (6.8,-2.04) \rectanglepaththree;

\draw [fill=green!200,color=gray,pattern=crosshatch dots,pattern color=green!75!black] (10.3,-2.25) \rectanglepaththree;

%

\def\rectanglepath2{-- ++(.06cm,0cm) -- ++(0cm,.5cm) -- ++(-.06cm,0cm) -- cycle}

\draw [fill=black,color=black] (3.3,0.1) \rectanglepath2;
\draw [fill=black,color=black] (6.8,0.1) \rectanglepath2;

\draw [fill=black,color=black] (3.3,-2.04) \rectanglepath2;
\draw [fill=black,color=black] (6.8,-2.04) \rectanglepath2;

\end{tikzpicture}}
\end{center}
\caption{Final test case $\TC$ after selection}
\label{fig:tc}
\end{figure}
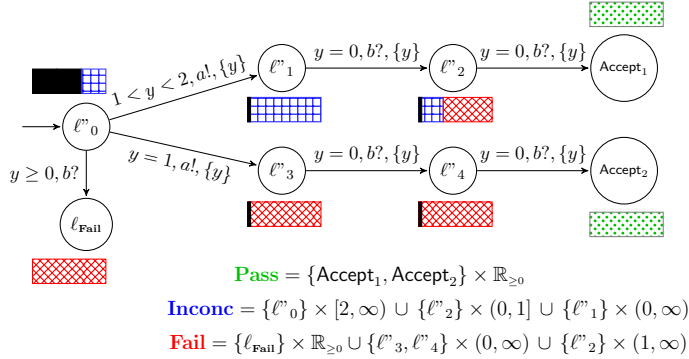

\begin{rem}
Notice that in the example, falling into $\Inconc$ in  
$\ell"_0$ could be avoided by adding the invariant $y<2$,
with the effect of forcing to output $a$.
More generally, invariants can be added to locations by rendering outputs 
urgent 
in order to avoid $\Inconc$, while taking care of keeping test cases non-blocking, \ie\ by ensuring that an output can be done just before the invariant becomes false.
More precisely, $I(\ell)$ is the projection of $\None$ on $\ell$ if 
$Inconc$ is reachable by letting time elapse and it preserves the non-blocking
property, $\true$ otherwise.
\end{rem}

\subsection*{Complexity}
 Let us discuss the complexity of the construction of $\TC$ from
  $\DP$. Note that the size of TAIO $\TC$ is linear in the size of
  $\DP$ but the difficulty lies in the computation of
  $\Verdicts$. Computing $\Pass$ is immediate. 
  The set $\coreach(\Pass)$ can be computed in polynomial time (more
  precisely in $\mathcal{O}(|L^\pdp| . |X^\pdp| . |M^\pdp|) $). To
  explain this, observe that guards in the TAIO $\DP$ are regions and
  with each location $\ell$ is associated an initial region $r_\ell$
  such that guards of transitions leaving $\ell$ are time successors of
  $r_\ell$. Thus during the computation of $\coreach(\Pass)$, for each
  location $\ell$, one only needs to consider these $ \mathcal{O}(
  |X^\pdp| . |M^\pdp|)$ different regions in order to determine the
  latest time-successor $r_\ell^{\max}$ of $r_\ell$ which is
  co-reachable from $\Pass$. Then $\None$ states with location $\ell$
  are exactly those within regions that are time-predecessors of
  $r_\ell^{\max}$. For the same reason (number of possible guards
  outgoing a given location) $E_{\lfail}$ can be computed in
  polynomial time. Last the $\Fail$ verdicts in locations (except for
  $\lfail$) are computed in linear time by complementing the
  invariants in $\DP$.
  The test selection can be done by inspecting all transitions: a
  transition is removed if either the source state is a verdict state,
  or it corresponds to an output action and the successor are
  $\Inconc$ states. This last step thus only requires linear time.
  To conclude, the overall complexity of construction of $\TC$ from
  $\DP$ is polynomial.  

\subsection{Test suite properties}
We have presented the different steps for the generation of 
a TAIO test case from a TAIO specification and an OTAIO test purpose.
The following results express their properties.

\begin{thm}\label{theorem-TC}
  Any test case $\TC$ built by the procedure is sound for $\A$.
  Moreover, if $\DP$ is an exact approximation of $\AxTP$ (\ie
  $\traces(\DP) = \traces(\AxTP)$), the test case $\TC$ is also strict
  and precise for $\A$ and $\TP$.
\end{thm}

The proof is detailed below, but we first give some intuition.
As a preamble, notice that, as explained in the paragraph on test selection, 
traces of $\TC'$ are not affected by the construction of $\TC$.
In particular, the transitions considered in the proof are identical in $\TC$ and $\TC'$.
Soundness comes from 
the construction of $E_{\lfail}$ in $\TC$ and preservation of soundness by the approximate determinization $\DP$ of $\AxTP$
given by Corollary~\ref{cor}.
When $\DP$ is an exact determinization of $\AxTP$,
$\DP$ and $\AxTP$ have same traces, which also equal traces of $\A$ 
since $\TP$ is complete. 
Strictness then comes from the fact that $\DP$ and $\A$ 
 have the same non-conformant traces,
which are captured by the definition of $E_{\lfail}$ in $\TC$.
Precision comes from $\traces_{\accept^\pdp}(\DP)=\traces_{\accept^\axtp}(\AxTP)$ and from the definition of $\Pass$.

When $\DP$ is not exact however, there is a risk that some
behaviors allowed in $\DP$ are not in $\AxTP$, thus some
non-conformant behaviors are not detected, even if they are executed
by $\TC$.
Similarly, some $\Pass$ verdicts may be produced for non-accepted 
or even non-conformant behaviors.
However, if a trace in $\traces_{\accept^\axtp}(\AxTP)$ is present in $\TC$ and observed during testing, 
a $\Pass$ verdict will be delivered. In other words, precision 
is not  always satisfied, but the ``only if'' direction of 
precision (Definition~\ref{def:precise}) is satisfied.

%

\begin{proof}

{\bf Soundness:}
To prove soundness, we need to show that for any $\Imp \in \Imp(\A)$,
$\Imp \, {\tt fails} \, \TC$ implies $\neg(\Imp \, \tioco \, \A)$.

Assuming that $\Imp \, {\tt fails} \, \TC$, there exists a trace
$\sigma \in \traces(\Imp) \cap \traces_{\pFail}(\TC)$.  By the
construction of the set $\Fail$ in $\TC$, there are two cases: either
$\sigma$ leads to a location $(\ell,\neg(\Inv(\ell))$ in $\DP$, or
$\sigma$ leads to a state with location $\lfail$.  In the first case,
$\sigma=\sigma'.\delta$ where $\sigma' \in \traces(\DP)$ and $\delta >
0$ violates the invariant in the location of $\DP \after \sigma'$, and
in the second case, by the construction of $E_{\lfail}$,
$\sigma=\sigma'.a$ where $\sigma' \in \traces(\DP)$ and $a \in
\Sigma_!^\pdp$ is unspecified in $\DP \after \sigma'$.  In both cases,
by definition, this means that $\neg (\Imp \, \tioco\,\DP)$, which
proves that $\TC$ is sound for $\DP$.  Now, as $\DP$ is an
io-abstraction of $\AxTP$ (\ie~ $\AxTP \preceq \DP$), by
Corollary~\ref{cor} this entails that $\TC$ is sound for $\AxTP$.
Finally, we have $\traces(\AxTP)=\traces(\A)$, which trivially implies
that $\A \preceq \AxTP$, and thus that $\TC$ is also sound for $\A$.

{\bf Strictness:}
For strictness, in the case where $\DP$ is an exact approximation of $\AxTP$, 
we need to prove that for any $\Imp \in \Imp(\A)$,
$\neg (\Imp \| \TC \, \tioco \, \A)$ implies that $\Imp \, {\tt fails} \,\TC$.
Suppose that  $\neg (\Imp \| \TC \, \tioco \, \A)$. 
By definition, there exists
$\sigma \in \traces(\A)$ and $a \in out(\Imp \| \TC  \after 
\sigma)$ such that $a \notin out(\A \after  \sigma)$.  
Since $\DP$ is an exact approximation of $\AxTP$, we have 
the equalities $\traces(\DP)
=\traces(\AxTP)=\traces(\A)$, thus $\sigma \in \traces(\DP)$ and $a
\notin out(\DP \after  \sigma)$.  
By construction of $\Fail$ in $\TC$, it
follows that $\sigma.a \in \traces_{\pFail}(\TC)$ which, together with
$\sigma.a \in \traces(\Imp)$, implies that $\Imp \, {\tt fails} \,\TC$.  Thus
$\TC$ is strict.  

{\bf Precision:}
To prove precision, in the case of exact determinization, 
we have to show that 
for any trace $\sigma$,
${\tt Verdict}(\sigma, \TC) = \Pass
\iff \sigma \in \traces(\seq_{\accept^\tp}(\TP) \cap \seq(\A))$.
The definition of 
$\Pass = \bigcup_{\ell\in\accept^\pdp} (\{\ell\} \times\Inv^\pdp(\ell))$  in $\TC$ 
implies that a $\Pass$ verdict is produced for $\sigma$ 
exactly when $\sigma \in \traces_{\accept^\pdp}(\DP)$ which equals
$\traces_{\accept^\axtp}(\AxTP)=
\traces(\seq_{\accept^\tp}(\TP) \cap \seq(\A))$   when $\DP$ is exact.
\end{proof}

\begin{exa}
The test case $\TC$ of Figure~\ref{fig:tc} comes from an approximate determinization. 
However, the approximation comes after reaching $\Inconc$ states.
More precisely, in the gray state of the game in Figure~\ref{ExJeu},
the approximation starts in the time interval $(2,\infty)$.
This state corresponds to location $\ell"_1$ in $\TC$
where the verdict is $\Inconc$ as soon as a non null delay is observed.
The test case is thus strict and precise, despite the over-approximation
in the determinization phase.
\end{exa}

In the following, we prove an exhaustiveness property of our test generation method when determinization is exact.
For technical reasons, we need to restrict to a sub-class of TAIOs defined below.
We discuss this restriction later.

\begin{defi}
We say that an OTAIO $\A$ is {\em repeatedly observable}
if from any state of $\A$, there is a future observable transition,
\ie~ 
$\forall s \in S^\A$, there exists $\mu$ such that 
$s \xrightarrow{\mu}$ and  $Trace(\mu)\notin \setRnn$.  
\end{defi}

\begin{thm}[Exhaustiveness]
\label{th:exhaustiveness}
Let $\A$ be a repeatedly observable TAIO which can be exactly
determinized by our approach.  Then the set of test cases that can be
generated from $\A$ by our method is exhaustive.
\end{thm}

\begin{proof}
  Let
  $\A=(L^\pA,\ell_0^\pA,\Sigma_?^\pA,\Sigma_!^\pA,\Sigma_\tau^\pA,X_p^\pA,\emptyset,M^\pA,\Inv^\pA,E^\pA)$
  be the TAIO specification, and
  $\Imp=(L^\imp,\ell_0^\imp,\Sigma_?^{\pA},\Sigma_!^{\pA},\Sigma_\tau^\imp,X_p^\imp,\emptyset,M^\imp,\Inv^\imp,E^\imp)$
  any non-conformant implementation in $\Imp(\A)$.  The idea is
  now to prove that from $\A$ and $\Imp$, one can build a test purpose
  $\TP$ such that the test case $\TC$ built from $\A$ and $\TP$ may
  detect this non-conformance, \ie~$\Imp\, {\tt fails} \, \TC $.

  By definition of $\neg(\Imp \;\tioco\; \A)$, there exists $\sigma
  \in \traces(\A)$ and $a \in \Sigma_!^\pA \sqcup \setRnn$ such that
  $a \in out(\Imp \after \sigma)$ but $a \notin out(\A \after
  \sigma)$.  Since $\A$ is repeatedly observable, there also exists
  $\delta \in \setRnn$ and $b \in \Sigma_{obs}^\pA$ such that
  $\sigma.\delta.b \in \traces(\A)$.

  As $\A$ can be determinized exactly by our approach, there must
  exist some resources $(k,M)$ and a strategy $\Pi$ for
  Determinizator in the game $\mathcal{G}_{\pA,(k,M)}$ such that
  $\traces(\mathsf{Aut}(\Pi)) = \traces(\A)$.

  From the non-conformant implementation $\Imp$, a test purpose
  $(\TP,\accept^\tp)$ can be built, with $\TP=
  (L^\tp,\ell_0^\tp,\Sigma_?^{\pA},\Sigma_!^{\pA},\Sigma_\tau ^{\pA},
  X_p^\tp,X_o^\tp,M^\tp,\Inv^\tp,E^\tp)$, $X_p^\tp= X_p^\imp\sqcup
  X^{\mathsf{Aut}(\Pi)}$ and $X_o^\tp=\emptyset$, and
  $\sigma.\delta.b \in \traces_{\accept^\tp}$ but none of its prefixes
  is in $\traces_{\accept^\tp}$.  The construction of $\TP$ relies on
  the region graph of $\Imp\|\mathsf{Aut}(\Pi)$.  First a TAIO $\TP'$
  is built which recognizes exactly the traces read along the path
  corresponding to $\sigma$ in the region graph of
  $\Imp\|\mathsf{Aut}(\Pi)$, followed by a transition $b$ with the
  guard corresponding to the one in $\mathsf{Aut}(\Pi)$.  In
  particular it recognizes the trace $\sigma.\delta.b$.  The test
  purpose $(\TP,\accept^\tp)$ is then built such that $\TP$ accepts in
  its states $\accept^\tp$ the traces of $\TP'$.  Note that $\TP$
  should be complete for $\Sigma$, thus locations of $\TP'$ should be
  completed by adding loops without resets for all actions in
  $\Sigma_\tau$, and adding, for all observable actions, transitions
  to a trap location guarded with negations of their guards in $\TP'$.

  Now consider our test generation method applied to $\TP$ and $\A$.
  First $\AxTP=\A \times \TP$ is built, and we consider the game
  $\mathcal{G}_{\pA,(k',M')}$ with $k'= k + |X_p^\tp|$ and
  $M'=\max(M,M^\tp)$.  One can then define a strategy $\Pi'$ composed
  of the strategy $\Pi$ for the $k$ first clocks, and following the
  resets of $\TP$ (which is deterministic) for the other clocks
  corresponding to those in $X_p^\tp$.  The construction of
  $(\DP,\accept^\pdp)$ following the strategy $\Pi'$ thus ensures that
  $\traces(\DP) = \traces(\AxTP)$ and $\traces_{\accept^{\pdp}}(\DP) =
  \traces_{\accept^{\axtp}}(\AxTP)$.

  Finally, let $\TC$ be the test case built from $\DP$.  Observe that
  $\TC  \after \sigma.\delta.b \subseteq \Pass$, but $\TC  \after
  \sigma.\delta \not\subseteq \Pass$. As a consequence, $\TC \after
  \sigma \subseteq \None$. 
  Moreover we have $a \notin out(\A \after \sigma)$, hence 
  $\sigma.a \in \traces_{\pFail}(\TC)$ and as
  $\sigma.a \in \traces(\Imp)$, we can conclude that  $\Imp \; {\tt fails} \;
  \TC$.
\end{proof}

\subsection*{Discussion:}
The hypothesis that $\A$ is repeatedly observable is in fact not
restrictive for a TAIO that is determinizable by our approach.
Indeed, such a TAIO can be transformed into a repeatedly observable
one with same conformant implementations, by first determinizing it,
and then completing it as follows.  In all locations, a transition
labeled by an input is added, which goes to a trap state looping for
all outputs, and is guarded by the negation of the union of guards of
transitions for this input in the deterministic automaton.

When $\A$ cannot be determinized exactly, the risk is that 
some non-conformance may be undetectable.
However, 
the theorem can be generalized to non-determinizable automata with no resets on internal action.
Indeed, in this case, in the game with resources $(k,M)$, 
where $k$ is the length of the finite non-conformant trace $\sigma.a$,
the strategy consisting in resetting a new clock at each observable
action allows to remain exact until the observation of non-conformance
(see remark after Theorem~\ref{th:game}).  The proof of
theorem~\ref{th:exhaustiveness} can be adapted using this strategy.


\section{Discussion and related work}
\label{sec-discuss}

\subsection*{Alternative definitions of test purposes}
The definition of test purposes depends on the semantic level at which
behaviors to be tested are described (\eg~sequences, traces).  This
induces a trade-off between the precision of the description of
behaviors, and the cost of producing test suites.  In this work, test
purposes recognize timed sequences of the specification $\A$, by a
synchronization with actions and observed clocks.  They also have
their own proper clocks for additional precision.  The advantage is a
fine tuning of selection.  The price to be paid is that, for each test
purpose, the whole sequence of operations, including determinization
which may be costly, must be done.  An alternative is to define test
purposes recognizing timed traces rather than timed sequences.  In
this case, selection should be performed on a deterministic io-abstraction $\B$ of
$\A$ obtained by an approximate determinization of $\A$.  
Then, test purposes should not refer to $\A$'s clocks as these
are lost by the approximate determinization.  Test purposes should
then either observe $\B$'s clocks, and thus be defined after
determinization, or use only proper clocks in order not to depend on
$\B$, at the price of further restricting the expressive power of test
purposes.  In both cases, test purposes should preferably be
deterministic in order to avoid a supplementary determinization after
the product with $\B$.  The main advantage of these approaches is that
the specification is determinized only once, which reduces the cost of
producing a test suite.  However, the expressive power of test
purposes is reduced.

\subsection*{Test execution}
Once test cases are selected, it remains to execute them on a real
implementation.  As a test case is a TAIO, and not a simple timed
trace, a number of decisions still need to be taken at each state of
the test case: (1) whether to wait for a certain delay, or to receive
an input or to send an output (2) which output to send, in case there
is a choice.  It is clear that different choices may lead to different
behaviors and verdicts.  Some of these choices can be made either
randomly (\eg~choosing a random time delay, choosing between outputs,
etc), or can be pre-established according to user-defined strategies.
One such policy is to apply a technique similar to the control
approach of~\cite{DavidLarsenLiNielsen-ICST09} whose goal is to avoid
$\rtraces(\A,\TP)$.

Moreover, the tester's time observation capabilities are limited in
practice: testers only dispose of a finite-precision digital clock (a
counter) and cannot distinguish among observations which elude their
clock precision.  Our framework may take this limitation into account.
In~\cite{KrichenTripakis09} assumptions on the tester's digital clock
are explicitly modeled as a special TAIO called $Tick$, synchronized
with the specification before test generation, then relying to the
untimed case.  We could imagine to use such a $Tick$ automaton
differently, by synchronizing it with the resulting test case after
generation.

\subsection*{Related work}
As mentioned in the introduction, off-line test selection is in
general restricted to deterministic automata or known classes of
determinizable timed automata.  An exception is the work
of~\cite{KrichenTripakis09} which relies on an over-approximate
determinization.  Compared to this work, our approximate
determinization is more precise (it is exact in more cases), it copes
with outputs and inputs using over- and under-approximations, and
preserves urgency in test cases as much as possible.  Another
exception is the work of~\cite{DavidLarsenLiNielsen-ICST09}, where the
authors propose a game approach whose effect can be understood as a
way to completely avoid $\rtraces(\A,\TP)$, with the possible risk of
missing some or even all traces in
$\pref(\traces_{\accept^\axtp}(\AxTP))$.  Our selection, which allows
to lose this game and produce an $\Inconc$ verdict when this happens,
is both more liberal and closer to usual practice.

In several related
works~\cite{KoneCastanetLaurencot-IC3N98,EnNouaryDssouli-Testcom2001},
test purposes are used for test case selection from TAIOs.  In all
these works, test purposes only have proper clocks, thus cannot
observe clocks of the
specification. 
\ignore{ The advantage of our definition is its generality and a fine
  tuning of selection.  One could argue that the cost of producing a
  test suite can be heavy, as for each test purpose, the whole
  sequence of operations, including determinization, must be done.  In
  order to avoid this, an alternative would be to define test purposes
  recognizing timed traces and perform selection on the approximate
  determinization $\B$ of $\A$.  But then, the test purpose should not
  use $\A$'s clocks as these are lost by determinization.  Then, test
  purposes are either defined after determinization and observe $\B$'s
  clocks, or their expressive power is further restricted by using
  only proper clocks in order not to depend on $\B$.  }

It should be noticed that selection by test purposes can be used for 
test selection with respect to coverage criteria~\cite{Zhu-Hall-May-97}.
Those coverage criteria define a set of  elements (generally syntactic ones) 
to be covered (\eg~locations, transitions, branches, etc). 
Each element can then be translated into a test purpose, 
the produced test suite covering the given criteria.

\section{Conclusion}
\label{sec-conclusion}

In this article, we presented a complete formalization and operations
for the automatic off-line generation of test cases from
non-deterministic timed automata with inputs and outputs (TAIOs).  The
model of TAIOs is general enough to take into account non-determinism,
partial observation and urgency.  One main contribution is the ability
to tackle any TAIO, thanks to an original approximate determinization
algorithm.  Another main contribution is the selection of test cases
with expressive test purposes described as OTAIOs having the ability
to precisely select behaviors to be tested based on clocks and actions
of the specification as well as proper clocks.  Test cases are
generated as TAIOs using a symbolic co-reachability analysis of the
observable behaviors of the specification guided by the test purpose.

A first perspective of this work is to implement the approach in a
test generation tool.  Currently, the approximate determinization has
been prototyped in Python thanks to a binding of the UPPAAL DBM
library~\cite{uppaal-dbm}.  Other perspectives could be to combine
this approach with the one of~\cite{jeannet05a} for models with data,
for the generation of test cases from models with both time and data
in the spirit of~\cite{amost}, but generalized to non-deterministic
models.

\paragraph{Acknowledgements:} we would like to thank the reviewers for their constructive comments 
that allowed us to improve this article.

\bibliographystyle{abbrv}
\bibliography{tacas}

\end{document}